\documentclass[10pt]{article}
\usepackage{amssymb,amsfonts,amsmath, amsthm, amsbsy}
\usepackage{caption,color,graphicx,paralist, subcaption, placeins, array, mathtools, url, mdwlist, color} 
\usepackage[text={6.75in,9.5in},centering,letterpaper]{geometry}
\usepackage[hypertexnames=false,colorlinks=true,urlcolor=blue,linkcolor=blue,citecolor=blue]{hyperref}
\usepackage[numbers, comma, square, sort&compress]{natbib}

\usepackage{todonotes}

\newcommand{\blue}[1]{{#1}}
\newcommand{\red}[1]{{#1}}
\usepackage[normalem]{ulem} 

\everymath={\displaystyle}
 \numberwithin{equation}{section}

\graphicspath{{figs/inkscape/}{figs/}{figs/xfig/}}

\newtheorem{theorem}{Theorem}[section]
\newtheorem{hypothesis}{Hypothesis}

\newtheorem{lemma}[theorem]{Lemma}
\newtheorem{proposition}[theorem]{Proposition}

\newtheorem{definition}[theorem]{Definition}

\newtheorem{remark}[theorem]{Remark}

\newcommand{\R}{\mathbb{R}}
\newcommand{\C}{\mathbb{C}}
\newcommand{\N}{\mathbb{N}}
\newcommand{\Z}{\mathbb{Z}}

\def\Re{\mathop{\mathrm{Re}}}
\def\Im{\mathop{\mathrm{Im}}}
\def\tlambda{\tilde\lambda}
\def\tnu{\tilde\nu}

\def\abs{\mathrm{abs}}
\def\Airy{\mathrm{Airy}}

\def\osc{\mathrm{osc}}
\def\tX{\tilde{X}}
\def\tzeta{{\tilde{\zeta}}}

\newcommand{\rmd}{\mathrm{d}}

\newcommand{\rmi}{\mathrm{i}}

\newcommand{\eps}{\varepsilon}

\def\calO{\mathcal{O}}

\def\b{\underline{b}}
\def\tu{{\widetilde u}}

\def\slow{\mathrm{slow}}
\def\fast{\mathrm{fast}}
\def\ss{\mathrm{ss}}
\def\uu{\mathrm{uu}}
\def\cc{\mathrm{c}}
\def\Re{\mathrm{Re}}
\def\Im{\mathrm{Im}}
\def\tDelta{\tilde\Delta}

\def\M{\mathcal{M}}
\def\Iabs{I_\abs}
\def\slowabs{\Sigma_\abs^{\rm slow}}

\def\specpt{\Sigma_{\rm pt}}

\title{Pulse replication and accumulation of eigenvalues}
\author{Paul Carter\footnote{School of Mathematics, University of Minnesota} \and  Jens D.M.\ Rademacher\footnote{ Fachbereich 3 -- Mathematik, Universit\"at Bremen} \and Bj\"orn Sandstede \footnote{Division of Applied Mathematics, Brown University}}

\begin{document}

\maketitle
\begin{abstract}
Motivated by pulse-replication phenomena observed in the FitzHugh--Nagumo equation, we investigate traveling pulses whose slow-fast profiles exhibit canard-like transitions. We show that the spectra of the PDE linearization about such pulses may contain many point eigenvalues that accumulate onto a union of curves as the slow scale parameter approaches zero. The limit sets are related to the absolute spectrum of the homogeneous rest states involved in the canard-like transitions. Our results are formulated for general systems that admit an appropriate slow-fast structure.
\end{abstract}



%
%

\section{Introduction}

\begin{figure}
\hspace{.025 \textwidth}
\begin{subfigure}{.3 \textwidth}
\centering
\includegraphics[width=1\linewidth]{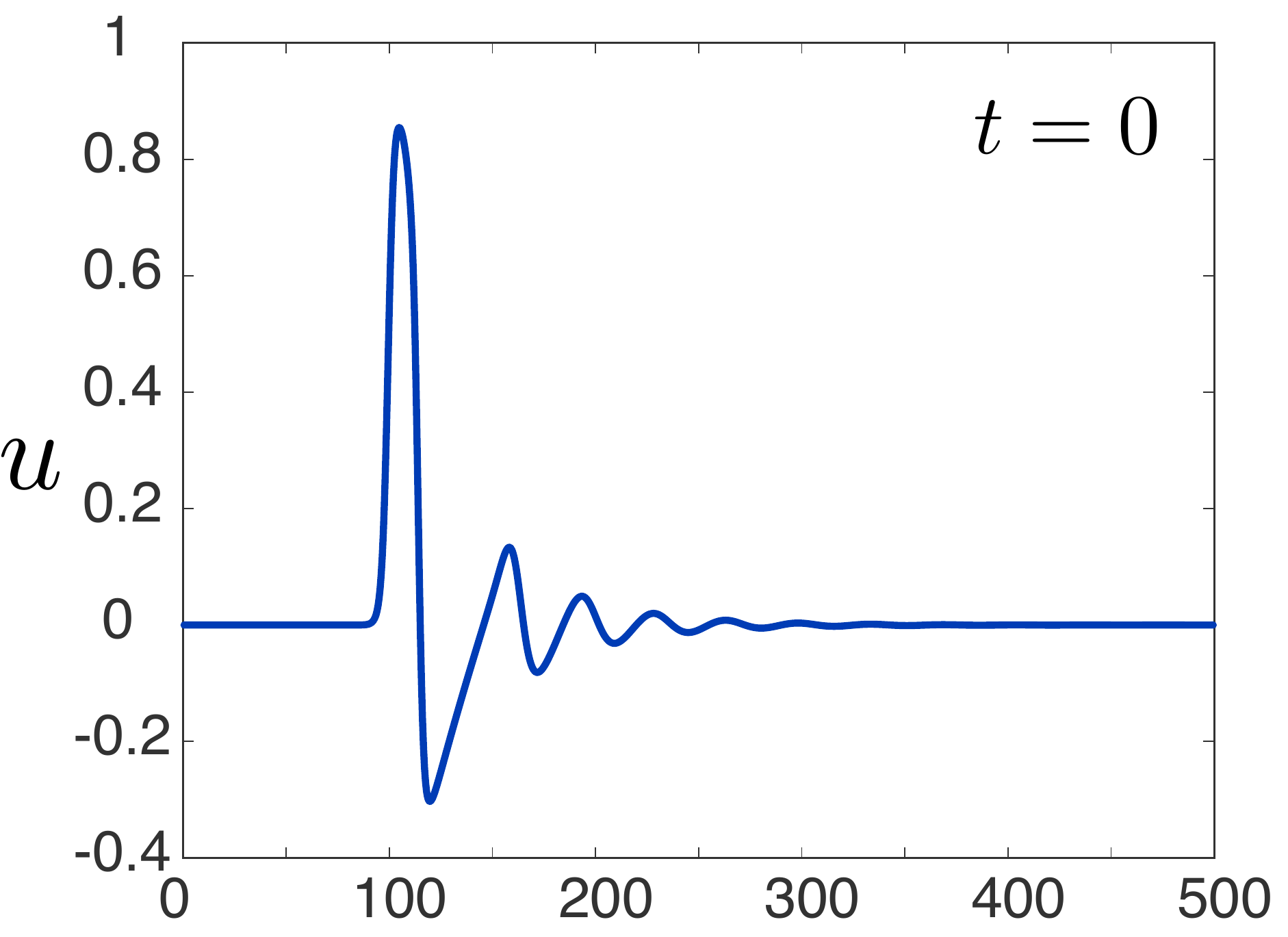}
\end{subfigure}
\hspace{.025 \textwidth}
\begin{subfigure}{.3\textwidth}
\centering
\includegraphics[width=1\linewidth]{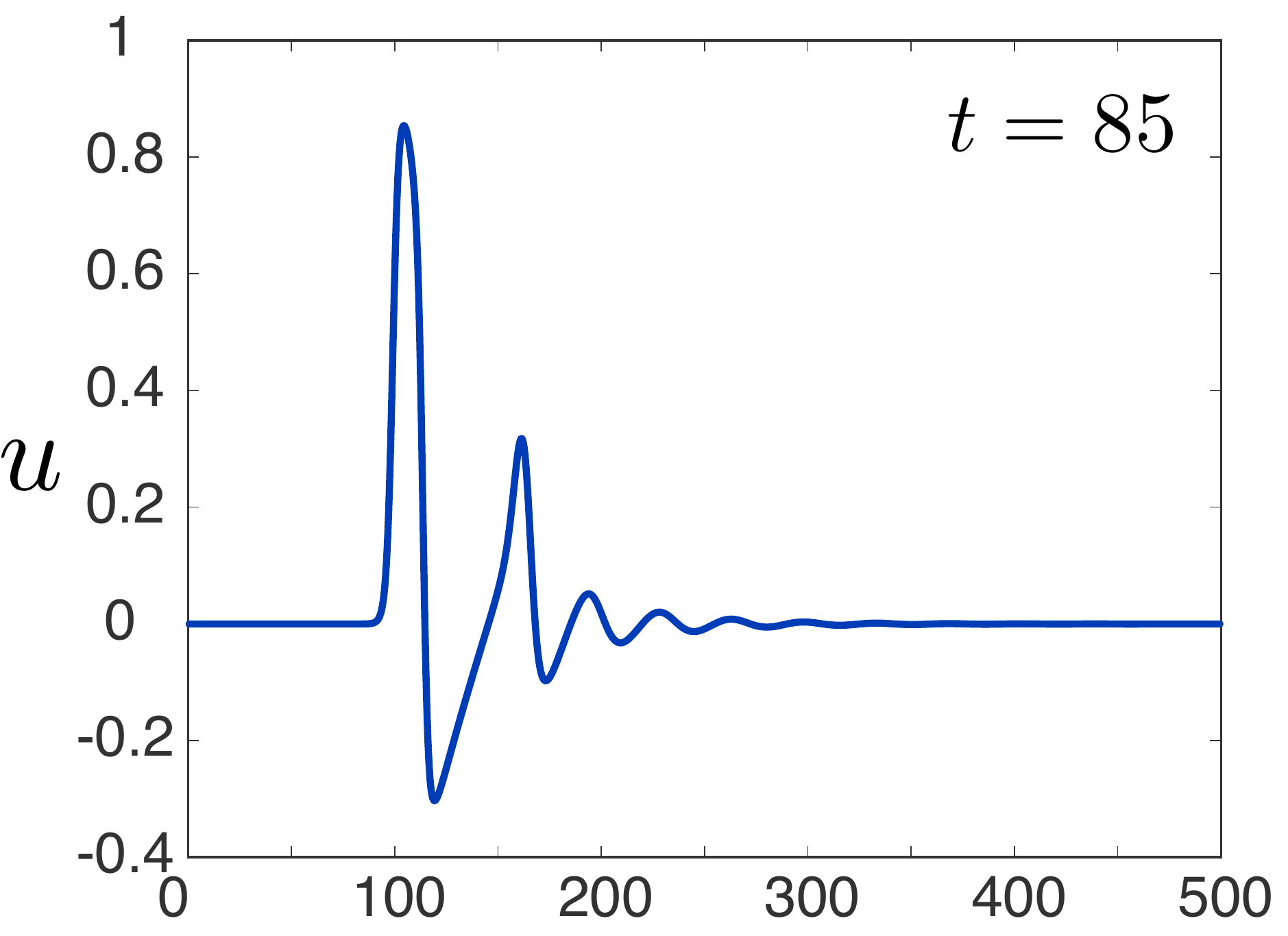}
\end{subfigure}
\hspace{.025 \textwidth}
\begin{subfigure}{.3 \textwidth}
\centering
\includegraphics[width=1\linewidth]{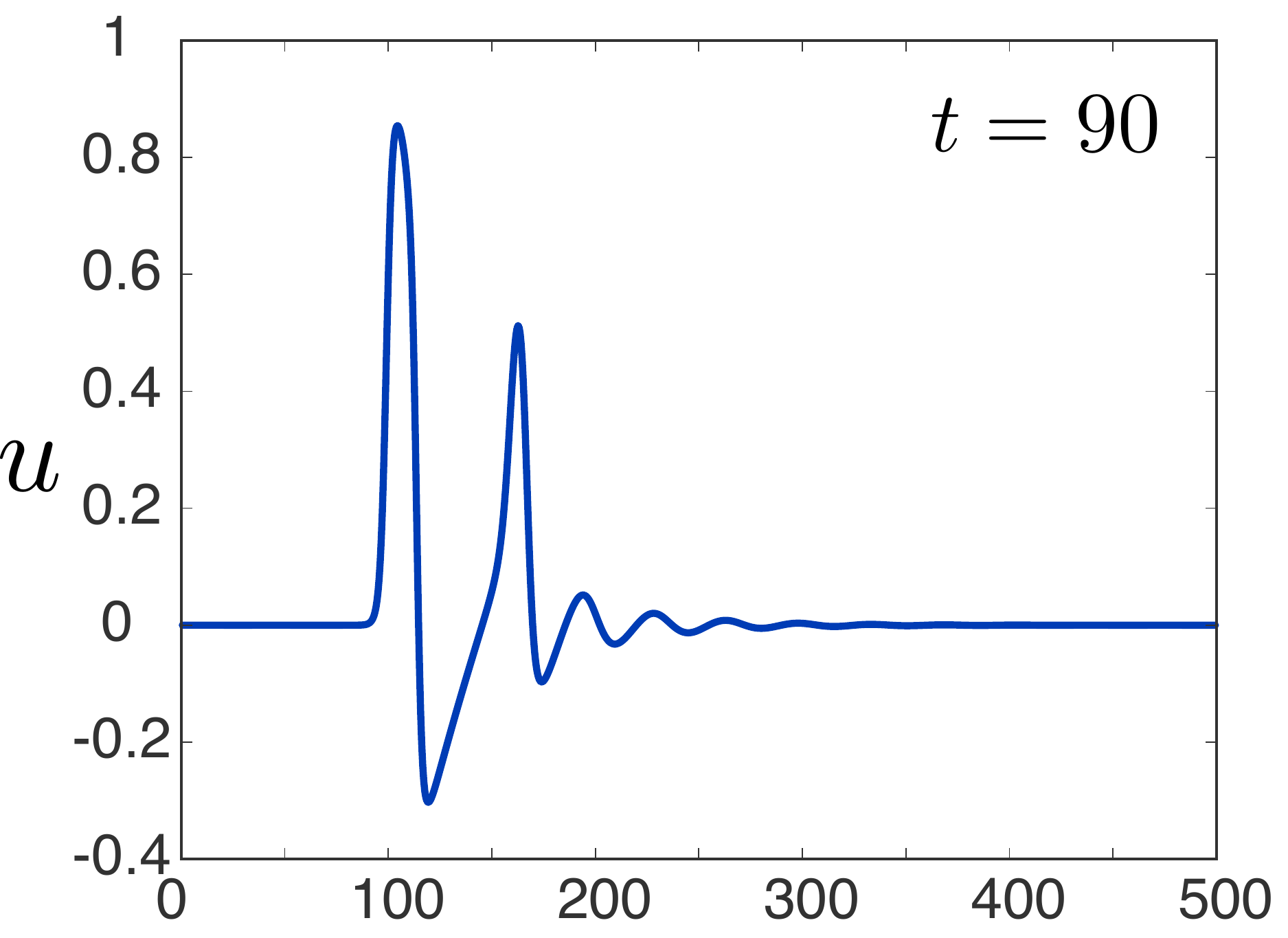}
\end{subfigure}\\ \vspace{10pt}

\hspace{.025 \textwidth}
\begin{subfigure}{.3 \textwidth}
\centering
\includegraphics[width=1\linewidth]{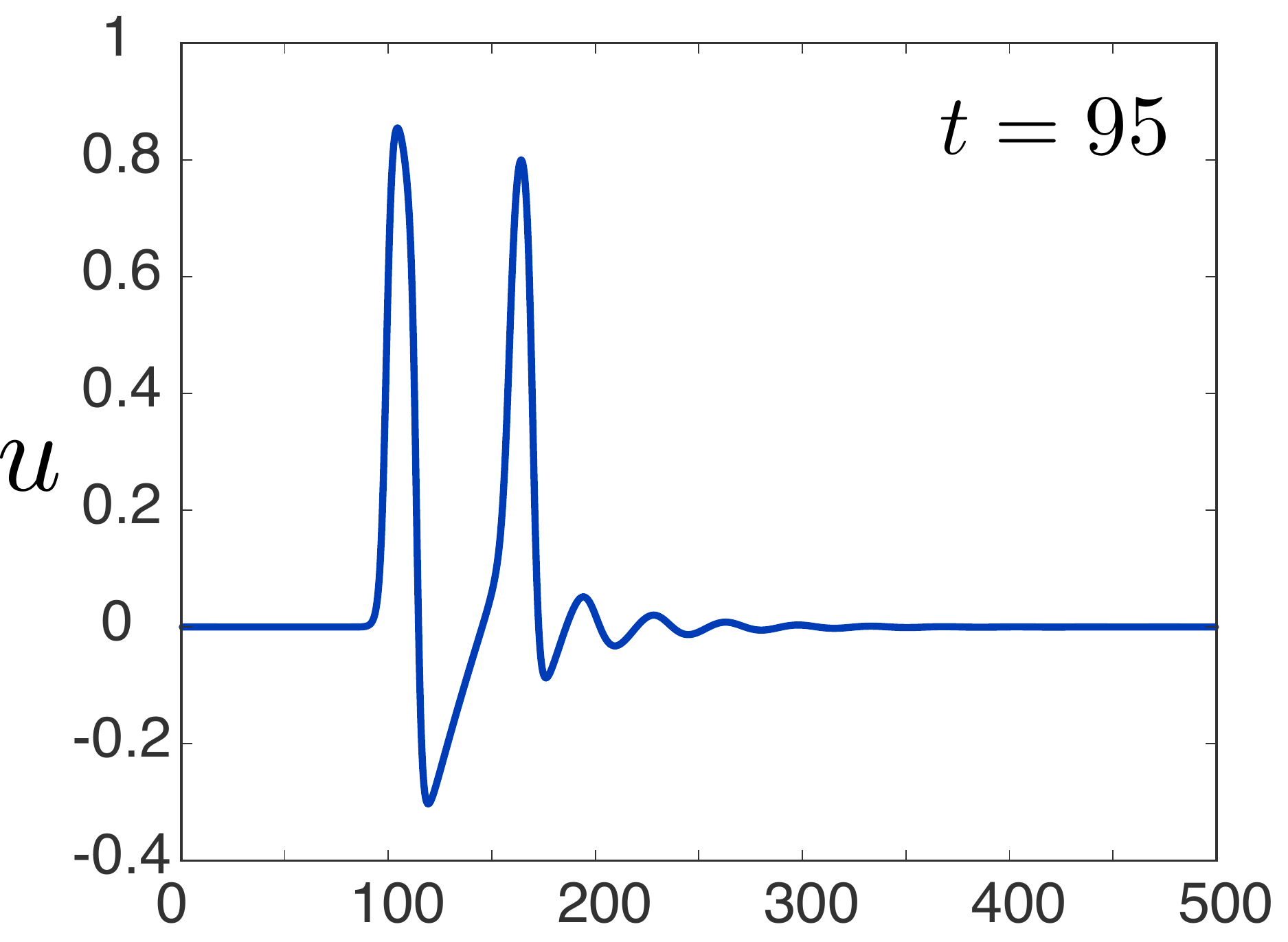}
\end{subfigure}
\hspace{.025 \textwidth}
\begin{subfigure}{.3 \textwidth}
\centering
\includegraphics[width=1\linewidth]{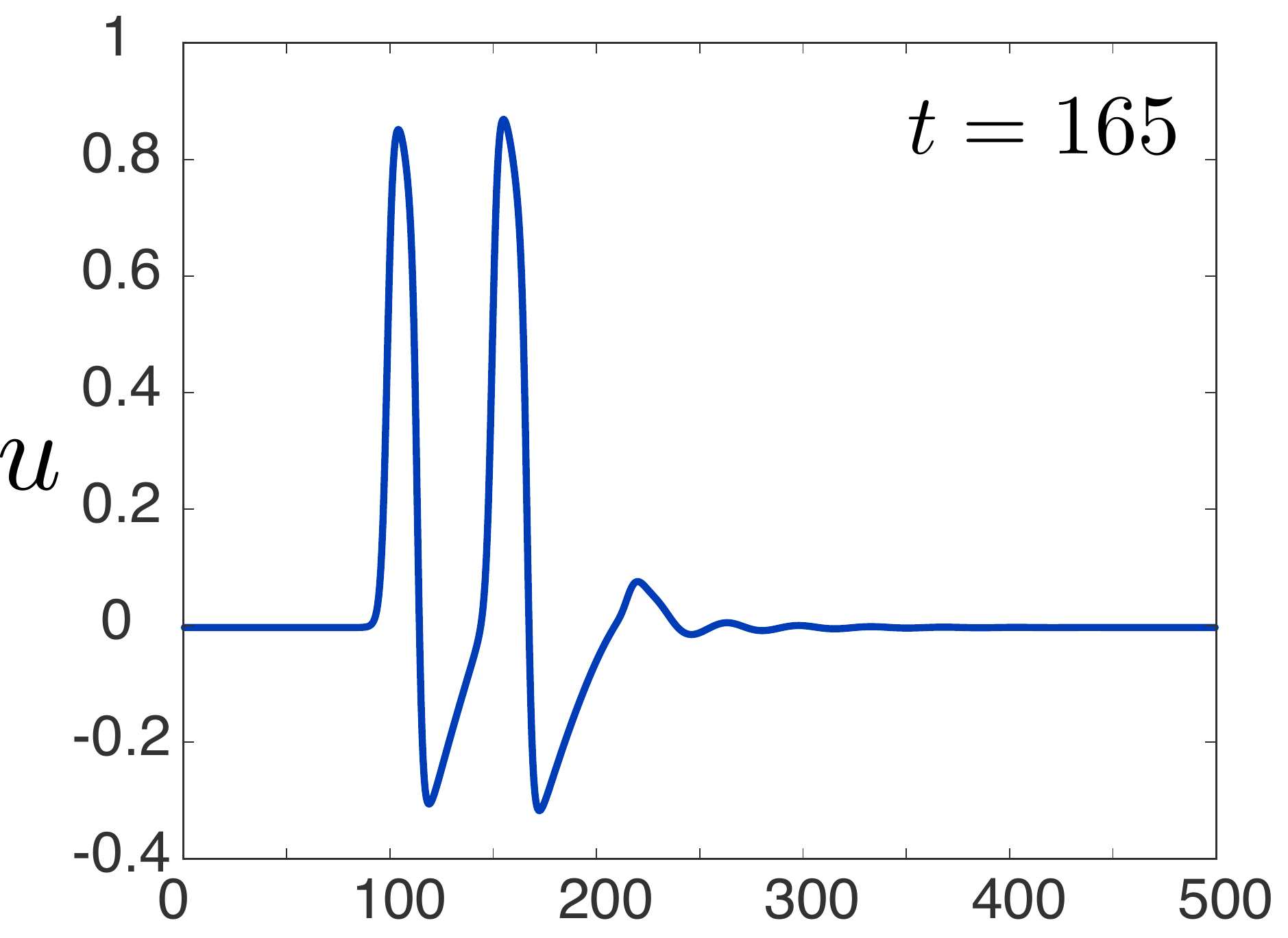}
\end{subfigure}
\hspace{.025 \textwidth}
\begin{subfigure}{.3 \textwidth}
\centering
\includegraphics[width=1\linewidth]{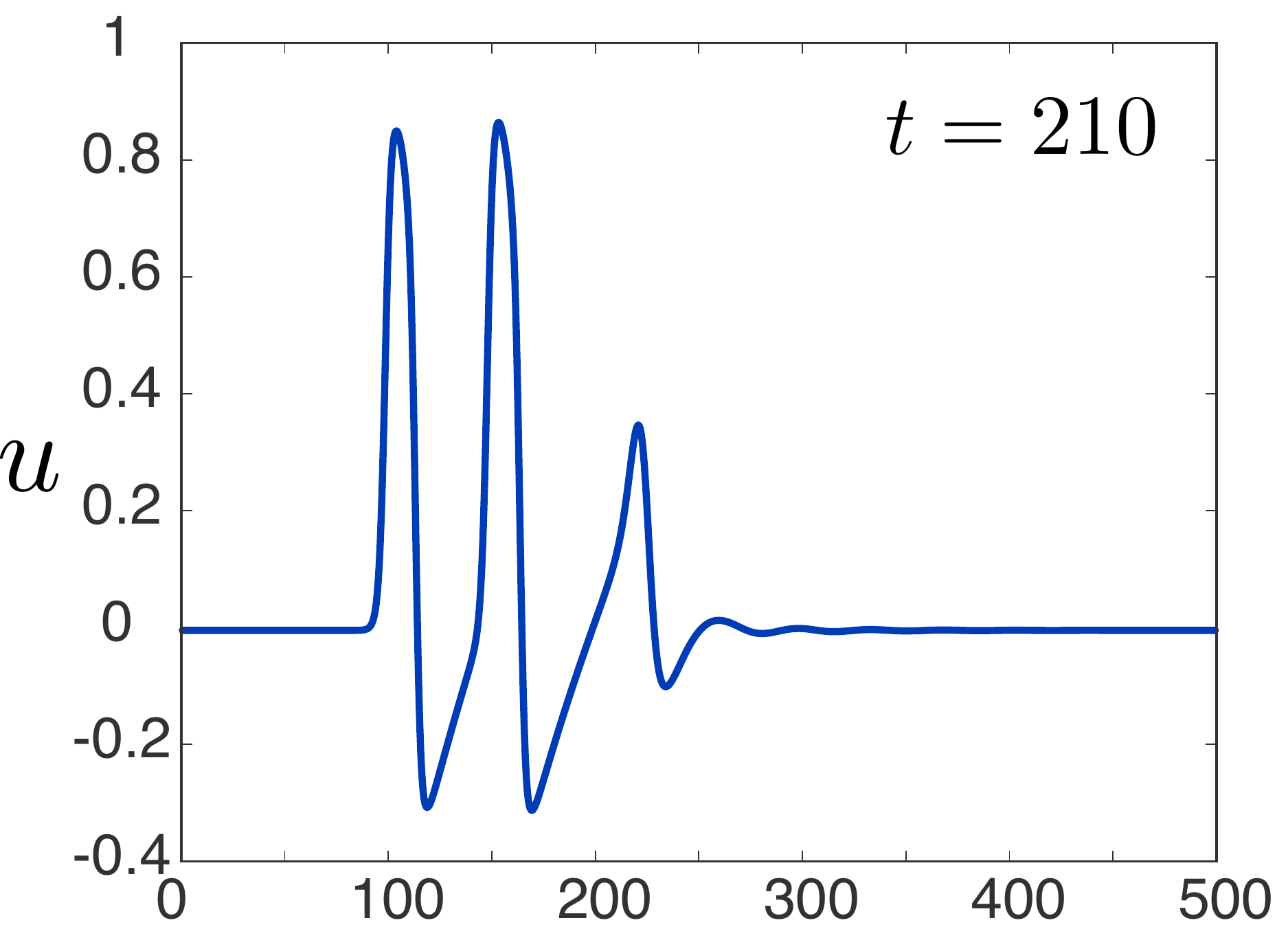}
\end{subfigure}
\\ \vspace{10pt}

\hspace{.2 \textwidth}
\begin{subfigure}{.3 \textwidth}
\centering
\includegraphics[width=1\linewidth]{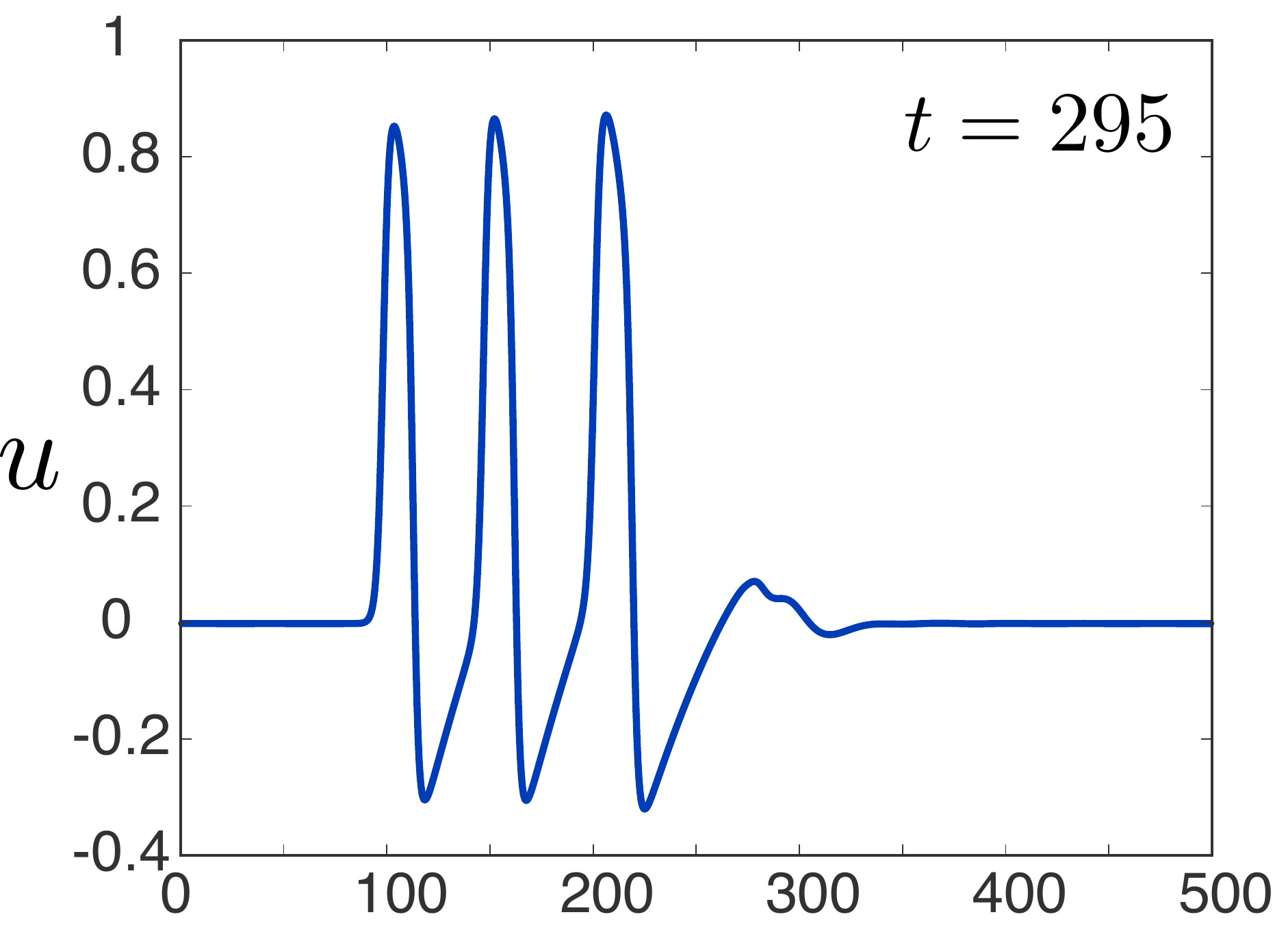}
\end{subfigure}
\hspace{.025 \textwidth}
\begin{subfigure}{.3 \textwidth}
\centering
\includegraphics[width=1\linewidth]{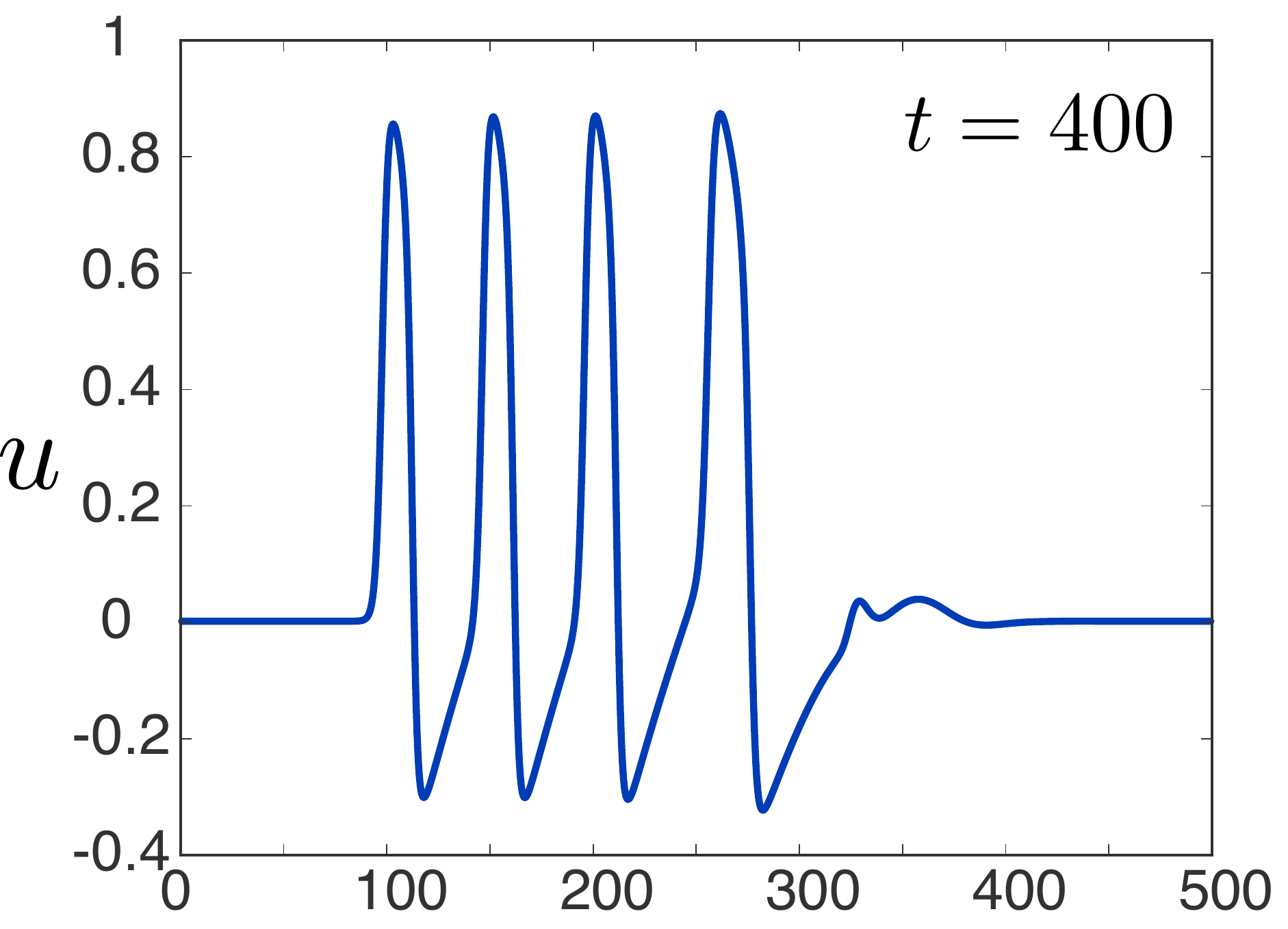}
\end{subfigure}
\caption{Shown is the time evolution of a solution to \eqref{e:1} for $a=-0.0295$ \blue{(adjacent to the single-to-double pulse transition along the homoclinic banana, cf.\ Fig.~\ref{f:ccurvebanana})}, $\gamma=2$, and $\eps=0.015$: starting with a single pulse, additional pulses emerge sequentially in the tail of the pulse.}
\label{f:temporal_replication}
\end{figure}

Spatially localized traveling waves, which we refer to as pulses, arise in many spatially extended systems. Much is known about the theoretical existence and stability properties of pulses, and these results have also been applied successfully to many different models and applications. Our work here is motivated by a phenomenon that is commonly referred to as pulse replication, or self-replication, which has been observed in at least two distinct forms in various models and experiments \cite{SMM1992,pearson1993,LeeEtAl1994,BarEtAl1994,OtterEtAl1996, DGK98,NU1999,KrishnanEtAl1999,SHM2000,SLY2000,hayase2000self,Rademacher_thesis,KWW2005,ManzSteinbock2006,chen2012simple,bauer2015,bordeu2016self}. We refer to \emph{pulse adding} as the process of sequentially adding pulses in the immediate wake of a traveling pulse as time increases: this process is illustrated in Figure~\ref{f:temporal_replication}. Alternatively, \emph{pulse splitting} may occur where traveling pulses splits into two identical but spatially separated traveling pulses as time increases. Both scenarios occur in the singularly perturbed FitzHugh--Nagumo system
\begin{eqnarray}\label{e:1}
u_t & = & u_{xx} + u (u-a)(1-u) - w \\ \nonumber
w_t & = & \varepsilon (u - \gamma w)
\end{eqnarray}
for $a$ close to zero and $0<\varepsilon\ll1$, and we refer to Figure~\ref{f:temporal_replication} for simulations of pulse-adding in this system.

\begin{figure}
\centering
\includegraphics{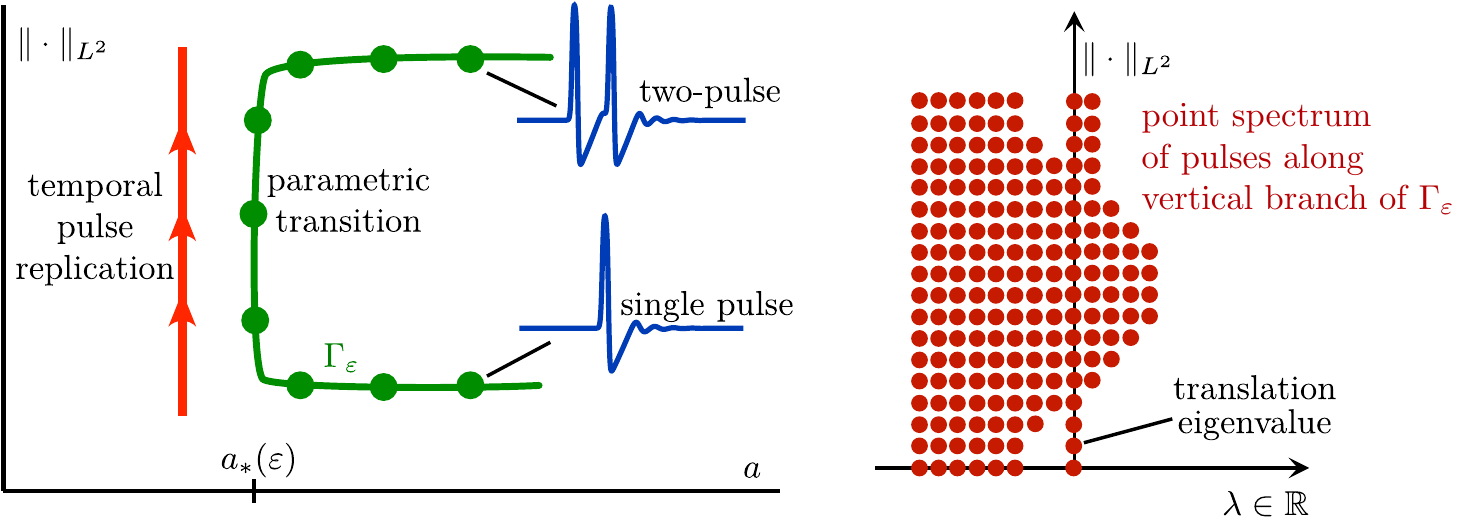}
\caption{Left: Shown is the curve $\Gamma_\varepsilon$ of traveling pulses together with a time-dependent trajectory for $a$ below $a_*(\varepsilon)$. Right: We illustrate how the rightmost real point spectra $\lambda \in \mathbb{R}$ of the traveling pulses vary along the vertical branch of $\Gamma_\varepsilon$ as the single pulse transitions to the double pulse for fixed $0<\varepsilon\ll1$. }
\label{f:2}
\end{figure}

We now outline briefly the mechanism that, we believe, is responsible for the emergence of pulse adding in the FitzHugh--Nagumo model (more details will be given in \S\ref{s:2}). As illustrated in Figure~\ref{f:2} and proved in \cite{CSbanana}, for each fixed $0<\varepsilon\ll1$, the FitzHugh--Nagumo system exhibits a one-parameter family $\Gamma_\varepsilon$ of pulses near $a=0$ that undergo a parametric transition from a single pulse to a two-pulse as the one-parameter family $\Gamma_\varepsilon$ is traversed. We emphasize that each pulse along the curve $\Gamma_\varepsilon$ is an ordinary traveling wave that does not change its profile as it propagates. The transition occurs due to canard dynamics near the asymptotic rest state and occurs in an exponentially thin region in parameter space. Thus, as $\varepsilon\to0$, the leftmost part of the curve $\Gamma_\varepsilon$ becomes more and more vertical in $(a,\|\cdot\|_{L^2})$ space, and the parametric transition therefore occurs essentially at almost a constant value $a_*(\varepsilon)$ of the parameter $a$. Numerically, we observe pulse replication for all parameters $a$ that are close to but slightly smaller than $a_*(\varepsilon)$. Furthermore, our simulations reveal that the profiles of solutions that undergo temporal pulse-adding for $a$ just below $a_*(\varepsilon)$ are similar to the profiles of the traveling pulses in the parametric-transition family $\Gamma_\varepsilon$. These observations suggest that pulse-adding solutions track the family $\Gamma_\varepsilon$ of traveling pulses, which therefore provides the backbone of the \blue{initial} pulse-adding process; see Figure~\ref{f:2}. \blue{Consistent with the ongoing adding process is our numerical finding of analogous branches for $n$- to $(n+1)$-pulse transitions as presented in \S\ref{s:disc}.}

In this paper, as a first step towards understanding pulse-adding in the FitzHugh--Nagumo system, we will determine the stability properties of the pulses that lie on $\Gamma_\varepsilon$. It is known that the single pulse along the lower horizontal branch of $\Gamma_\varepsilon$ is spectrally stable \cite{CdRS}, and theoretical results suggest that the two-pulse along the upper horizontal branch should have at most one unstable eigenvalue \cite{Sstab}. Our analysis will therefore focus on examining the stability properties of the intermediate traveling pulses that facilitate the transition from a single pulse to a two-pulse along $\Gamma_\varepsilon$.

Our main finding is that these intermediate pulses are spectrally unstable due to a specific mechanism for generating eigenvalues in systems with an  appropriate slow-fast structure. In fact, as illustrated in Figure~\ref{f:2}, we will show that the pulse's spectrum reaches a finite distance into the right-half plane regardless of the value of $\varepsilon$. The cause of this severe instability is the so-called absolute spectrum \cite{SSabs, SSgluing} of the spatially homogeneous rest states that are unstable along the middle branch of the cubic nonlinearity $w=u(u-a)(1-u)$ for $0<\varepsilon\ll1$. During the parametric transition from single to double pulses, the profile of each pulse in $\Gamma_\varepsilon$ will, as a function of the spatial variable $x$, slowly traverse the middle branch, and its spectrum will collect eigenvalues from the union of the unstable absolute spectra along the middle branch: in particular, eigenvalues accumulate along the resulting `slow' absolute spectrum. \blue{In later sections}, we will provide a precise definition of the absolute spectrum and also discuss its impact on spectral stability during the parametric transition along $\Gamma_\varepsilon$. The implications of our analysis for pulse replication \blue{in the FitzHugh--Nagumo system} will be discussed in \S\ref{s:disc}. 

\blue{To prove our results, we use spatial dynamics to frame the eigenvalue problems associated with the traveling pulses as linear slow-fast ordinary differential equations. These equations are then analysed using geometric singular perturbation theory, exponential dichotomies, and trichotomies for linear systems with slowly varying coefficients, which will allow us to track and match prospective eigenfunctions as the pulse profile traverses a slow manifold that corresponds to unstable homogeneous rest states.}

\blue{We emphasize that our results show that accumulation of unstable eigenvalues for traveling pulses occurs generically in slow-fast systems when the pulse profile exhibits long plateaus close to \red{absolutely} unstable homogeneous rest states. We also stress that we do not claim that pulse-replicating pulses necessarily possess unstable absolute spectra or that pulses with unstable absolute spectrum necessarily self-replicate; however, for the FitzHugh-Nagumo system, these two phenomena appear to be linked, which motivated our analysis of absolute spectra here.}

We briefly comment on prior mathematical work on self-replicating pulses. Early work \cite{NU1999} identified the role of (1) saddle-node bifurcations of standing single and multi-pulses and (2) solutions that connect single to double pulses at these saddle-node bifurcation points as the building blocks that enable pulse replication in, for instance, the Gray--Scott model. While, to our knowledge, the existence of the connecting solution has been established only numerically, several papers \cite{DGK98,NU1999,KWW2005,doelman2006homoclinic,chen2012simple,QiZhu2017} focused on proving the existence of saddle-node bifurcations and instability mechanisms of standing or slowly moving pulses in a range of models. Complementary to these results is work that focused on absolute instability mechanisms of traveling pulses that emerge near T-points, and we refer to \cite{Rademacher_thesis} and the references therein for details in relation to pulse replication. Our contribution is the direct link to a canard transition, where qualitatively the transition manifold $\Gamma_\eps$ discussed here plays the role of the connecting solution for the Gray--Scott model.

The remainder of the paper is organized as follows. In \S\ref{s:2}, we give a more comprehensive discussion of pulses in the FitzHugh--Nagumo system to motivate the analysis carried out in this paper. In \S\ref{s:proto}, we discuss two prototypical examples that illustrate our results. The general theory is developed in \S\ref{s:slowabsnew} and \S\ref{s:gentheory}, and our main results are then applied to the FitzHugh--Nagumo system in \S\ref{sec:applytofhn}. We end with a discussion of pulse replication in \S\ref{s:disc}.

\section{Pulses in the FitzHugh--Nagumo system}\label{s:2}

Our primary motivating example is the FitzHugh--Nagumo (FHN) equation
\begin{align}
\begin{split}
u_t &= u_{xx} + f(u) - w\\
w_t &= \eps(u-\gamma w),
\end{split}
\label{e:fhn}
\end{align}
with $f(u)=u(u-a)(1-u)$, $0<\eps\ll 1$, $\gamma>0$. Viewed in the spatial comoving frame $\zeta=x+ct$, we obtain the ODE for the spatial profiles of travelling waves
\begin{equation}
\label{e:twode}
\begin{aligned}
u_\zeta &= v\\
v_\zeta &= cv - f(u) + w\\
w_\zeta &= \frac{\eps}{c}(u-\gamma w).
\end{aligned}
\end{equation}
This system is known to admit homoclinic orbits, corresponding to traveling pulse solutions in~\eqref{e:fhn}, asymptotic to the homogeneous rest state $(u,w)=(0,0)$. For each $0<a<1/2$, and each sufficiently small $\eps>0$, it has been shown that there are two traveling $1$-pulses with distinct speeds $c_\mathrm{fast}(a,\eps)=\frac{1}{\sqrt{2}}(1-2a)+\mathcal{O}(\eps)$, and $c_\mathrm{slow}(a,\eps)=\mathcal{O}(\eps^{1/2})$, referred to as the ``fast" pulse and the ``slow" pulse, respectively. It is also known that these two distinct branches of pulses are connected near the critical value $a=1/2$, where the slow and fast pulse are no longer distinguishable~\cite{krupa1997fast}. When undergoing parameter continuation in $(a,c)$-space for fixed $\eps$, these two families of pulses form one continuous branch, which manifests in a ``backwards C-shaped" curve, when the speed $c$ of the pulses is plotted against the bifurcation parameter $a$~\cite{CSosc, champneys2007shil, guckenheimer2009homoclinic, guckenheimer2010homoclinic}; see the upper left panel of Figure~\ref{f:ccurvebanana}. Further, in the PDE~\eqref{e:fhn}, it is known that the fast pulses are spectrally and nonlinearly stable~\cite{jones1984stability,yanagida1985stability}, while the slow pulses are unstable~\cite{flores1991stability}, with an exchange of stability occurring as one transitions from the fast branch to the slow branch near $a=1/2$.

\begin{figure}
\hspace{.05 \textwidth}
\begin{subfigure}{.35 \textwidth}
\centering
\includegraphics[width=1\linewidth]{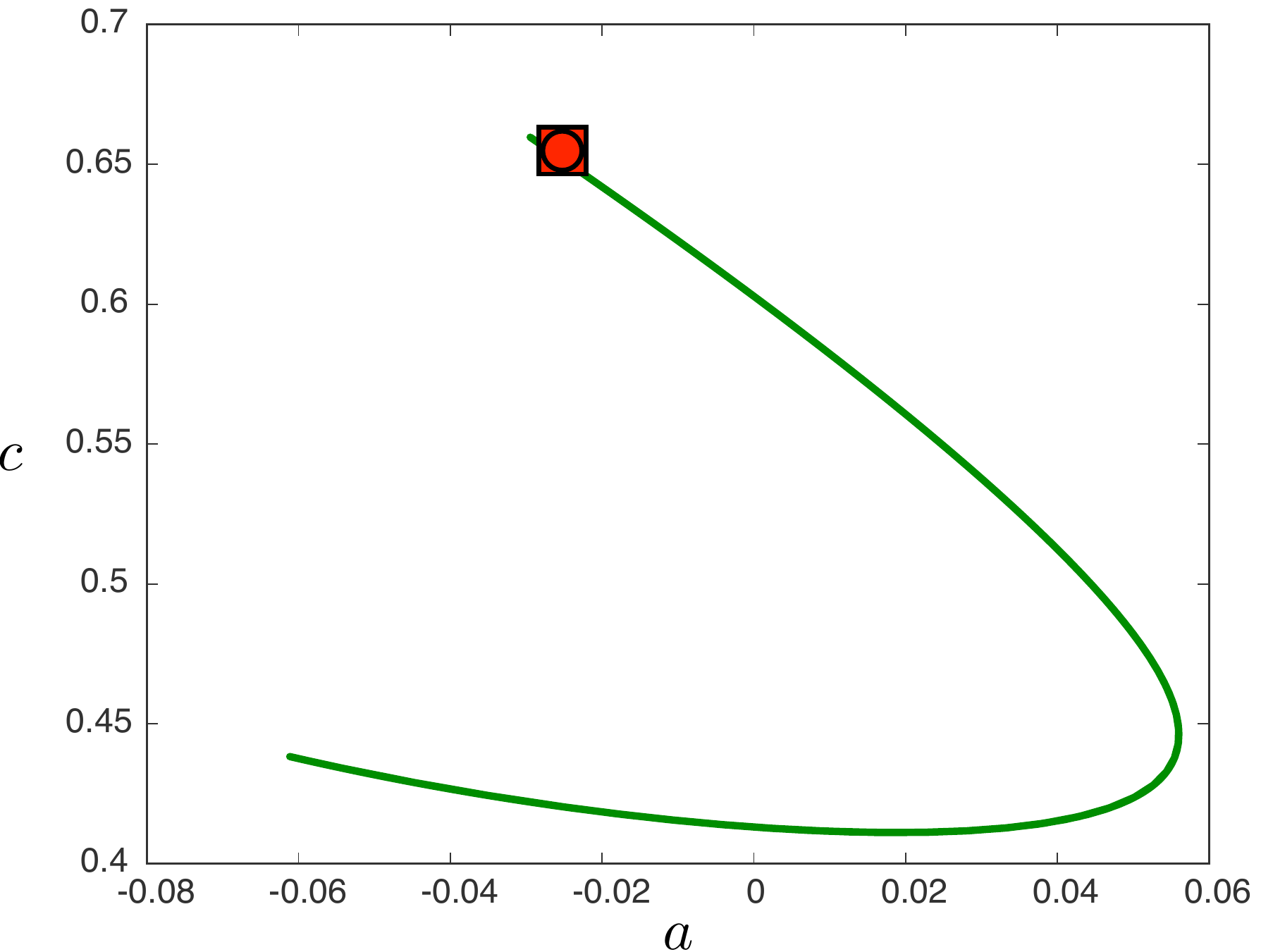}
\end{subfigure}
\hspace{.1 \textwidth}
\begin{subfigure}{.35 \textwidth}
\centering
\includegraphics[width=1\linewidth]{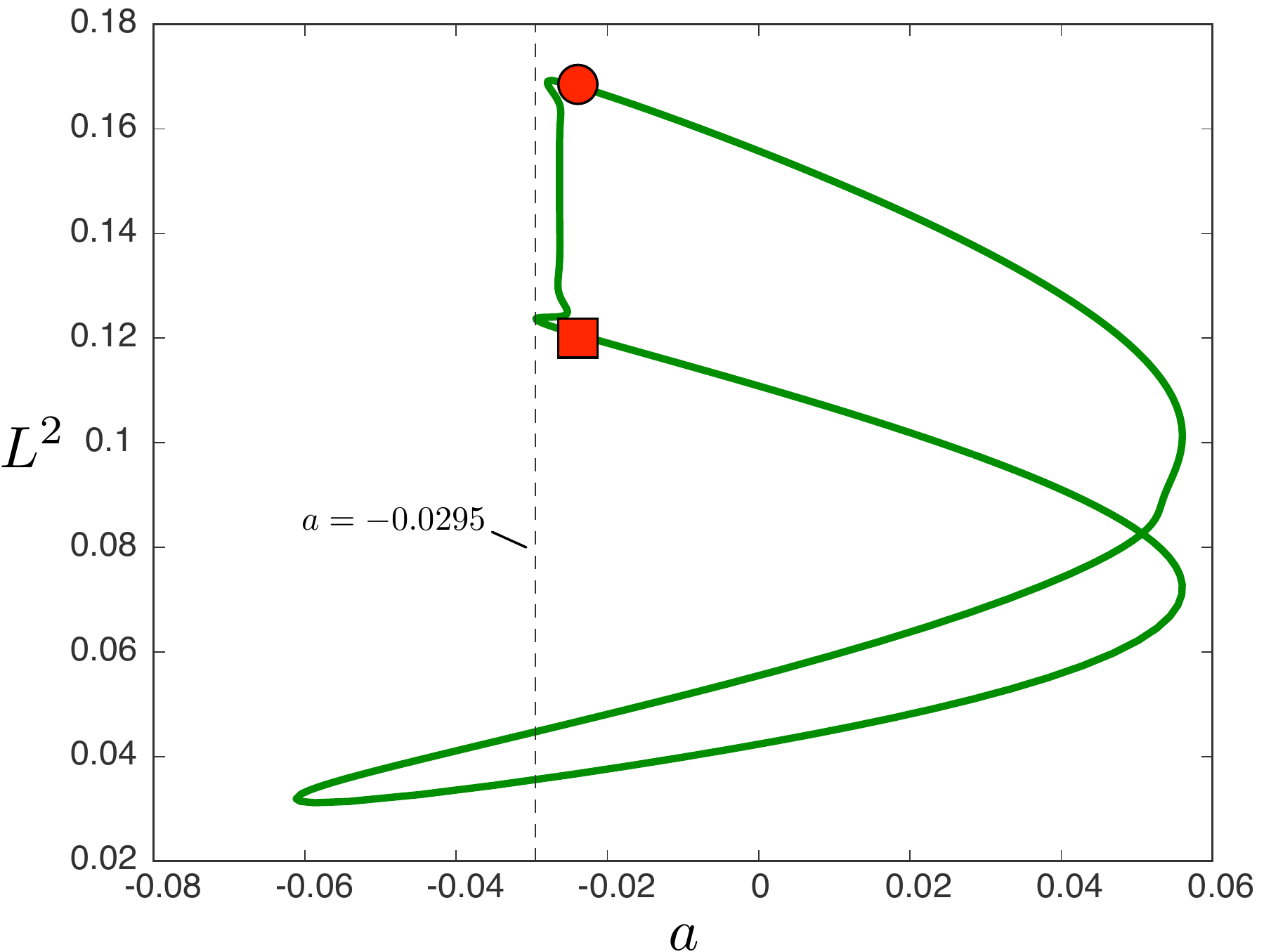}
\end{subfigure}\\

\hspace{.05 \textwidth}
\begin{subfigure}{.35 \textwidth}
\centering
\includegraphics[width=1\linewidth]{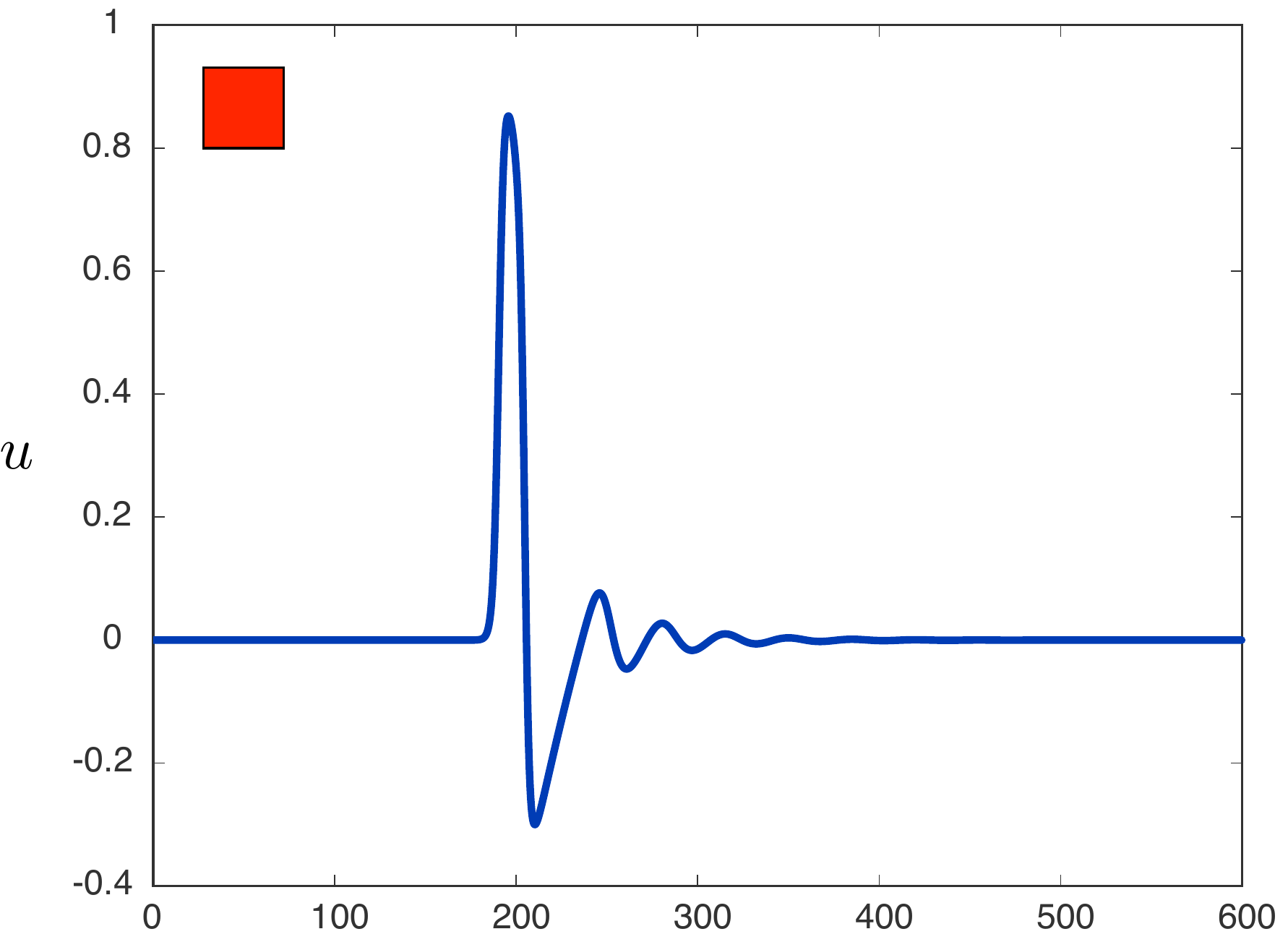}
\end{subfigure}
\hspace{.1 \textwidth}
\begin{subfigure}{.35 \textwidth}
\centering
\includegraphics[width=1\linewidth]{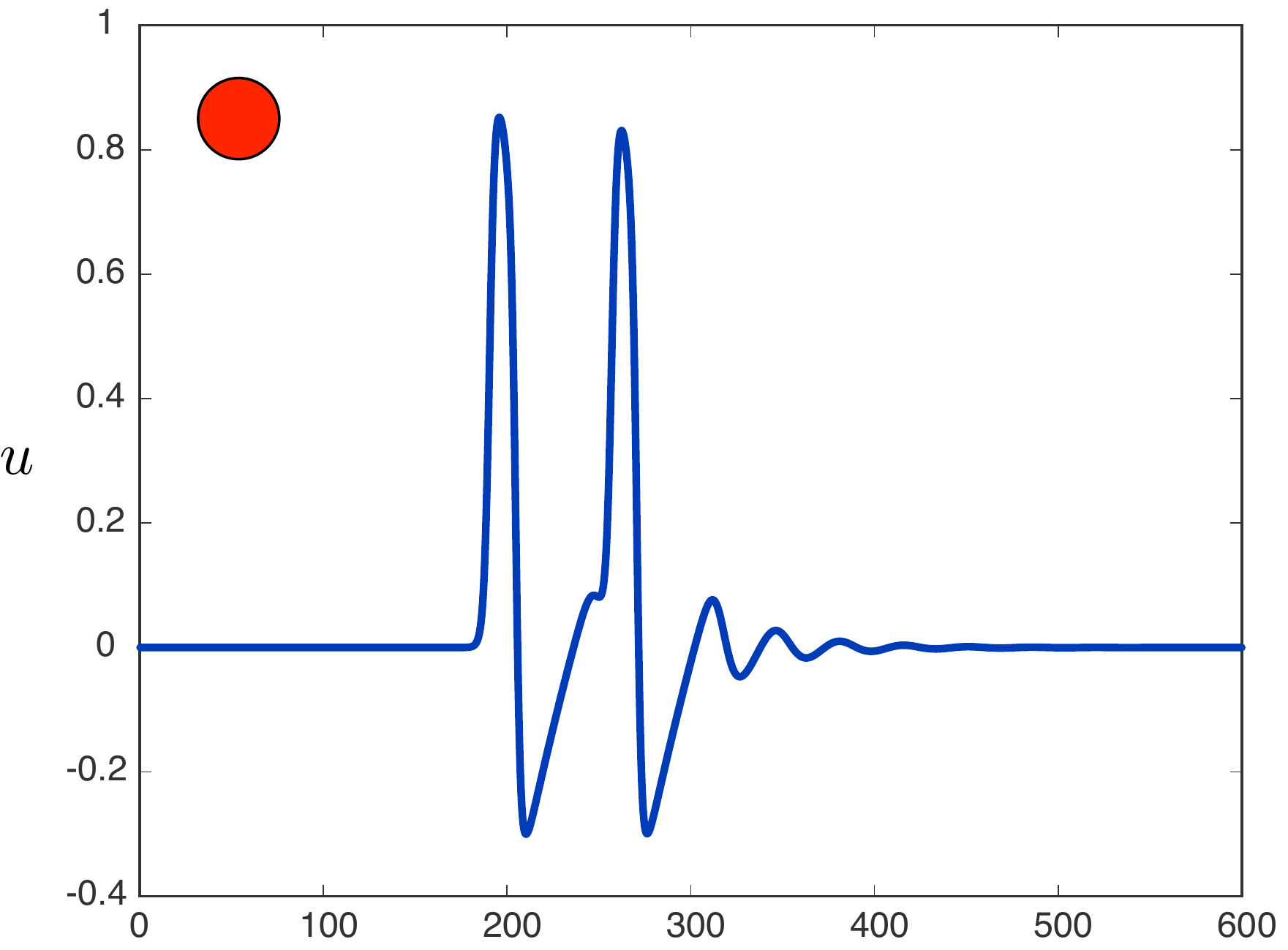}
\end{subfigure}
\caption{Plotted are the backwards C-curve (upper left panel) and homoclinic banana (upper right panel) obtained by parameter continuation in the traveling wave equation~\eqref{e:twode} in the parameters $(a,c)$ for $\eps = 0.015, \gamma=2$. The red square and circle refer to the locations of an $1$-pulse with oscillatory tail~(lower left panel) and $2$-pulse~(lower right panel), respectively, along the parametric pulse replication transition. \blue{We note that our choice of $\eps=0.015$ is not small enough to see that the branch passes close to $a=0.5$ and also not small enough for the wave speed of the slow pulse, predicted to be $\mathcal{O}(\eps^{1/2})$, to be small. We also remark that there is a second transition from slow 1- to slow 2-pulses with small amplitudes near $(c,a)=0$ for which no rigorous results are known. }}
\label{f:ccurvebanana}
\end{figure}

\subsection{Temporal and parametric pulse adding and eigenvalue accumulation}
When continuing the branch of fast pulses towards $a\approx 0$, it is known that small exponentially decaying oscillations develop in the tails of the pulses, due to a Belyakov transition which occurs when $a\sim \eps^{1/2}$, resulting in a pair of complex conjugate eigenvalues in the linearization at the equilibrium $(u,w)=(0,0)$~\cite{CSosc}. In particular this fact, along with a transversality condition, implies that the homoclinic solution corresponding to the fast pulse is a Shilnikov homoclinic orbit, and nearby one expects to find $N$-round homoclinic orbits for any $N$, composed of identical copies of the primary homoclinic~\cite{homburg2010homoclinic}. Continuing the branch of fast pulses further, into the region $a<0$, results in a phenomenon whereby the oscillations in the tails grow continuously along an almost vertical, canard-like transition into a secondary excursion which eventually becomes a nearly identical copy of the primary pulse, resulting in a $2$-round homoclinic orbit, or $2$-pulse~\cite{CSosc,CSbanana, champneys2007shil, guckenheimer2010homoclinic} (the corresponding bifurcation diagram obtained by plotting the $L^2$-norm of the solution versus $a$ forms one of the ends of the so-called ``homoclinic banana", see Figure~\ref{f:ccurvebanana}). In other words, along a continuous branch in parameter space, the fast $1$-pulses are connected to a branch of fast $2$-pulses which resemble two identical copies of the $1$-pulse; we refer to this phenomenon as \emph{parametric pulse adding}.

This $1$-to-$2$-pulse transition was analyzed extensively in~\cite{CSbanana}, and the mechanism which allows for this transition is rooted in the canard dynamics present near the equilibrium when $a\approx0$ \blue{and $c\approx1/\sqrt{2}$}. In fact, one can view this transition as a kind of canard explosion of homoclinic orbits. Much like the planar canard explosion of periodic orbits~\cite{krupaszmolyan20012}, the entire transition happens in an exponentially small region in parameter space, and thus all of the intermediate homoclinic orbits exist at nearly indistinguishable parameter values, despite the solutions themselves being well separated in norm. We (informally) restate the following from~\cite{CSbanana}.

\begin{theorem}{\cite[Theorem 2.2]{CSbanana}} \label{thm:mainexistence}
For each $0<\gamma<4$ and each sufficiently small $\eps>0$, there exists a one-parameter family $\Gamma_\eps$ of traveling pulse solutions $\phi(\cdot;s,\eps)$ with $(c,a)=(c,a)(s,\eps), s\in(0,8/27)$ to~\eqref{e:fhn} which is $C^1$ in $(s,\sqrt{\eps})$. Furthermore, for $s$ sufficiently small, the pulse solutions $\phi(\cdot;s,\eps)$ are single pulses with oscillatory tails, while the solutions $\phi(\cdot;s,\eps)$ are double pulses for $s$ sufficiently close to $8/27$. Away from the endpoints $s=0,8/27$, we have that
\begin{align}
(c,a)(s,\eps) = (c_*,a_*)(\eps)+\mathcal{O}(e^{-q/\eps})
\end{align}
for some $q>0$ where
\begin{align}
\begin{split}\label{e:maxcanard}
c_*(\eps)&=\frac{1}{\sqrt{2}}+\frac{3+2\gamma}{4}\eps +\mathcal{O}\left(\eps^{3/2}\right)\\
a_*(\eps)&=-\frac{(3+2\gamma)}{4\sqrt{2}}\eps +\mathcal{O}\left(\eps^{3/2}\right).
\end{split}
\end{align}
\end{theorem}
\begin{remark}
The parameter values~\eqref{e:maxcanard} denote the location of the so-called maximal canard associated with the canard point near the equilibrium at the origin, and the bulk of the transition -- that is, for those pulses not too close to the $1$-pulse near the beginning of the transition, nor too close to the $2$-pulse near the end of the transition -- occurs in a parameter regime exponentially close to the location of the maximal canard. We refer to these intermediate pulses during this `canard-explosion' portion of the transition as \emph{transitional pulses}.
\end{remark}
\begin{remark}
In~\cite{CSbanana}, the pulse solutions of Theorem~\ref{thm:mainexistence} are constructed in three pieces: a primary excursion, a secondary excursion, and an oscillatory tail. The primary excursion and oscillatory tail are nearly identical for most of the pulses in the family $\phi(\cdot;s,\eps), s\in(0,8/27)$, while the shape of the secondary excursion varies; see Figure~\ref{f:transition}. Thus, the pulses $\phi(\cdot;s,\eps), s\in(0,8/27)$ are parameterized by the secondary excursion, which undergoes a canard-explosion-like transition in phase space. See Remark~\ref{r:parameterization} for a more precise description of the relation between the parameter $s$ and the shape of this second excursion.
\end{remark}

\begin{figure}
\begin{subfigure}{.29 \textwidth}
\centering 
\includegraphics[width=1\linewidth]{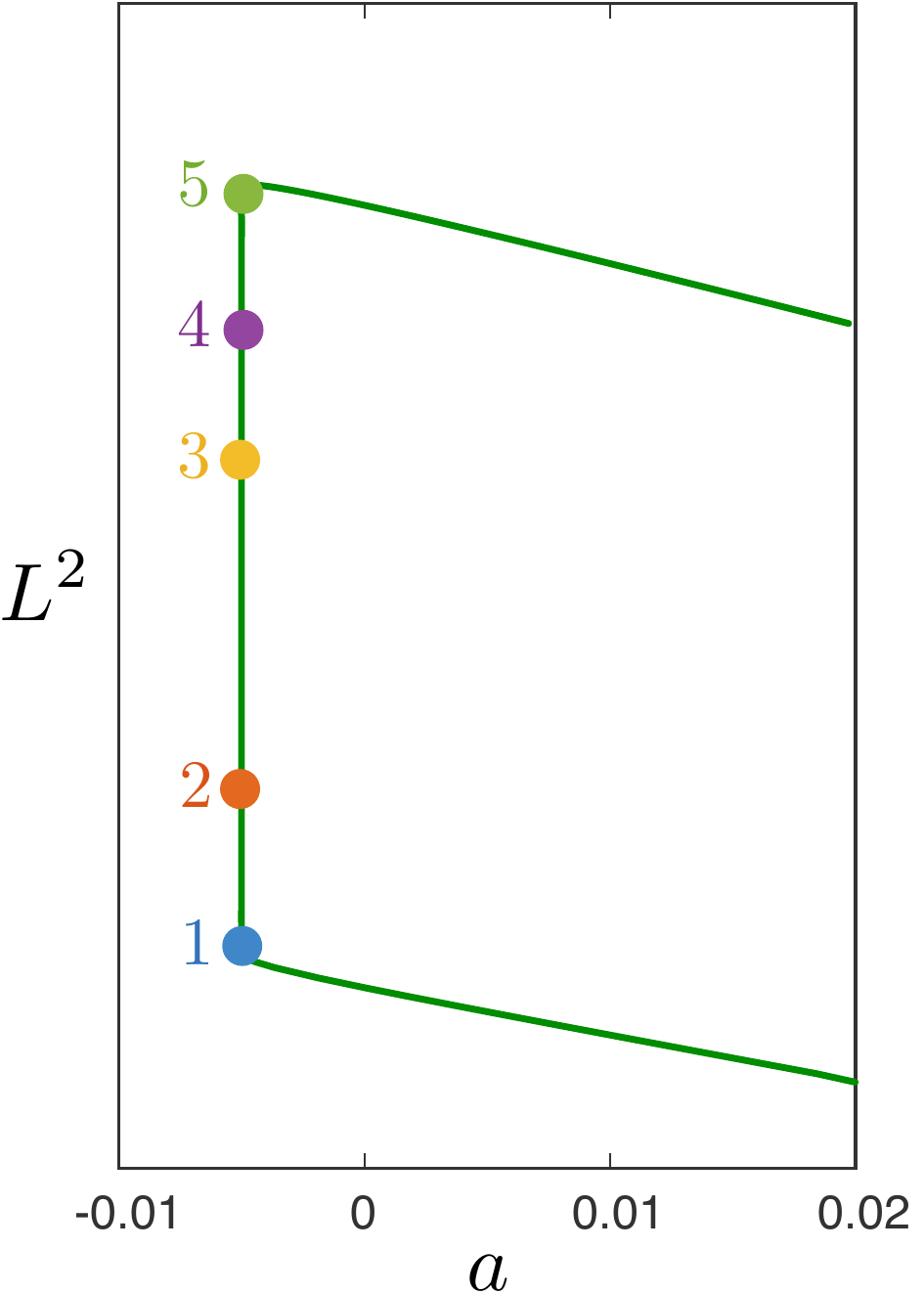}
\caption{Zoom of homoclinic banana: $L^2$-norm vs. $a$}
\label{f:bananazoom}
\end{subfigure}
\hspace{0.05 \textwidth}
\begin{subfigure}{.66\textwidth}
\centering
\includegraphics[width=1\linewidth]{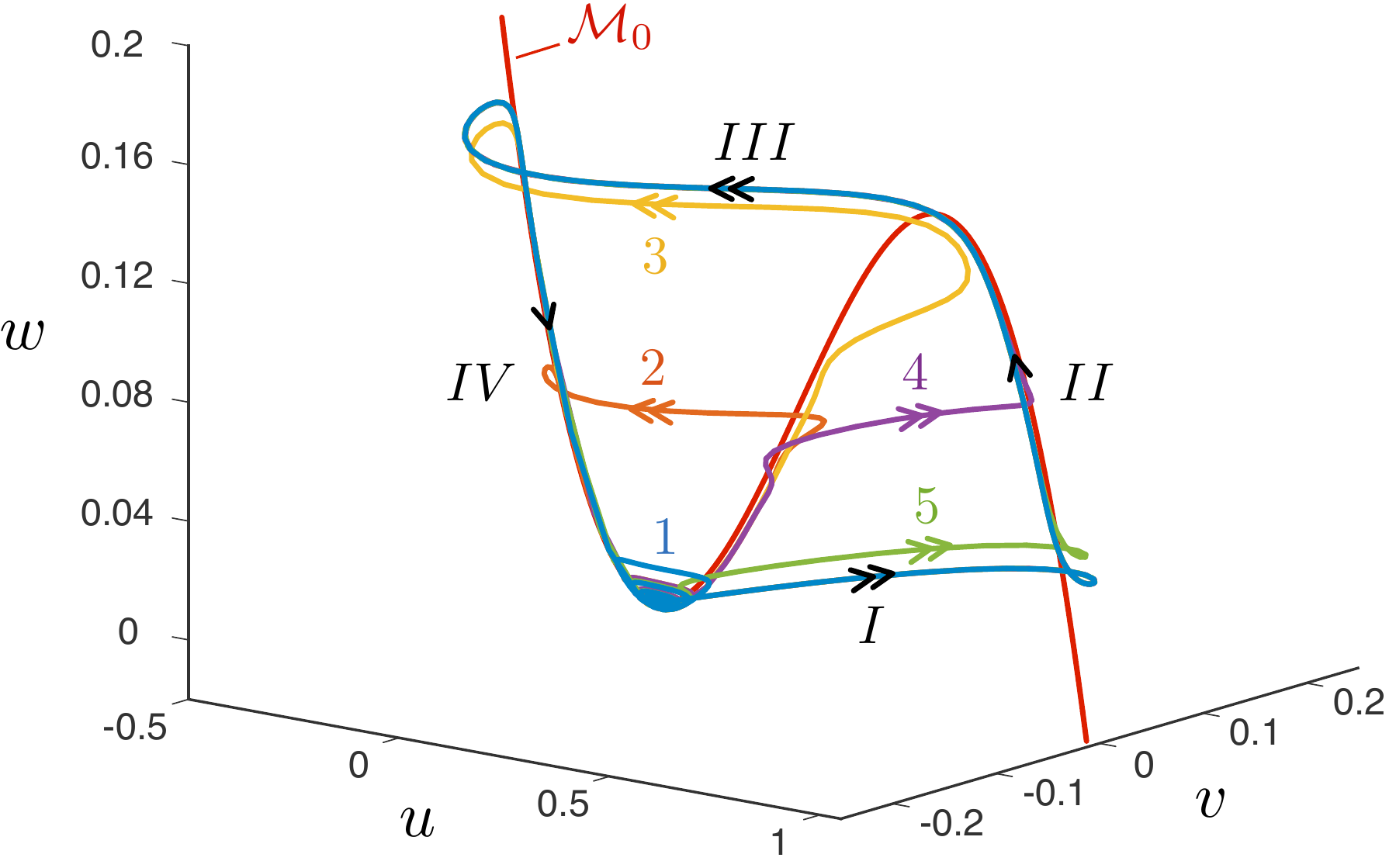}
\caption{Five pulses along the $1$-to-$2$-pulse transition in $uvw$-space. Also depicted is the cubic critical manifold $\mathcal{M}_0$ (red).}
\label{f:transition_uvw}
\end{subfigure}
\caption{Transition from single to double pulse in the top left of the homoclinic banana for $\eps=0.005$. Pulse $1$ follows the sequence $I,II,III,IV$. Pulses $2,3,4,5$ first follow the sequence $I,II,III,IV$, followed by a secondary excursion along the corresponding colored trajectory. 
}
\label{f:transition}
\end{figure}

We now turn to the behavior of the PDE~\eqref{e:fhn} in this parameter regime. As stated previously, the $1$-pulses are known to be nonlinearly stable, and it is further known~\cite{CdRS} that these pulses are also stable in the oscillatory tail regime, that is, near $a\sim \eps^{1/2}$. However, beyond this regime, in particular when $a<0$, the pulses are no longer guaranteed to be stable, and numerical evidence shows that under the temporal evolution of~\eqref{e:fhn}, an initial condition resembling a fast $1$-pulse exhibits an instability resulting in pulse-adding behavior in which additional pulses emerge dynamically from the oscillatory tail of the $1$-pulse; see Figure~\ref{f:temporal_replication} for the results of direct numerical simulation in~\eqref{e:fhn}. The additional pulses grow from the tail sequentially, and much in the way that the parametric transition occurs between the $1$ and $2$ pulse upon parameter continuation in the traveling wave equation~\eqref{e:twode} as described above. We refer to this phenomenon as \emph{temporal pulse adding}. We remark that the behavior is sensitive to parameters and initial conditions, especially in the limit $\eps \to0$, and that a variety of dynamic behavior, including both pulse adding and pulse splitting, can be observed in a narrow parameter regime.

We expect that these two phenomena, that of parametric vs. temporal pulse adding, in~\eqref{e:fhn} are linked, analogous to the observations for the Oregonator model \cite{Rademacher_thesis}. In particular, the fact that the two transitions heavily resemble each other, coupled with the fact \blue{that} all of the intermediate pulses along the parametric transition exist at computationally indistinguishable parameter values, leads one to believe that the parametric transition may guide the temporal transition in the PDE. A possible geometric explanation for this phenomenon is that the family of transitional traveling pulses comprise an invariant manifold for the PDE~\eqref{e:fhn}, along which the dynamics are characterized by a slow drift along this family.

A natural question therefore concerns the stability of the transitional pulses. As mentioned previously, the $1$-pulses with oscillatory tails near the start of the parametric transition have been shown to be stable~\cite{CdRS}, with the spectrum bounded away from the imaginary axis except for a translational eigenvalue at $\lambda=0$ and a critical real, negative eigenvalue $\lambda=\mathcal{O}(\eps^\nu)$ where $\nu\in [2/3,1]$ depends on the precise value of the parameter $a$. Furthermore, it has been demonstrated numerically~\cite[\S 7]{CSosc} (though to our knowledge not rigorously) that the $2$-pulse near the end of the parametric transition is unstable, with a single real, positive eigenvalue. The transition along the homoclinic banana (see Figure~\ref{f:ccurvebanana}, upper right panel) between the $1$ and $2$ pulse is comprised of a sequence of folds as the parameter $a$ wiggles back and forth. This leads us to posit two potential mechanisms for the appearance of the positive eigenvalue in the spectrum of the resulting $2$-pulse: 
\begin{itemize}
\item[(a)] the critical negative eigenvalue from the stable $1$-pulse crosses back-and-forth into the right half plane and finally remains in the right half plane after an odd number of such crossings. 
\item[(b)] a sequence of eigenvalues cross from the left to right half plane, with all but one of them crossing back into the left half plane by the end of the transition. 
\end{itemize}
We remark that the first case (a) appears in relation to other bifurcation phenomena, including snaking in the Swift-Hohenberg equation~\cite{burke2007homoclinic, burke2007snakes}, and in the spectrum of long wavelength periodic traveling waves which bifurcate from a localized pulse~\cite{SSlongwav}. However, in the present situation, we claim that it is the latter case (b) which occurs. In order resolve this, we continue the transitional pulses numerically while also tracking their spectra. The branch of pulse solutions is computed on a large spatial domain with periodic boundary conditions using numerical continuation with the {\tt matlab} based package {\tt pde2path} \cite{p2p}. It is well known that this provides an accurate approximation of the traveling pulses as well as their spectra \cite{SSabs, SSgluing}.  We remark that projection boundary conditions provide higher accuracy, but for our purposes this proved unnecessary. As usual for traveling waves we add a phase condition to fix the translational mode so that the associated zero eigenvalue is removed. The numerical discretization is based on one-dimensional first-order finite elements.

We focus on the part of the homoclinic banana containing the $1$-to-$2$ pulse transition, that is, the portion between the red square and red circle in Figure~\ref{f:ccurvebanana}. Figure~\ref{f:fhn-eps0p01} depicts the results of the continuation for $\eps=0.01$, where the continuation begins and ends, approximately, at pulses analogous to the $1$-pulse and $2$-pulse solutions depicted in the lower panels of Figure~\ref{f:ccurvebanana}, obtained for $\eps=0.015$. Labelled in Figure~\ref{f:fhn-eps0p01} are the location of certain pulses along the transition, along with the number of unstable eigenvalues present in their spectra, as well as crosses depicting the location of fold points along the transition. We immediately notice that following the branch from the bottom of the parametric transition to the top, the number of \blue{unstable} eigenvalues increases at each fold point until near the middle height and then decreases at subsequent folds, until a single \blue{unstable} eigenvalue remains. Eigenfunctions corresponding to unstable eigenvalues for three of the pulses along the continuation in Figure~\ref{f:fhn-eps0p01} are depicted in Figure~\ref{f:evecs}. \blue{We remark that the dense blue parts of the eigenvalues plotted in Fig.~\ref{f:fhn-eps0p01} lie close to the essential spectrum of the background state, which is not of interest here.}

\begin{figure}
\begin{subfigure}{1 \textwidth}
\includegraphics[width=0.95\textwidth]{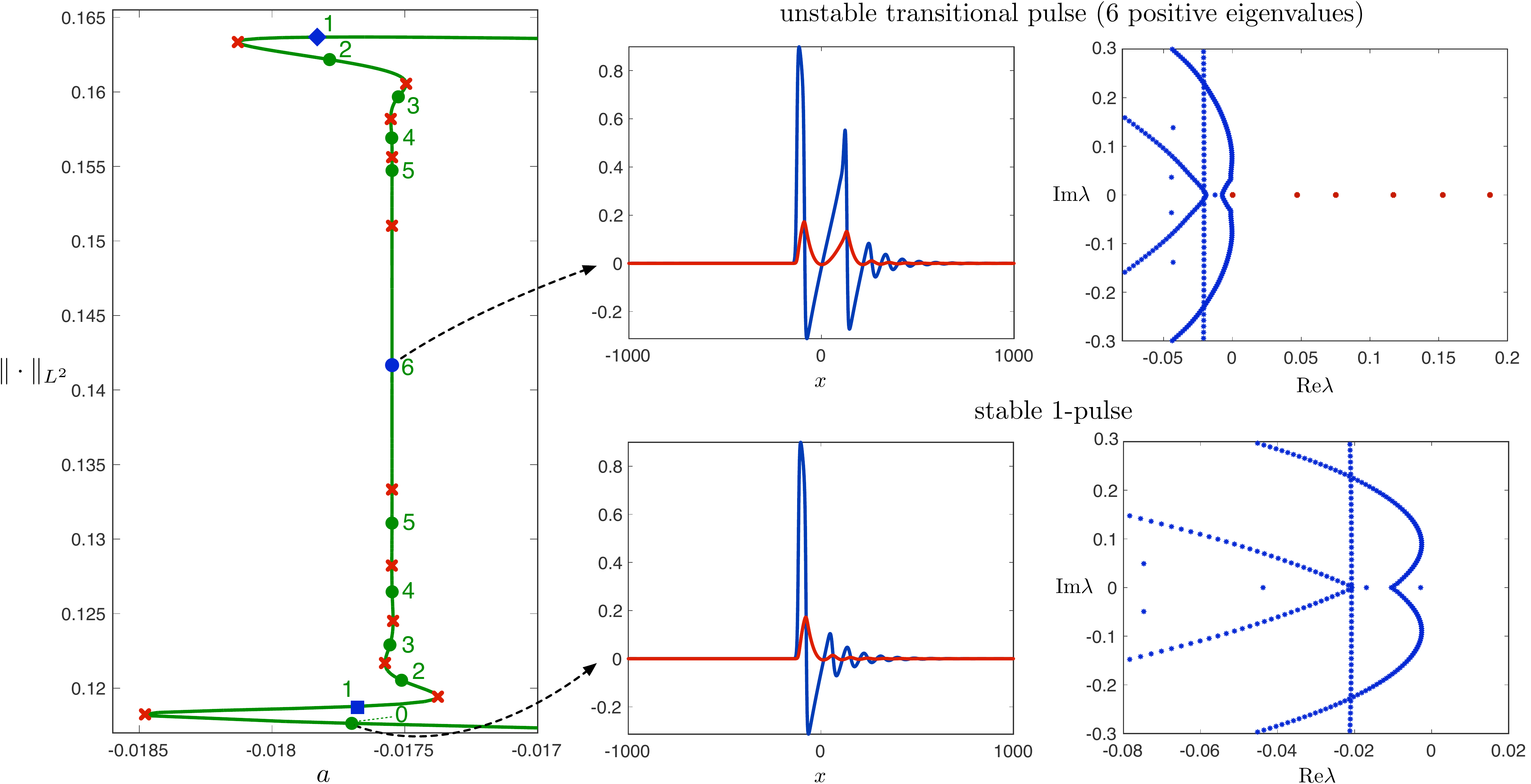}
\caption{\emph{Left panel:} Bifurcation diagram of the $1$-to-$2$ pulse transition: crosses mark fold points, and the numbers of unstable PDE eigenvalues is given for profiles along the transition marked with a colored bullet/shape. 
\emph{Middle and right panels:} Sample solution profiles ($u$ blue, $w$ red), along with their spectra (note that the translation eigenvalues at the origin are not shown).}
\label{f:fhn-eps0p01}
\end{subfigure}\\~\\
\begin{subfigure}{1 \textwidth}
\centering
\includegraphics[width=0.28 \textwidth]{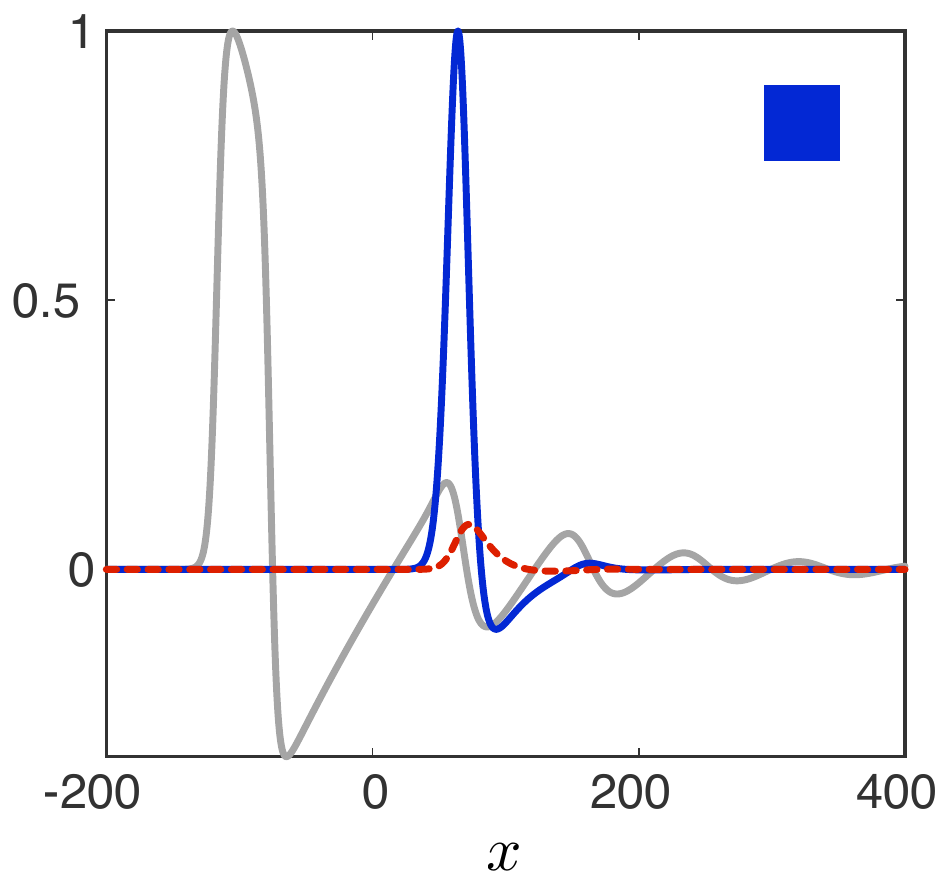} \hspace{0.02 \textwidth}
\includegraphics[width=0.28 \textwidth]{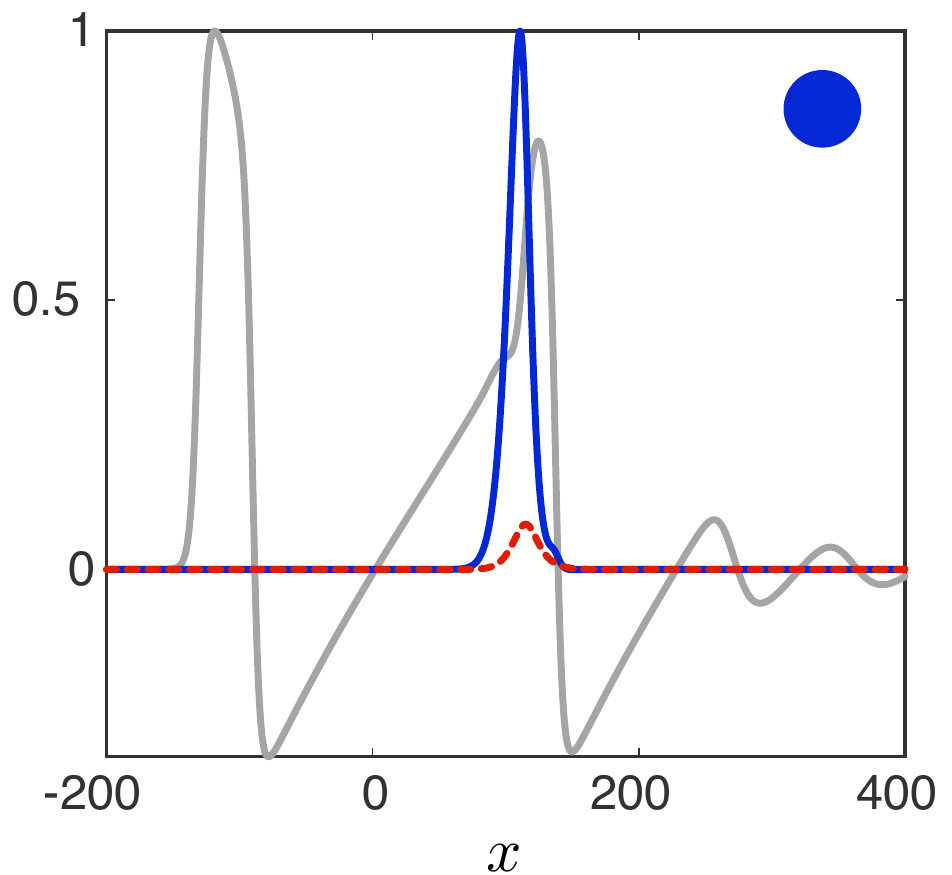}\hspace{0.02 \textwidth}
\includegraphics[width=0.28 \textwidth]{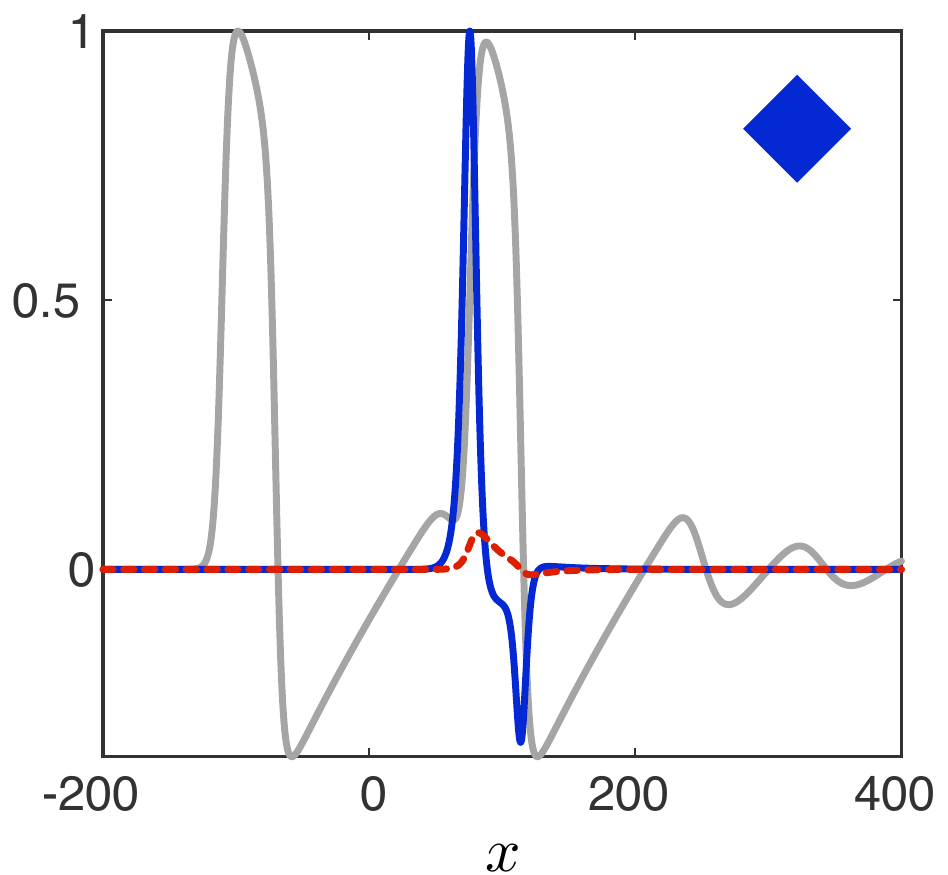}\\
\includegraphics[width=0.75\textwidth]{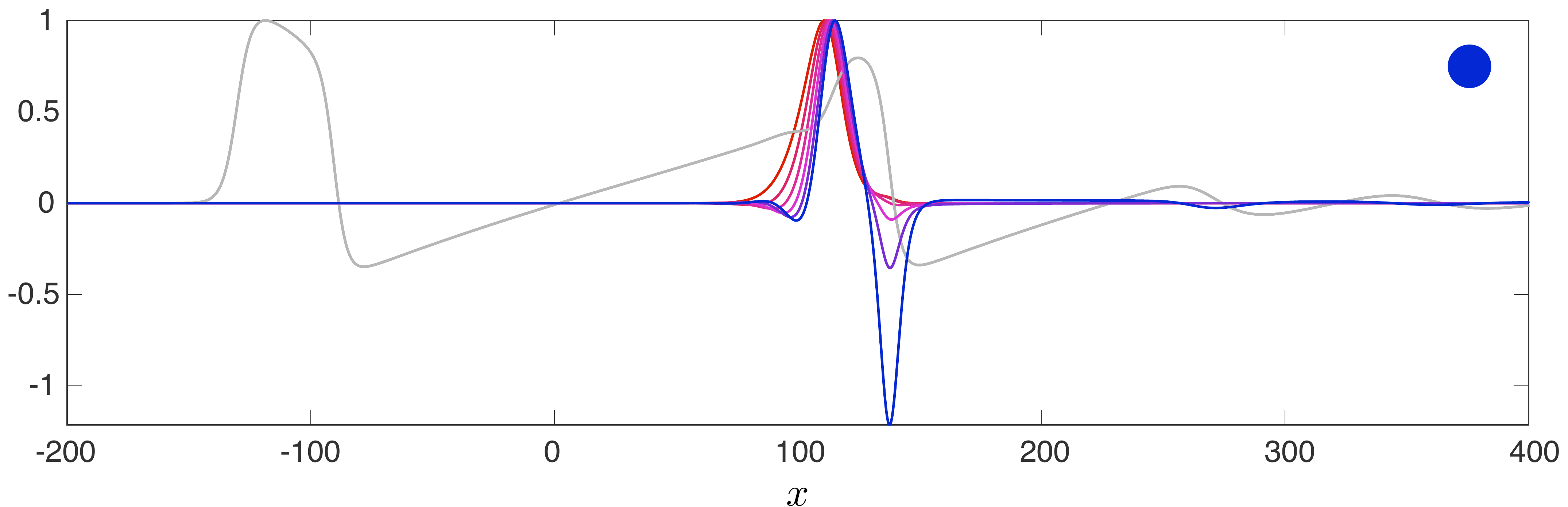}
\caption{\emph{Top row:} Plotted are the $u$-component of profiles of pulses (gray) with 1 (left), 6 (center) and 1 (right) unstable eigenvalue(s), corresponding to the solutions labelled with a blue square, circle, and diamond, respectively, along the continuation depicted in the left panel of (a). Vertical axis scale is normalised to the maximum of $u$ in each case. Also plotted are eigenfunctions for the most unstable eigenvalues (first component blue, second red dashed) in each case.
\emph{Bottom row:} The selected pulse is that with 6 unstable eigenvalues, and the first component of the eigenfunctions for all six unstable eigenvalues are plotted, ranging from blue (least unstable) to red (most unstable). }\label{f:evecs}
\end{subfigure}
\caption{Results of numerical continuation and spectral computations for~\eqref{e:fhn} for $\eps=0.01, \gamma=2$. }
\end{figure}

The goal of this work is to determine what mechanism is responsible for the appearance of this family of eigenvalues; in order to understand this, we must briefly review the existence analysis of the traveling pulses, and in particular the slow/fast geometry which allows for their construction.

\subsection{Geometry of traveling pulses}\label{sec:pulse_geometry}
The analysis of traveling pulses in the FitzHugh--Nagumo equation has a long history, and a range of techniques have been employed in their construction, including classical singular perturbation theory~\cite{hastings1976existence} and the Conley index~\cite{carpenter1977geometric}. Here we focus on the geometric singular perturbation approach~\cite{jones1991construction}, which was used extensively in~\cite{CSbanana} to construct the $1$-to-$2$-pulse parametric transition at hand. The basic idea of this construction is to separately analyze the singular fast and slow limits obtained by taking $\eps \to0$ in the traveling wave equation~\eqref{e:twode}, and its rescaled counterpart
\begin{equation}
\label{e:twode_slow}
\begin{aligned}
\eps u_\xi &= v\\
\eps v_\xi &= cv - f(u) + w\\
w_\xi &= \frac{1}{c}\left(u-\gamma w\right),
\end{aligned}
\end{equation}
respectively, where~\eqref{e:twode_slow} is obtained from~\eqref{e:twode} by rescaling $\xi=\eps \zeta$. This results in two limiting systems: The first is the layer problem, obtained by setting $\eps=0$ in~\eqref{e:twode}
\begin{equation}
\label{e:layer}
\begin{aligned}
u_\zeta &= v\\
v_\zeta &= cv - f(u) + w\\
w_\zeta &= 0,
\end{aligned}
\end{equation}
in which $w$ is no longer dynamic, and instead parameterizes the family of planar ODEs in the $(u,v)$-variables. The set of equilibria of this system, given by
\begin{align}
\mathcal{M}_0 :=\{(u,v,w):v=0, w=f(u)\}
\end{align}
is called the critical manifold. The second limiting system is the reduced problem, obtained by setting $\eps=0$ in~\eqref{e:twode_slow}
\begin{equation}
\label{e:reduced}
\begin{aligned}
0&= v\\
0&= cv - f(u) + w\\
w_\xi &= \frac{1}{c}\left(u-\gamma w\right),
\end{aligned}
\end{equation}
in which the flow is restricted to the critical manifold $\mathcal{M}_0$, and the dynamics on it are governed by the $w$-equation.

Singular solution profiles can be built by concatenating solutions from the two limiting systems, composed of portions of the critical manifold and fast jumps between different branches of the critical manifold, which arise as heteroclinic orbits in the layer problem~\eqref{e:layer}. Methods of geometric singular perturbation theory then guarantee that normally hyperbolic portions of the critical manifold perturb for $\eps>0$ as slow manifolds on which the flow is a perturbation of~\eqref{e:reduced}; homoclinic orbits in the full traveling wave equation can then be obtained for sufficiently small $\eps>0$ by building solutions which spend long times near these slow manifolds and matching these solutions across the fast heteroclinic jumps.

In fact, the entire parametric transition can be constructed in this fashion, with some technical caveats due to the fact that normal hyperbolicity is lost at several points along the critical manifold; however, while these issues add to the delicacy of the construction, they are not immediately relevant to the discussion of stability. In Figure~\ref{f:transition}, we plot several of the pulses along this transition for $\eps=0.005$. The equilibrium of the full system is located at the origin $(u,v,w)=(0,0,0)$, and all pulses are homoclinic orbits to this equilibrium. The right figure depicts five of the transitional pulses in $(u,v,w)$-space, as well as the critical manifold $\mathcal{M}_0$. At the start of the transition, we have the $1$-pulse with oscillatory tail which is composed of passage near the left and right branches of the cubic critical manifold, along with two fast jumps in between. As the transition progresses, the tail of the oscillatory pulse grows into a secondary excursion, which spends longer and longer (spatial) times along the middle branch of the critical manifold until reaching the upper right fold point. At this point, the secondary excursion then continues down the ``other side" of the middle branch, spending shorter and shorter times along the middle branch, until eventually tracing a second pulse identical to the primary pulse. The entire sequence is highly reminiscent of the classical planar canard explosion of periodic orbits.

\begin{figure}
\centering
\includegraphics[width=0.75\linewidth]{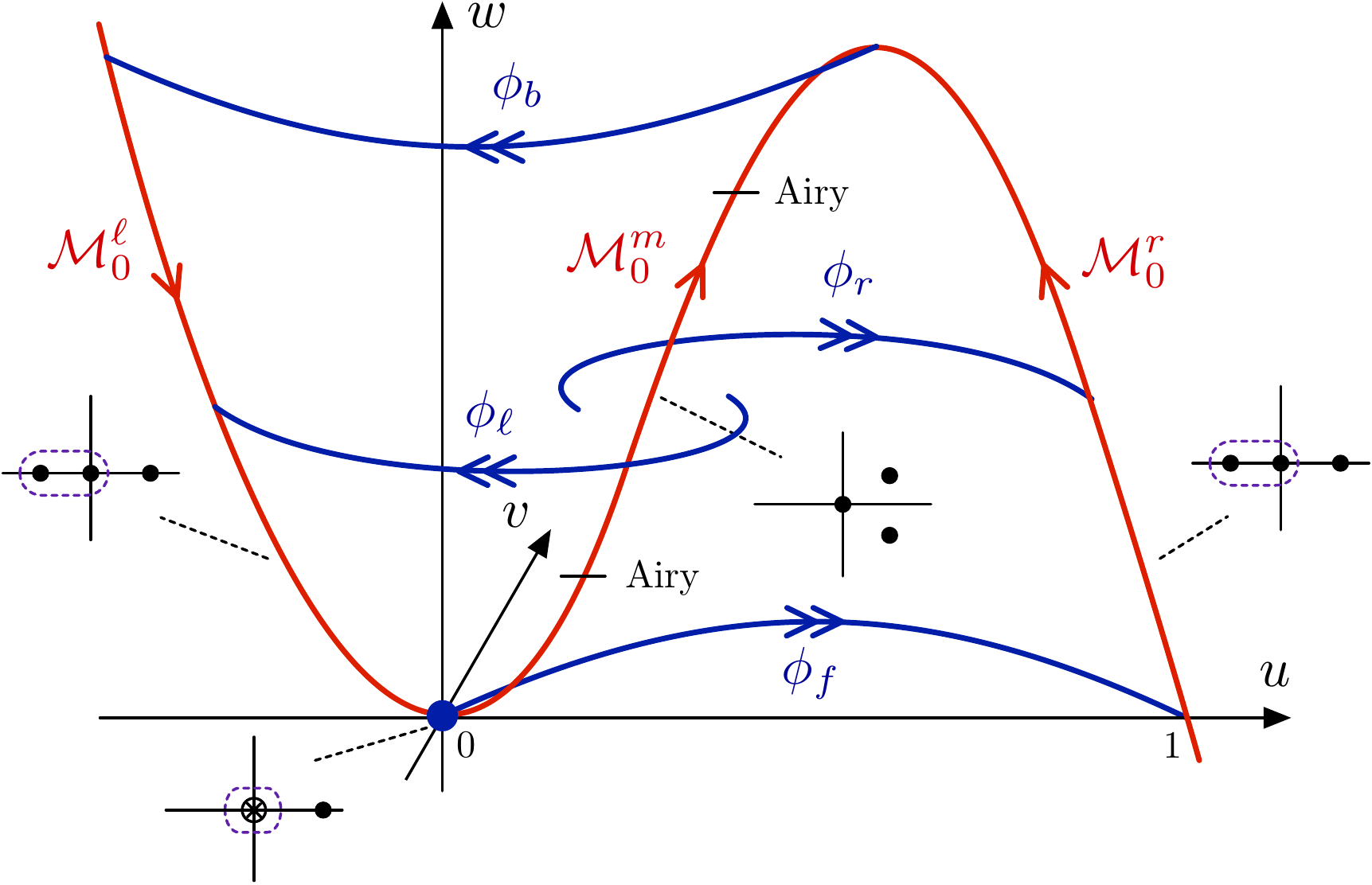}
\caption{Plotted is the singular limit associated with the traveling wave equation~\eqref{e:twode} for $(a,c,\eps)=(0,1/\sqrt{2},0)$. The critical manifold $\mathcal{M}_0$ is shown in red, with the dynamics on it governed by the reduced problem~\eqref{e:reduced}. Heteroclinic orbits, which arise as solutions to the layer problem~\eqref{e:layer} are shown in blue. Also depicted are the (spatial) eigenvalues of the linearization of~\eqref{e:twode} at various locations along portions of the slow manifolds traversed by a traveling pulse solution of~\eqref{e:twode}. We note that the middle branch $\mathcal{M}^m_0$ contains two "Airy points", at which the fast dynamics transition from node to focus behavior and vice versa.}
\label{f:singular_limit}
\end{figure}

Figure~\ref{f:singular_limit} depicts the singular limit associated with the traveling wave equation~\eqref{e:twode} for $a=\eps=0$ and $c=1/\sqrt{2}$, which is the singular limit from which all transitional pulses are constructed. The cubic critical manifold $\mathcal{M}_0$ is comprised of three normally hyperbolic branches
\begin{align}
\mathcal{M}^\ell_0 = \mathcal{M}_0\cap\{u<0\}, \quad \mathcal{M}^m_0 = \mathcal{M}_0\cap\{0<u<2/3\}, \quad \mathcal{M}^r_0 = \mathcal{M}_0\cap\{u>2/3\};
\end{align}
the left and right branches $\mathcal{M}^{\ell,r}_0$ are of saddle type, while the middle branch $\mathcal{M}^m_0$ is normally repelling. The  repelling middle branch takes the structure of a repelling focus for $u\in(u^-_{\mathrm{A},0},u^+_{\mathrm{A},0})$ and repelling node for $u\in(0,u^-_{\mathrm{A},0})\cup(u^+_{\mathrm{A},0},2/3)$, transitioning from node to focus and vice versa at so-called Airy points~\cite{CSbanana} $u=u^\pm_{\mathrm{A},0} $ where
\begin{align}
u^\pm_{\mathrm{A},0} = \frac{1}{3}\pm\frac{\sqrt{10}}{12}
\end{align}
The middle branch $\mathcal{M}^m_0$ meets the other two branches $\mathcal{M}^{\ell,r}_0$ at two nonhyperbolic fold points at $u=0,2/3$. The lower left fold point at $u=0$ is of canard type, while the upper right fold point is of generic fold type. 

Additionally, the layer problem~\eqref{e:layer} admits a collection of heteroclinic orbits for $\eps=a=0$ and $c=\frac{1}{\sqrt{2}}$. In particular, there exist fast jumps $\phi_{f,b} = (u_{f,b},u'_{f,b})$ from the upper and lower fold points of the critical manifold to its opposite branches $\mathcal{M}^{\ell,r}_0$. The orbit $\phi_f$ solves~\eqref{e:layer} for $w=0$, while $\phi_b$ solves~\eqref{e:layer} for $w=4/27$. Furthermore, for each $w\in(0,4/27)$ between the two folds, there are three equilibria in the layer problem~\eqref{e:layer}, and there exist heteroclinic orbits $\phi_{\ell,r}(\cdot; w) = (u_{\ell,r},u'_{\ell,r})(\cdot; w)$ which connect the unstable middle branch $\mathcal{M}^m_0$ to the left and right saddle-type branches $\mathcal{M}^{\ell,r}_0$. The structure of the \blue{layer} problem and the related heteroclinic orbits for different values of $w$ are depicted in Figure~\ref{f:layer}.

\begin{figure}
\centering
\includegraphics[width=0.8\linewidth]{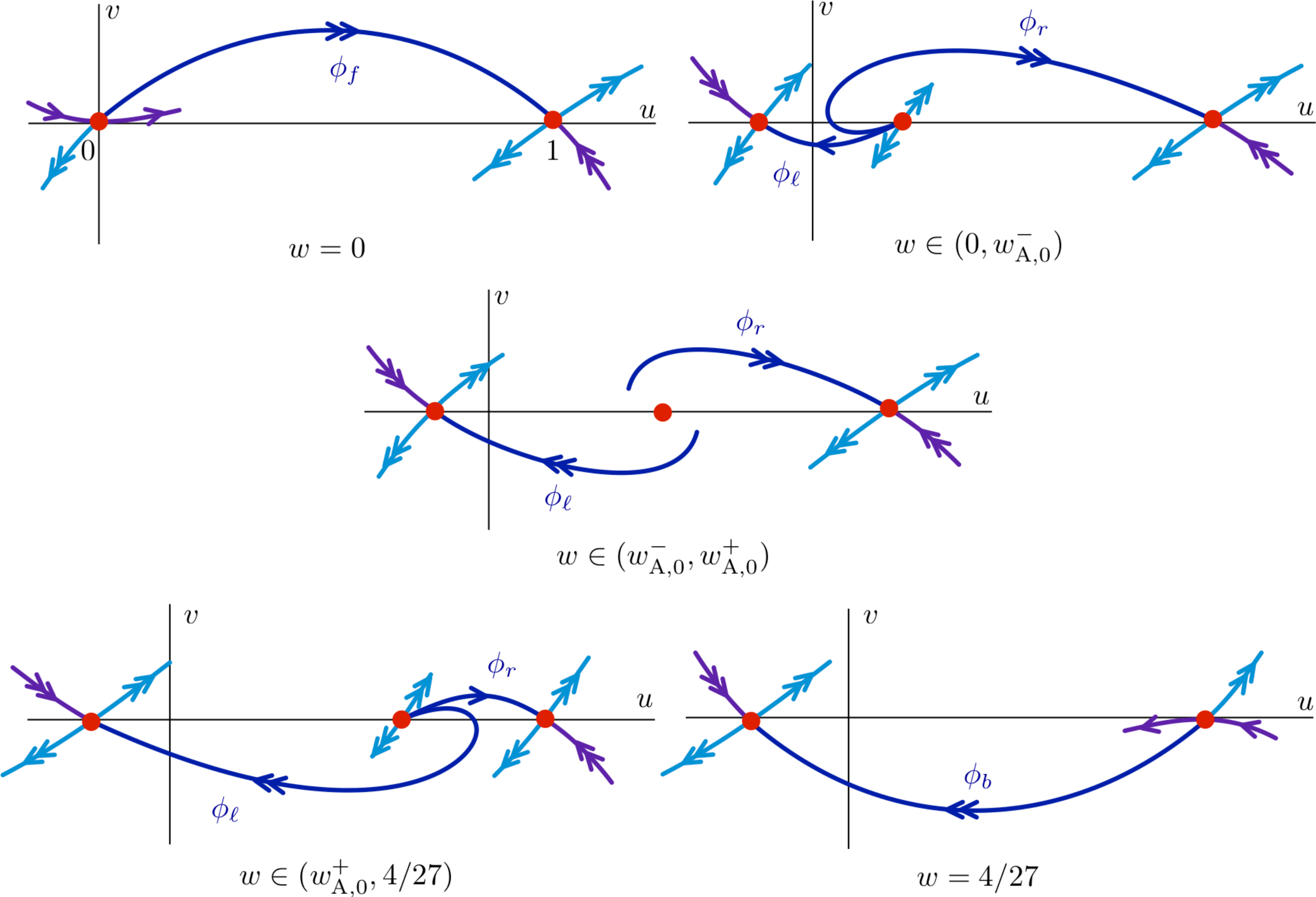}
\caption{Structure of the heteroclinic orbits $\phi_f,\phi_b,\phi_\ell,\phi_r$ present in the layer problem~\eqref{e:layer} for different values of $w\in[0,4/27]$. The structure of the repelling middle branch changes from node to focus and vice versa at the critical values $w=w^\pm_{\mathrm{A},0}=f(u^\pm_{\mathrm{A},0})$.}
\label{f:layer}
\end{figure}


From this, we can see how to build singular homoclinic orbits for each of the pulses along the parametric transition (see~\cite[\S2]{CSbanana}), and these singular profiles can all be constructed for the same parameter values. It is this construction that was used to prove their existence in~\cite{CSbanana}. 

\begin{remark}The construction of the pulses along this transition differs from that of the classical fast $1$-pulses in that passage near the middle portion of the critical manifold is necessary in the case of the transitional pulses; the classical fast pulses are constructed using only two fast jumps between normally hyperbolic portions of the two outer saddle-type branches of the critical manifold. We will see that passage near the middle branch is essential to understanding the appearance of the family of unstable eigenvalues along the parametric transition.
\end{remark}

\subsection{Slow absolute spectrum and summary of main results}\label{s:summary}
Following the discussion in the previous section, we see that the key to understanding the stability of traveling pulse solutions of~\eqref{e:fhn} lies in understanding the slow and fast pieces in their construction. One could expect that isolated eigenvalues arise due to the fast heteroclinic jumps (and indeed this is the case), but that eigenvalues could also arise due to passage near the slow manifolds. In Figure~\ref{f:singular_limit}, we plot the singular limit corresponding to the reduced~\eqref{e:reduced} and layer~\eqref{e:layer} problems as well as the spectrum of the three-dimensional linearization of~\eqref{e:twode} about points on the slow manifolds, which are slowly traversed by the pulse profile. While the two outer slow manifolds are of saddle-type, we draw attention to the portion of the unstable middle branch, in between the two Airy points at $u=u^\pm_{\mathrm{A},0}$, at which the layer dynamics transition from node to focus behavior. This creates a region along the middle slow manifold where the linearization of~\eqref{e:twode} possesses a pair of unstable complex conjugate spatial eigenvalues. As the dynamics near such manifolds is of course slow (in fact of $\mathcal{O}(\eps)$ on the fast timescale), passage near a slow manifold can be thought of as ``near-equilibrium" dynamics. Considering for the moment the points along slow manifolds as equilibria, we will later see that this structural change in the layer dynamics of the middle branch in between the Airy points has an influence on the PDE stability of this branch (when thought of as a family of equilibria). 

In a related context, it is known~\cite{SSgluing} that eigenvalues can accumulate for pulses, which are asymptotic to one equilibrium while spending long times near a second (different) equilibrium with a different relative Morse index (i.e. number of unstable spatial eigenvalues, computed in an appropriately weighted space). The accumulation set is the so-called absolute spectrum of that second equilibrium relative to the asymptotic equilibrium. Ignoring the precise definition of this set for a moment, we first note that in this case the number of eigenvalues increases asymptotically with $L$, where $L$ is the time spent near the second equilibrium~\cite{SSgluing}. In our context, the middle branch of the slow manifold plays the role of this second equilibrium, in that -- in an appropriate sense -- its absolute spectrum relative to the equilibrium of the full system is \emph{unstable}, and this results in the accumulation of unstable eigenvalues, where the number of unstable eigenvalues is related to the amount of time spent near this middle branch. In particular, this provides an explanation why along the 1-2-pulse transition shown in Figure~\ref{f:fhn-eps0p01} the number of unstable eigenvalues first increases, reaches a maximum at approximately the midpoint of the transition, and then decreases. Indeed, this exactly corresponds to the amount of (spatial) time the pulse profile spends near the middle branch; see Figure~\ref{f:transition}. Since the dynamics near the middle branch are on the slow timescale, the amount of time spent is of $\mathcal{O}(1/\eps)$ on the fast timescale, and hence we would expect the number of eigenvalues to increase as $\mathcal{O}(1/\eps)$. Indeed performing the same type of computation as in Figure~\ref{f:temporal_replication} when doubling the value of $\eps$, i.e., $\eps=0.02$ lends evidence to this prediction. The results are plotted in Figure~\ref{f:fhn-eps0p02}: As expected, the canard explosion is less pronounced and the number of fold points decreased. Accordingly, the maximal number of unstable eigenvalues near the absolute spectrum decreased from 6 to 3. 

\begin{figure}
\includegraphics[width=0.95 \textwidth]{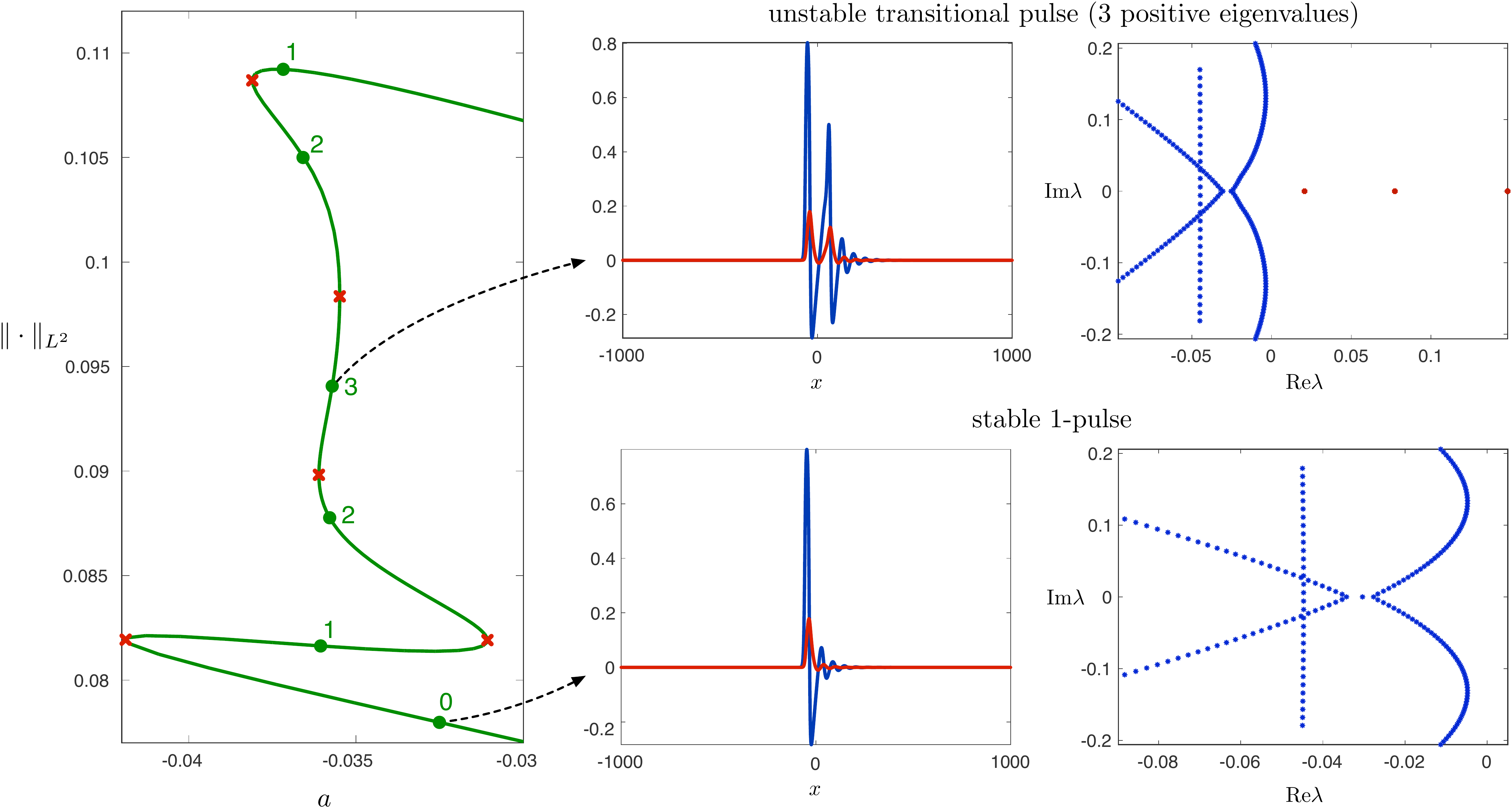}
\caption{Analogue of Figure~\ref{f:fhn-eps0p01} for $\eps=0.02$. Notably the number of folds is half of that for $\eps=0.01$ as is the largest numbers of unstable eigenvalues (translational eigenvalues are not shown). Comparing the profile of the transitional pulse with three unstable eigenvalues shown here with that of the pulse with six unstable eigenvalues from Figure~\ref{f:fhn-eps0p01}, the width of the secondary excursion and thus the spatial extent near the absolutely unstable slow manifold is larger for $\eps=0.01$.}
\label{f:fhn-eps0p02}
\end{figure}

\begin{remark}\label{r:accumulation_s_vs_eps}
We note that the implication for the family of transitional pulses $\phi(\cdot;s,\eps)$ is that there are effectively two ways to invoke the accumulation mechanism for eigenvalues from the absolute spectrum: (1) fix $s$ such that the $w$-component intersects $w\in (w_{A,0}^-,w_{A,0}^+)$ and decrease $\eps$ \blue{(leading to infinitely many unstable eigenvalues as $\eps\to0$ for fixed $s$)}, or (2) fix sufficiently small $\eps$ and traverse with $s$ through the 1-2-pulse transition \blue{(leading to finitely many unstable eigenvalues as $s$ changes for fixed $\eps$)}. However, both of these rely on modifying a sufficiently large (spatial) time that the profile spends near the absolutely unstable part of $\mathcal{M}^m_0$.
\end{remark}

The main result of our paper is an abstraction of this: for travelling waves whose profile passes near a slow manifold, we identify a set on which eigenvalues accumulate in the singular limit and give a rate for the accumulation. This result and the set is directly related to the absolute spectrum mentioned above, and we hence refer to it as \emph{slow absolute spectrum}. The precise formulation and required hypotheses will be given in \S\ref{s:slowabsnew}. 

As mentioned above for the FHN case, the key point is a relative change in (spatial) Morse index of the eigenvalue problem along the profile, that is, the number of stable spatial eigenvalues change in an appropriately weighted space. Roughly speaking, slow pieces with the same Morse index as the stable asymptotic state will not generate an accumulation of eigenvalues -- this can be viewed as an implicit result of the various known results on eigenvalues for slow-fast travelling waves. However, accumulation can occur on real eigenvalue parameters for which the Morse index changes along the profile relative to that of the asymptotic state. More specifically in our context, accumulation occurs near eigenvalue parameters for which a portion of a slow manifold traversed by the pulse has\red{, in any exponentially weighted space,} different spatial Morse index from the asymptotic state; in FHN, this occurs on the middle branch precisely between the two Airy points for $\lambda=0$, and, as we will show, on a decreasing portion of this branch for $\lambda$ on an interval of the positive real axis. 

Generalising from the FHN case, here we distinguish different types of entry/exit of the pulse into the part of the slow manifold with different relative Morse index -- see Figure~\ref{f:slowabscartoon}. For the pulses of FHN this occurs while the profile passes along the slow manifold -- for $\lambda=0$ exactly at the Airy points (see e.g. the entry of pulses $2,3,4$ into the Airy region in Figure~\ref{f:transition}) -- and we refer to this as the \emph{Airy transition case}. The other natural option is that the pulse enters/exits a region of slow manifold with a different Morse index along a fast jump; the pulses labelled $2,4$ exit the Airy region in this manner - see Figure~\ref{f:transition}. As we see from the pulses $2,4$ in Figure~\ref{f:transition}, a combination of these is possible; that is, the pulse can entry the region through an Airy transition, and then exit directly via a fast layer jump without passing through a second Airy transition.

\begin{figure}
\hspace{.025 \textwidth}
\begin{subfigure}{.45 \textwidth}
\centering
\includegraphics[width=1\linewidth]{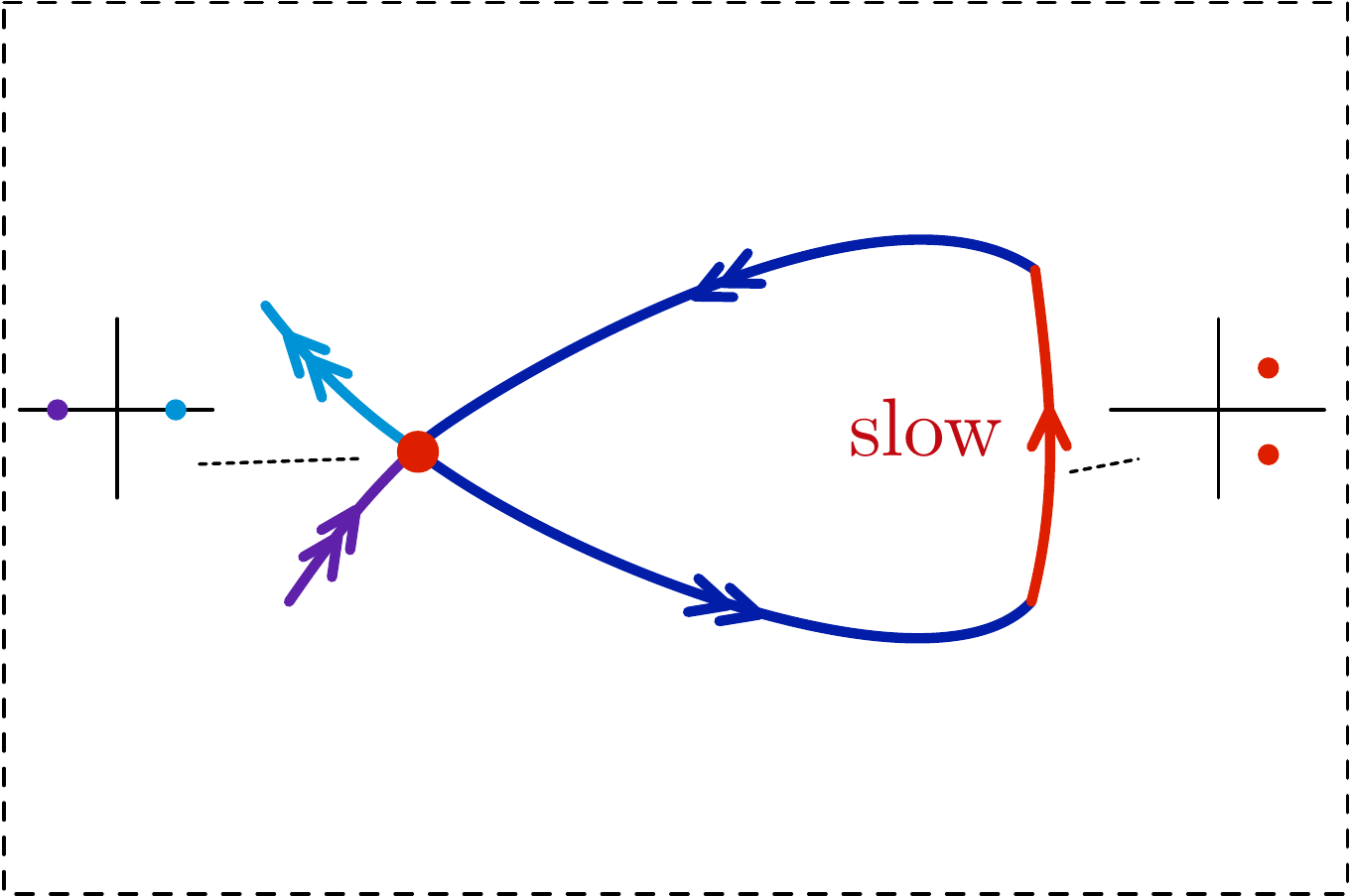}
\caption{}
\end{subfigure}
\hspace{.05 \textwidth}
\begin{subfigure}{.45 \textwidth}
\centering
\includegraphics[width=1\linewidth]{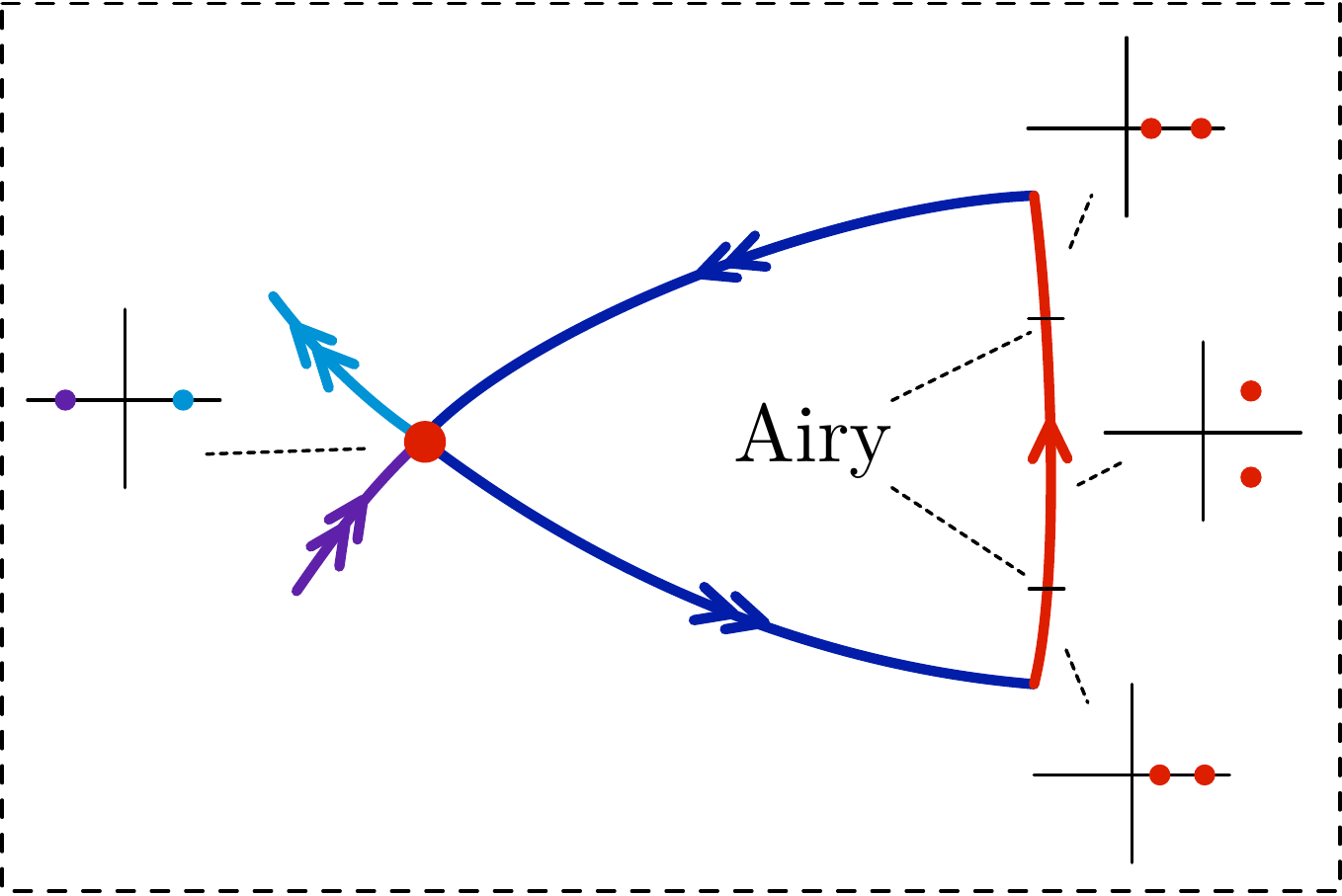}
\caption{}
\end{subfigure}
\caption{Illustrated are two distinct types of entry/exit along a slow manifold that exhibits\red{, in any exponentially weighted space,} a different Morse index from the asymptotic equilibrium. Panel~(a): Via fast jumps directly to and from the portion of the slow manifold that has a different Morse index. Panel~(b): Via a change of the Morse index during the drift along the slow manifold, through an ``Airy transition", and again fast jumps to or from the manifold.}
\label{f:slowabscartoon}
\end{figure}

\section{Spectra near slow manifolds -- prototypical examples}\label{s:proto}

In preparation of the main results and technicalities involved, we consider a few simple examples which demonstrate how the accumulation of eigenvalues can arise for eigenvalue problems with constant or slowly-varying coefficients, which are posed on increasingly large domains with fixed boundary conditions. We informally define the absolute spectrum $\Sigma_\mathrm{abs}$ in each case to be the set on which such accumulation occurs.

\subsection{Scalar equations with constant coefficients}

As a first (and perhaps simplest) example, recall the situation for the basic Sturm-Liouville problem
\[
\eps^2 u_{\xi\xi} = \lambda u,
\]
posed on $\xi\in[0,1]$ with homogeneous Dirichlet boundary conditions. The dispersion relation $\eps^2\nu^2=\lambda$ readily implies $\lambda = -(\eps n \pi)^2$, $n\in \Z_{>0}$, that is, a harmonic oscillator. Moreover, this implies an accumulation of eigenvalues as $\eps\to 0$ on any non-positive real number. Indeed, $\R_-$ is in this case the absolute spectrum $\Sigma_\abs$. Note that the corresponding eigenfunctions oscillate with increasing spatial frequency as $n\to\infty$, and that, upon rescaling $\xi=\eps \zeta$, this is equivalent to the eigenvalue problem 
\[
u_{\zeta\zeta} = \lambda u.
\]
posed on the domain $\zeta \in[0,1/\eps ]$ which grows without bound as $\eps \to0$.

\subsection{Scalar equations with slowly varying coefficients}

We now consider a slightly more general situation in which the coefficients of the eigenvalue problem are allowed to vary slowly; that is, they vary on the $\xi$-timescale, so that the eigenvalue problem now reads
\begin{equation}\label{e:proto}
u_{\zeta\zeta} = \Lambda(\eps \zeta;\lambda) u
\end{equation}
with slowly varying coefficient $\Lambda(\xi;\lambda)$ depending on both $\zeta$ and the eigenvalue parameter $\lambda$, where again we consider the domain $\zeta\in[0,1/\eps ]$ as $\eps \to0$. We also remark that an advection term $c u_\zeta$ (e.g.\ from a comoving frame) can be removed by an exponential weight, setting $\tu(\zeta) = {\rm e}^{-c\zeta/2} u(\zeta)$. The dispersion relation at a frozen value of $\xi=\xi_0$ reads $\nu^2 = \Lambda(\xi_0;\lambda)$ so that the spatial eigenvalues are
\[
\nu_\pm = \pm\sqrt{\Lambda(\xi_0;\lambda)}.
\]

Guided by the previous example (and standard Sturm-Liouville theory), one may expect that eigenvalues appear for values of $\xi_0,\lambda$ such that the quantity $\Lambda(\xi_0;\lambda)$ is non-positive, which gives a local (in $\xi$) frequency of the eigenfunction. The corresponding set of $\lambda$ on which eigenvalues accumulate corresponds to the absolute spectrum; in particular this spectrum is unstable if it extends into the region of positive $\lambda$. 

To see how eigenvalues accumulate on this set, we consider the specific example of $\Lambda(\xi;\lambda):= \Lambda_0 + \lambda -\Lambda_1 \xi$, $\Lambda_1>0$, $\lambda\in\R$, \blue{$\Lambda_0\in\R$}. Here we compute the spatial eigenvalues as
\[
\nu_\pm = \pm\sqrt{\Lambda_0 + \lambda -\Lambda_1 \xi_0},
\]
which gives the absolute spectrum as $\Sigma_\abs(\xi_0) = (-\infty,\Lambda_1 \xi_0-\Lambda_0]$. Since $\Lambda_1>0$, it is unstable for $\xi_0>\Lambda_0/\Lambda_1$ and stable otherwise, being marginal at $\xi_0=\Lambda_0/\Lambda_1$. 

In order to solve for eigenfunctions, it does not suffice to consider the system for frozen $\xi$; we solve on the full domain, taking \blue{ for simplicity again homogeneous Dirichlet } boundary conditions at $\xi=0$ and $\xi=1$, i.e.,  $\zeta$ varies over the interval $[0,1/\eps]$. Without loss of generality, we take $\Lambda_1>0$, and we note that the union of values of $\lambda$ which contribute to the absolute spectrum as the slow variable $\xi$ varies over the domain $[0,1]$ is given by 
\[
\slowabs([0,1]):=\bigcup_{\xi\in[0,1]}\Sigma_\abs(\xi)=\Sigma_\abs(1)=(-\infty,\Lambda_1-\Lambda_0],
\]
which we refer to as the associated \emph{slow absolute spectrum}. 

When solving the resulting boundary value problem, we distinguish two cases: The \emph{uniformly oscillating} case occurs if $\Lambda(\xi;\lambda)<0$ on $\xi\in[0,1]$, i.e., due to monotonicity $\lambda\in\mathring\Sigma_\abs(0)$. In this case, by Sturm-Liouville theory (and similarly to the case of constant coefficients), we find eigenfunctions which oscillate on the full domain, though now with a slowly varying frequency of oscillation. The second case, which we call the \emph{Airy transition} case, occurs if there is a sign change of $\Lambda(\xi;\lambda)$ on $[0,1]$; the corresponding eigenfunctions then oscillate only on part of the domain. Note this is necessarily such that the fixed $\lambda$ moves into the absolute spectrum. If this occurs for $\lambda=0$, the absolute spectrum changes from being stable to being unstable along $\xi$.  We discuss these cases in further detail below.

\paragraph{Uniformly oscillating case.} Let $\lambda\in \mathring\Sigma_\abs(0)$, where $\omega(\xi):= \sqrt{-\Lambda(\xi;\lambda)}\in\R$. Scaling $q:=\omega^{-1}u_\zeta$ gives 
\[
q_\zeta = \frac{u_{\zeta\zeta}}{\omega} - \frac{\omega_\zeta}{\omega^2}u_\zeta = 
-\omega u \blue{-} \frac{\eps\Lambda_1}{2\omega^2}q
\]
and thus \eqref{e:proto} can be written as the system
\begin{align}
\begin{pmatrix} u_\zeta\\ q_\zeta\end{pmatrix} = 
\begin{pmatrix}0 & \omega\\ -\omega & \blue{-}\eps \frac{\Lambda_1}{2\omega^2}\end{pmatrix} 
\begin{pmatrix}u\\q\end{pmatrix}. 
\end{align}
Polar coordinates $(u,q)=r(\cos(\theta),\sin(\theta))$ yield the angular equation
\[
\theta_\zeta = -\omega(\eps \zeta) \blue{-} \eps\frac{\Lambda_1 \sin(\theta)\cos(\theta)}{2\omega^2(\eps \zeta)} = -\omega(\eps \zeta) + \calO(\eps)
\]
since $\xi\in [0,1]$ and $\omega$ is bounded away from zero. We thus have a phase jump of $\theta$ along $\xi\in[0,1]$ given by
\begin{align}\label{e:protocond}
\theta(0)-\theta(1/\eps) = \int_{0}^{1/\eps}\omega(\eps \zeta)\rmd \zeta + \calO(1),
\end{align}
\blue{and for linear boundary conditions} $\lambda\in \mathring \Sigma_\abs(0)$ is an eigenvalue precisely \blue{ if } $\theta(0)-\theta(1/\eps)\equiv 0 \mod\pi$. 
With $\lambda$-dependence made explicit, we rescale \eqref{e:protocond} as
\[
\Delta(\lambda,\eps):= \eps (\theta(0)-\theta(1/\eps)) =  \eps \int_{0}^{1/\eps}\omega(\eps \zeta;\lambda)\rmd\zeta + \calO(\eps) 
= \int_{0}^1 \omega(\xi;\lambda)\rmd \xi + \calO(\eps)
\]
so that  the condition for an eigenvalue is
\begin{align}\label{e:protocond0}
\Delta(\lambda,\eps) \in \{\eps k\pi : k\in\Z\}.
\end{align}
Notably the right hand side forms a uniform grid with gap size $\eps\pi$.
The key observation is that $\Delta$ converges as $\eps\to 0$ to the average 
\[
\Delta(\lambda,0) = \int_{0}^1 \omega(\xi;\lambda)\rmd \xi,
\]
and if $\partial_\lambda \Delta(\lambda,0) \blue{\neq} 0$ then, for  $0<\eps \ll 1$, we have $\Delta(\lambda,\eps)$ varies linearly in $\lambda$ over a uniform interval. Hence, \eqref{e:protocond0} holds for a set of $\lambda$-values with uniformly increasing density as $\eps\to 0$ \blue{ and cardinality } $\calO(1/\eps)$. In particular, for real $\lambda$ we have
\begin{align}
\partial_\lambda \omega(\xi;\lambda) = -\blue{\frac 1 2} \omega^{-1}(\xi;\lambda)\blue{<}0,
\end{align}
which implies $\partial_\lambda \Delta(\lambda,0) \blue{<} 0$.

\medskip
\paragraph{Airy transition case.} Here we consider a sign change in $\Lambda(\xi;\lambda)$ between the boundary conditions. For simplicity \blue{ we take the boundary conditions} at $\xi=-1,1$ and perturb $\lambda= -\Lambda_0 $\blue{ , where $\Lambda(\xi;-\Lambda_0)=-\xi\Lambda_1$} so that the sign change occurs initially at $\xi=0$; we refer to this point as an Airy point. We also assume $0\blue{<}\Lambda_1\blue{<}1$ \blue{ so that $0> \Lambda(\xi;\lambda)> -1$} for $\xi\in\blue{(}0,1]$. 

Since $\Lambda$ passes through zero we cannot rescale $u_\zeta$ as in the previous case, but instead write \eqref{e:proto} in the canonical first order form with $q=u_\zeta$ as
\begin{align}
\begin{pmatrix}u_\zeta\\ q_\zeta\end{pmatrix} = 
\begin{pmatrix}0 & 1\\ \Lambda & 0\end{pmatrix} 
\begin{pmatrix}u\\q\end{pmatrix}. 
\end{align}
Note that the matrix changes from hyperbolic with eigenvalues $-\sqrt{\Lambda}<0<\sqrt{\Lambda}$ to elliptic with eigenvalues $\pm\sqrt{-\Lambda}\in \rmi \R$ as $\xi$ increases through zero. The angular equation for polar coordinates with angle $\psi$ reads
\begin{align}\label{e:psi}
\dot\psi = (1+\Lambda(\eps \zeta))\cos^2(\psi)-1 \blue{=: f(\eps\zeta, \psi).}
\end{align}
\blue{We are interested in the solution $\psi_\eps$ with $\psi_\eps(-1/\eps)=\pi/2$ at the left boundary and its phase dynamics until the right endpoint, $\psi_\eps(1/\eps)$. The scaled phase jump $\tDelta(\xi_2):= \eps (\psi_\eps(-1/\eps)-\psi_\eps(\xi_2/\eps))$ for $\xi_2\in(0,1)$ captures this over the Airy point and into a uniformly oscillating region, where the behaviour is as above. Since $\partial_\lambda f(\xi, \psi) = \cos^2(\psi)$ we have $\partial_\lambda\dot\psi_\eps(\zeta)$ is strictly increasing and $\partial_\lambda\psi_\eps(\zeta)>0$ for all $\zeta\in(-1/\eps,1/\eps)$. 
This means $\partial_\lambda\tDelta(\xi_2)<0$ for any $\xi_2\in(-1,1]$, even before the Airy point. Hence, in this case the behaviour over the Airy point amplifies that over the uniformly oscillating region as derived above. The accumulation of eigenvalues now follows  from convergence of $\partial_\lambda \tDelta(\xi_2)$ as $\eps\to 0$ to a non-zero value for any fixed $\xi_2\in(0,1)$, i.e., the left boundary condition on the uniformly oscillating region $\xi\in [\xi_2,1]$; we omit the details and refer to the more general results of \S\ref{s:gentheory}.}

\section{Absolute spectra near slow manifolds}\label{s:slowabsnew}

In the previous section, we showed how eigenvalues can accumulate when solving a linear boundary value problem with slowly varying coefficients as a timescale parameter goes to zero. We now extend this notion to PDE eigenvalue problems associated with traveling waves to singularly perturbed PDEs whose profiles pass near a slow manifold with real absolute spectrum relative to the asymptotic rest state. \blue{In this section, we set up the general assumptions and structure of the resulting eigenvalue problem, and in~\S\ref{s:gentheory} we present our results on eigenvalue accumulation in this general setting. }

\subsection{Setup}\label{s:slowabssetup}

\blue{We now describe the abstract problem that we will investigate and relate it to the PDE eigenvalue problem associated with a traveling pulse solution. In order to simplify the exposition we will focus on the abstraction of cases that occur in the FitzHugh--Nagumo system~\eqref{e:1}. 

In \S\ref{s:slowabsnew} and \S\ref{s:gentheory}, we focus on the following setup. 
Given $\eps_0,T>0$, we assume that $\xi$ lies in $[0,T]$, the eigenvalue variable $\lambda$ varies in a fixed compact domain in $\mathbb{C}$, and $\eps$ lies in the interval $[0,\eps_0]$. We assume that $A_\slow(\xi;\lambda, \eps)$ is a function with values in $\mathbb{C}^{n\times n}$ that is continuous in its arguments $(\xi,\lambda,\eps)$ and continuously differentiable in $(\xi,\lambda)$. We also assume that there is a constant  $C>0$ so that
\[
\|A_\slow(\xi;\lambda, \eps) - A_\slow(\xi;\lambda, 0) \| \leq C \eps
\]
uniformly in $(\xi,\lambda,\eps)$. We define $A(\zeta;\lambda,\eps):=A_\slow(\eps \zeta;\lambda, \eps)$ for $0\leq\zeta\leq T/\eps$ and consider the linear differential equation
\begin{align}\label{e:genslowlinear}
U_\zeta = A(\zeta;\lambda, \eps)U, \quad 0\leq\zeta\leq T/\eps,
\end{align}
with corresponding evolution operator $\Phi(\zeta,\tilde{\zeta};\lambda,\eps)$.
We assume that $Q^\pm(\lambda, \eps)$ are families of vector spaces in $\mathbb{C}^n$, which are continuous in $\eps\in(0,\eps_0]$ and continuously differentiable in $\lambda$ for $\eps\in(0,\eps_0]$. Finally, we assume that $A_\slow(\xi;\lambda, \eps)$ and $Q^\pm(\lambda, \eps)$ are real when $\lambda$ is real. Our goal is to identify, for each fixed $0<\eps\leq\eps_0$, the set of $\lambda$ for which \eqref{e:genslowlinear} has a nontrivial solution $U(\zeta)$ that satisfies $U(0)\in Q^-(\lambda,\eps)$ and $U(T/\eps)\in Q^+(\lambda,\eps)$. We will accomplish this task by comparing solutions to \eqref{e:genslowlinear} with solutions to the family
\begin{align}\label{e:genslowlinear0}
U_\zeta = A_\slow(\xi;\lambda, 0)U
\end{align}
of autonomous problems where $\xi\in I_\mathrm{slow}:=[0,T]$ enters as a parameter\red{, and where $\lambda$ lies in the slow absolute spectrum as described below}.

We now relate the abstract problem we just introduced to the eigenvalue problem associated with a traveling pulse solution $\phi_\eps(\zeta)$ in a singularly perturbed traveling wave ODE, where $\zeta=x+ct$ is the traveling wave coordinate, and $\eps$ is the timescale separation parameter. In general, the linear stability problem for the wave can be written as a first order nonautonomous linear system of the form (\ref{e:genslowlinear}); see, for instance, \cite{sandstede2002stability}, or~\S\ref{sec:applytofhn} below in the specific case of the FitzHugh--Nagumo system. Nontrivial bounded solutions of (\ref{e:genslowlinear}) for a given value of $\lambda$ that decay exponentially for $\zeta \to \pm \infty$ correspond to PDE eigenfunctions with eigenvalue $\lambda$; the set of such $\lambda$ characterizes the point spectrum associated to the pulse. 

We can characterize decaying solutions by inspecting the limiting constant coefficient system obtained by taking $\zeta \to \pm \infty$ in~\eqref{e:genslowlinear}. If the limiting matrix is hyperbolic, its stable/unstable eigenspaces can be evolved forwards/backwards to finite values of $\zeta$, resulting in $\zeta$-dependent subspaces $Q^\pm(\zeta;\lambda, \eps)$. Any eigenfunction must lie in the intersection of $Q^\pm(\zeta;\lambda, \eps)$, and evaluating these spaces at $\zeta=0,T/\eps$ results in the boundary conditions $U(0)\in Q^-(0;\lambda, \eps)$ and $U(T/\eps)\in Q^+(T/\eps;\lambda, \eps)$ mentioned in our abstract setup. 

Next, we describe the assumptions on the profile $\phi_\eps(\zeta)$ of the traveling pulse that will justify the hypotheses on the structure of the matrix $A(\zeta;\lambda,\eps)$ in \eqref{e:genslowlinear}. We assume that $\phi_\eps(\zeta)$ lies within distance $\mathcal{O}(\eps)$ of a normally hyperbolic slow manifold $\M_\eps$ for $\zeta \in [0,T/\eps]$ for some fixed $T>0$. We assume that the flow on the slow manifold is $\mathcal{O}(\eps)$, which allows us to write $\phi_\eps(\zeta)=\phi_\mathrm{slow}(\eps\zeta; \eps)+\mathcal{O}(\eps)$, where $\phi_\mathrm{slow}(\xi; \eps)$ is a solution in $\M_\eps$ formulated in the slow variable $\xi=\eps \zeta$. Finally, we assume that $\phi_\mathrm{slow}(\xi; \eps)$ is continuous in $\eps\in[0,\eps_0]$ for some $\eps_0>0$ with $|\phi_\mathrm{slow}(\xi; \eps)-\phi_\mathrm{slow}(\xi; 0)|\leq C\eps$. In other words, the base points of $\phi_\eps(\zeta)$ in the slow manifold $\M_\eps$ lie on a trajectory $S_\eps$ given by the slow orbit $\phi_\mathrm{slow}(\xi; \eps)$, and we have $S_\eps\to S_0$ as $\eps\to0$.  We remark that the endpoint $T$ may, in general, also depend on $\eps$, but we can remove this dependency through a linear rescaling of the independent variable $\zeta$ in Fenichel's normal form near the slow manifold $\M_\eps$.

In \S\ref{s:accumulation_withlayer}, we will generalize our abstract setup to the case where the traveling pulse solution $\phi_\eps(\zeta)$ leaves the $\mathcal{O}(\eps)$ neighborhood of the normally hyperbolic slow manifold $\M_\eps$ along a fast unstable fiber, thus adding a fast exit layer to its profile. The transition along the unstable fiber away from the slow manifold will occur for $\zeta$ in the interval $[T/\eps,T/\eps+L_\eps]$ of length $L_\eps=\calO\left(|\log\eps|\right)$. We will assume that the fast layer part of the profile converges to a fixed unstable fiber as $\eps\to0$. We will also need to modify our assumptions on the boundary subspace $Q^+(\lambda, \eps)$ as this subspace is now evaluated at $\zeta=T_\eps:=T/\eps+L_\eps$. The precise assumptions will be described in Hypotheses~\ref{h:bl} and~\ref{h:bl_vector} in \S\ref{s:accumulation_withlayer} and verified for the FitzHugh--Nagumo system in \S\ref{sec:applytofhn}.

Though the boundary value problem for \eqref{e:genslowlinear} can be considered independently of the travelling wave context, we will continue to make references to the family of travelling wave profiles $\phi_\eps$ and the slow manifold $\M_\eps$ as this eases the exposition.
}

\blue{In the remainder of \S\ref{s:slowabsnew}, we define the notion of slow absolute spectrum along the slow manifold $\M_\eps$ and collect some preliminary results concerning the above boundary value problem: In~\S\ref{s:slowabs} we define the slow absolute spectrum, and we outline several generic cases which arise for the resulting eigenvalue problem in~\S\ref{s:cases}. A key tool in setting up and solving this boundary value problem are exponential trichotomies, and we briefly review their definition in~\S\ref{s:exptrich}. In~\S\ref{s:exptrichslow}, we construct trichotomies along the slow manifold $\mathcal{M}_\eps$ and use them to write the relevant linear systems in block-diagonal form. In the case when $\M_\eps$ admits such slow absolute spectrum, we will use these trichotomies in \S\ref{s:gentheory} to solve the above boundary value problem (under certain regularity assumptions) and show that eigenvalues accumulate on this set as $\eps\to0$.}

\subsection{Slow absolute spectrum}\label{s:slowabs}

The key idea to explain the accumulation of eigenvalues is to view each point on the slow manifold $\M_\eps$ as a spatially homogeneous equilibrium point of the PDE. We are then led to accumulation of eigenvalues of the linearisation in such an equilibrium: it has been shown in \cite{SSabs} that, for the PDE posed on increasingly long intervals, eigenvalues accumulate precisely on the so-called absolute spectrum, which will be defined below in detail. 
In the slow-fast context, this naturally corresponds to extending the spatial interval near the slow part by decreasing the singular perturbation parameter. We may determine the absolute spectrum for a point on the slow manifold in the same way as for a homogeneous equilibrium state. However, we will consider real absolute spectrum only and only those connected components that consist of intervals. Let $\Sigma^0_\abs(u_0)\subset\R$ denote this part of the absolute spectrum for a point $u_0\in S_0$. 

We then define the \emph{slow absolute spectrum of a slow trajectory} as the union of these absolute spectra for all points on the slow trajectory in the singular limit, that is,
\begin{equation}\label{e:slowabs}
\slowabs(S_0) := \bigcup_{u_0\in S_0} \Sigma^0_\abs(u_0) \subset \R.
\end{equation}
We will show that for \eqref{e:genslowlinear}, as $\eps\to 0$, eigenvalues in the above sense for the problem on $[0,T/\eps]$ accumulate at each \red{$\lambda\in\slowabs(S_0)$} under regularity assumptions.

For the spectrum of a travelling wave on $\R$ there is a caveat: in the same way as for the case of a second equilibrium of \cite{SSgluing} mentioned above, only the `most unstable' part of the absolute spectrum is relevant. Let $\phi_\eps$ be a pulse  and $\Sigma^0_{\rm ess}$ the limit as $\eps\to 0$ of the essential spectrum of the asymptotic rest state. As in \cite{SSgluing}, the accumulation of eigenvalues is constrained to $\Sigma^0_\abs(u_0)\cap \Omega_1$, where $\Omega_1\subset\C$ is the connected component of $\C\setminus \Sigma_{\rm ess}^0$ that contains an unbounded part of $\R^+$. In particular, if the base state is stable, then $\Omega_1$ contains any unstable eigenvalues.

Towards determining the absolute spectrum following \cite{SSabs}, we first note that for well-posed PDE the number of unstable eigenvalues $i_\infty$ of the matrices $A$ is constant for $\Re(\lambda)$ above a uniform bound, cf. \cite{SSabs}, which we assume is uniform in $\eps$ as well. The corresponding condition for the boundary condition spaces is as follows.

\begin{hypothesis}\label{h:well}
There is $\rho\in\R$ such that the number $i_\infty=\#\{\nu\in\C|\nu\in\mathrm{spec}(A_\slow(\xi;\lambda,\eps)),\; \Re(\nu)\geq 0\}$ of unstable eigenvalues is independent of $\xi,\lambda,\eps$ \blue{for $0\leq\eps\ll1$, $\xi\in I_\slow:=[0,T]$, and all $\lambda\in\C$ with $\Re(\lambda) \geq \rho$}, and the boundary subspaces satisfy
\begin{align}
\dim(Q^-(\lambda,\eps))=i_\infty, \qquad \dim(Q^+(\lambda,\eps))=n-i_\infty.
\end{align}
\end{hypothesis}


\begin{definition}[\cite{SSabs}]\label{d:absspec}
For $0\leq \eps \ll 1$ we say that $\lambda\in\R$ is in the (regular real) absolute spectrum $\Sigma^\eps_\abs(\xi)$ for $\xi\in I_\slow$ of \eqref{e:genslowlinear} if the spatial eigenvalues $\nu_j$, $j=1,\ldots, n$ of $A_\slow(\xi;\lambda,\eps)$ ordered by decreasing real parts satisfy
\begin{align*}
\Re \nu_{i_\infty-1} >\Re \nu_{i_\infty}= \Re \nu_{i_\infty+1}>\Re \nu_{i_\infty+2},
\end{align*}
with $i_\infty$ from Hypothesis~\ref{h:well}, and $\lambda$ is not isolated in $\R$ with this property. If  $\nu_{i_\infty}=\nu_{i_\infty+1}$ we refer to $\lambda$ as a pinched double root. 
\end{definition}

Note that since $\lambda\in \Sigma^\eps_\abs(\xi)$ is real we have $\nu_{i_\infty} =\overline{ \nu_{i_\infty+1}}$ are complex conjugate. Note also that we switched to parameterise by $\xi\in I_\slow$ rather than by $u\in S_\eps$, using the relation $u=\phi_\slow(\xi;\eps)\in S_\eps$ in the singular limit. The slow absolute spectrum is now defined by \eqref{e:slowabs}. 

\subsection{Cases for slow passage with slow absolute spectrum}\label{s:cases}

In preparation of determining eigenvalues, we discuss the different cases for setting up the associated boundary value problem as already outlined in \S\ref{s:summary}, cf.\ Figure \ref{f:slowabscartoon}. We first disregard any fast part along $\phi_\eps$. In analogy to the examples in \S\ref{s:proto}, we refer to a given $\lambda$ that lies in the absolute spectrum along the entire slow part as a `uniformly oscillating' case, while $\lambda$ on the boundary of the absolute spectrum for some $\xi_\Airy(\lambda) \in I_\slow=[0,T]$ is a pinched double root, which occurs in an `Airy transition'. In particular, we have the following cases:

\begin{itemize}
\item \textbf{Uniformly oscillating}: Here $\lambda\in\cap_{\xi\in I_\slow}\Sigma^0_\abs(\xi)$ and there is $\delta>0$ such that $|\Im(\nu_{i_\infty}-\nu_{i_\infty+1})|\geq \delta$ on $I_\slow$ or a subinterval $I_\osc$. This case lies at the core of all our accumulation results. If appropriate boundary conditions can be prescribed at the endpoints of this interval, then we will show that this system can be solved in a straightforward manner, similar to the prototype equations in~\S\ref{s:proto}.
However, in practice it will only be possible to prescribe boundary conditions outside of this interval. In particular, the solution will enter/exit a uniformly oscillatory region by passing through either an Airy transition or a fast layer. 

\item \blue{\textbf{Airy transition entry}: Here $\lambda$ is such that there is $0<\xi_\Airy(\lambda)<T$ so that $\lambda$ a pinched double root at $\xi=\xi_\Airy(\lambda)$, $\lambda\notin\red{\slowabs}([0,\xi_\Airy))$, 
and for each $\xi_+>\xi_\mathrm{Airy}(\lambda)$ the uniformly oscillating case occurs on $\left[\xi_+,T\right]$. In other words, a uniformly oscillating region is encountered as the underlying traveling wave passes \emph{along the slow manifold itself}. 
We refer to the location of the pinched double root $\xi=\xi_\Airy(\lambda)$ as an Airy point.}

Since the passage through Airy transition is still on the slow timescale, the boundary value problem across this region can be solved in a similar manner as in the uniformly oscillating case. In the case of FitzHugh--Nagumo the relevant absolute spectrum destabilises and restabilises through the origin when moving along the slow manifold, i.e., as Airy transitions. Hence, any $\lambda\in(0,\lambda_*)$ will be an accumulation point with $\lambda_*$ the maximally unstable point of these absolute spectra. 

\end{itemize}
 
 \blue{In summary, fixing a value of $\lambda_0\in \mathbb{R}$ and considering values of $\lambda$ in a small neighborhood of $\lambda_0$ in $\mathbb{C}$, in the canonical case -- without fast layers -- the interval $I_\slow=[0,T]$ on the slow time scale decomposes as 
\begin{equation}\label{e:decomp}
[0,T] = I_\slow = I_\Airy \cup I_\osc,
\end{equation}
where in the `Airy case' $I_\Airy=[0,\xi_+]$ is a non-trivial compact interval that contains an Airy point in its interior, and $I_\osc=[\xi_+,T]$ is a non-trivial compact interval with uniform oscillations. In the `uniformly oscillating case', we treat the interval $I_\mathrm{Airy}$ as empty, i.e. $\xi_+=0$. 

We note that the transition point $\xi_+=I_\Airy\cap I_\osc$ is not unique and we shall make a suitable choice below.  In Figure~\ref{f:intervals} we show the notation in relation to the interval $[0,T]$ that we have introduced so far and will further develop in the following sections. We refer to~\S\ref{s:exptrichslow} for the precise hypotheses concerning the uniformly oscillating and Airy transition cases, as well as the construction of exponential trichotomies for~\eqref{e:genslowlinear} on the interval $\zeta\in\eps^{-1}I_\slow$.  We emphasize that only minor modifications to the following analysis are required to consider other configurations, e.g. the inclusion of an Airy transition exit interval.

 Lastly we will consider the addition of fast layers. We briefly mention this case here, but we delay discussion of the precise hypotheses for this case and modifications to the prescribed boundary conditions to~\S\ref{s:accumulation_withlayer}.}
 
 \begin{itemize}
\item  \textbf{Fast layers}:  As mentioned previously, in general there may be fast layers of $\phi_\eps$ entering and exiting from the slow region, on which an Airy transition and/or uniform oscillations may take place. The most relevant case for FitzHugh--Nagumo is that the solution enters the oscillatory region on the slow manifold via an Airy transition, and exits via a fast layer or a second Airy transition. Since a fast entry layer is not relevant for our consideration of the FitzHugh--Nagumo equation, and since a fast entry layer gives the same order of error terms as a fast exit layer, we only discuss the latter. In the case of a fast exit layer, the solution leaves the uniformly oscillatory region of the slow manifold by ``jumping off" along one of the hyperbolic fast directions, and the exit boundary condition is prescribed at some small, but positive, distance from the slow manifold. 
\end{itemize}




\subsection{Exponential dichotomies and trichotomies}\label{s:exptrich}
We briefly review the notion of exponential dichotomies and trichotomies. Consider the linear system
\begin{align}\label{e:expdichex}
U_\zeta = A(\zeta)U
\end{align}
and let $\Phi(\zeta,\tilde{\zeta})$ denote the corresponding evolution operator. 

\begin{definition}
The linear system~\eqref{e:expdichex} admits an \emph{exponential dichotomy} on an interval $I\subseteq \mathbb{R}$, with constants $C,\mu>0$, and projections $P^{\mathrm{u},\mathrm{s}}(\zeta), \zeta\in I$ if the following hold for all $\zeta,\tilde{\zeta}\in I$.
\begin{enumerate}[(i)]
\item $P^\mathrm{u}(\zeta)+P^\mathrm{s}(\zeta)=1$
\item $\Phi(\zeta,\tilde{\zeta})P^{\mathrm{u},\mathrm{s}}(\tilde{\zeta})=P^{\mathrm{u},\mathrm{s}}(\zeta)\Phi(\zeta,\tilde{\zeta})$
\item $|\Phi(\zeta,\tilde{\zeta})P^{\mathrm{s}}(\tilde{\zeta})|, |\Phi(\tilde{\zeta},\zeta)P^{\mathrm{u}}(\zeta)|\leq Ce^{-\mu(\zeta-\tilde{\zeta})},  \zeta\geq \tilde{\zeta}.$
\end{enumerate}
\end{definition}
It is occasionally convenient to define corresponding (un)stable evolution operators via $\Phi^{\mathrm{u},\mathrm{s}}(\zeta,\tilde{\zeta}):=\Phi(\zeta,\tilde{\zeta})P^{\mathrm{u},\mathrm{s}}(\tilde{\zeta})$.

A related notion is that of exponential trichotomies, which allow for a center subspace. 
\begin{definition}
The linear system~\eqref{e:expdichex} admits an exponential trichotomy on an interval $I\subseteq \mathbb{R}$, with constants $C>0$ and $\mu_1>\mu_2>0$ and projections $P^{\mathrm{uu},\mathrm{c}, \mathrm{ss}}(\zeta), \zeta \in I$, if the following hold for all $\zeta,\tilde{\zeta}\in I$
\begin{enumerate}[(i)]
\item $P^\mathrm{uu}(\zeta)+P^\mathrm{c}(\zeta)+P^\mathrm{ss}(\zeta)=1$
\item $\Phi(\zeta,\tilde{\zeta})P^{\mathrm{uu}, \mathrm{c}, \mathrm{ss}}(\tilde{\zeta})=P^{\mathrm{uu}, \mathrm{c}, \mathrm{ss}}(\zeta)\Phi(\zeta,\tilde{\zeta})$
\item $|\Phi(\zeta,\tilde{\zeta})P^{\mathrm{ss}}(\tilde{\zeta})|, |\Phi(\tilde{\zeta},\zeta)P^{\mathrm{uu}}(\zeta)|\leq Ce^{-\mu_1(\zeta-\tilde{\zeta})}, \zeta\geq \tilde{\zeta}$
\item $|\Phi(\zeta,\tilde{\zeta})P^{\mathrm{c}}(\tilde{\zeta})|\leq Ce^{\mu_2|\zeta-\tilde{\zeta}|}.$
\end{enumerate}
\end{definition}
We \blue{will sometimes use the corresponding evolution operators defined} via $\Phi^{\mathrm{uu},\mathrm{c},\mathrm{ss}}(\zeta,\tilde{\zeta}):=\Phi(\zeta,\tilde{\zeta})P^{\mathrm{uu},\mathrm{c},\mathrm{ss}}(\tilde{\zeta})$.

Exponential dichotomies and trichotomies are useful in the analysis of nonautonomous linear systems through the separation of the dynamics into spaces of exponentially growing or decaying solutions; furthermore, they are robust under perturbations of the underlying linear system~\eqref{e:expdichex}, a property referred to as roughness. It is known that exponential di-/trichotomies can be constructed for linear systems~\eqref{e:expdichex}, whose coefficients vary slowly, a property we will exploit below. For background on the theory of exponential di-/trichotomies and their properties, see~\cite{Cop67, Cop78, Pal82}.

This notion of exponential trichotomies in particular will be essential in solving the boundary value problem which arises for eigenvalue accumulation due to slow absolute spectrum, by showing that the dynamics restricted to an appropriate center subspace can be analyzed similarly to the prototype examples from~\S\ref{s:proto}.

\subsection{Exponential trichotomies along the slow region}\label{s:exptrichslow}

\begin{figure}
\begin{center}
\includegraphics{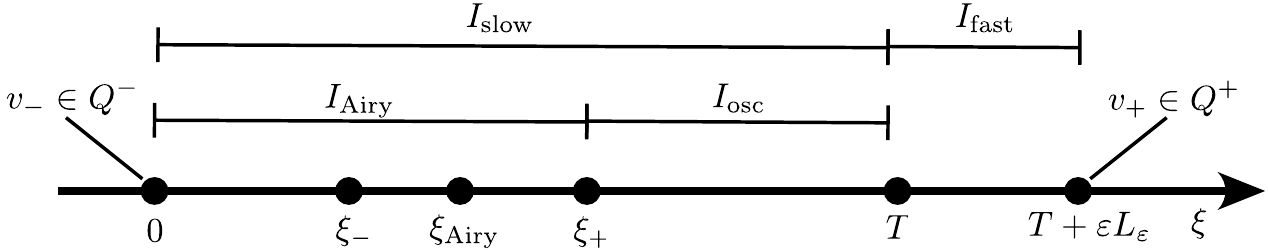}
\caption{\blue{Sketch of the decomposition of the interval on the $\xi$-scale as in \eqref{e:decomp} and the notation used. In the case without fast layers, the boundary condition $Q^+$ space is fixed at $\xi=T$. When including a fast exit layer, interval $I_\mathrm{fast}$ is appended to the end of $I_\mathrm{slow}$ and the boundary subspace $Q^+$ is instead fixed at $\xi=T+\eps L_\eps$; see~\S\ref{s:accumulation_withlayer}.}}
\label{f:intervals}
\end{center}
\end{figure}

\blue{In this section we construct exponential trichotomies for the system~\eqref{e:genslowlinear} on the interval $[0,T/\eps]=\eps^{-1}I_\mathrm{slow}$. We begin with the uniformly oscillating regime in~\S\ref{s:exptri_osc}, and then discuss modifications to this argument when including an Airy transition interval in~\S\ref{s:exptri_airy}.}

\subsubsection{Uniformly oscillating regime \blue{$I_\mathrm{slow}=I_\osc$}}\label{s:exptri_osc}
We first consider the dynamics in the uniformly oscillatory regime $\xi\in I_\osc $

\begin{hypothesis}\label{h:osc}
$\lambda_0\in \R$ is such that for all $\xi\in I_\osc$, cf.\ \eqref{e:decomp}, and $0\leq \eps\ll1$ the eigenvalues $\nu_i(\xi;\lambda_0,\eps)$ of $A_\slow(\xi;\lambda_0,\eps)$ satisfy, 
\begin{align}\label{e:spatialevals}
\Re \nu_{i_\infty-1} >\Re \nu_{i_\infty}= \Re \nu_{i_\infty+1}>\Re \nu_{i_\infty+2},\; 
\Im(\nu_{i_\infty}-\nu_{i_\infty+1}) \neq 0.
\end{align}
\end{hypothesis}


\begin{remark}\label{r:Iabs}
The setting of this hypothesis is the core for eigenvalue accumulation and relates to the slow absolute spectrum as follows. Since $\nu_{i_\infty} =\overline{ \nu_{i_\infty+1}}$ is robust against perturbing $\lambda$ along $\R$, the hypothesis in fact holds for $\lambda_0$ on an interval $\Iabs\subset \R$. Then $\Iabs\subset \Sigma_\abs^\slow(S_0)$ and each $\lambda_0\in \Iabs$ satisfies the property of the uniformly oscillating case. 
\end{remark}

\begin{remark}\label{r:weight}
If  \eqref{e:spatialevals} holds, then $\nu_{i_\infty}(\xi;\lambda,\eps) \in \rmi\R$ can always be arranged by introducing the $\xi$-dependent exponential weight $\Re\nu_{i_\infty}$ in \eqref{e:genslowlinear}, which shifts all spatial eigenvalues and does not change the relevant point spectrum. 
As mentioned, for the spectrum of a travelling wave on $\R$ only the absolute spectrum in $\Omega_1\subset \C$ is relevant. This stems from the fact that the exponential weight should not overcompensate exponential convergence in the tails of $\phi_\eps$ to the asymptotic state, cf.\ \cite[Def. 3.4]{SSabs}. 
%
\end{remark}

The oscillatory dynamics on $I_\osc$ can be made explicit in suitable coordinates, which will be the basis to solve for eigenvalues.

\begin{lemma}\label{l:diag}
Under Hypothesis~\ref{h:osc}, for $\lambda$ in a neighborhood of $\lambda_0$ in $\C$, equation \eqref{e:genslowlinear} admits exponential trichotomies for \red{$\zeta \in \eps^{-1} I_\osc$} relative to the exponential weight $\Re\nu_{i_\infty}$ \blue{with $\eps$-independent constants $C,\mu$,  stable/unstable/center trichotomy projections $P^{\ss,\uu,\cc}(\zeta;\lambda,\eps)=P^{\ss,\uu,\cc}_\mathrm{slow}(\eps\zeta;\lambda,\eps)$, subspaces $E^{\ss,\uu,\cc}(\zeta;\lambda,\eps)=E^{\ss,\uu,\cc}_\mathrm{slow}(\eps\zeta;\lambda,\eps)$, and evolutions $\Phi^{\ss,\uu,\cc}(\zeta,\tilde{\zeta};\lambda,\eps)=\Phi^{\ss,\uu,\cc}_\mathrm{slow}(\eps\zeta,\eps \tilde{\zeta};\lambda,\eps)$, satisfying the following}
\begin{enumerate}[(i)]
\item \blue{The projections are analytic in $\lambda$, and for $\xi\in I_\osc$, we have that
\begin{align}\label{e:proj}
\left| P^{\ss,\uu,\cc}_\mathrm{slow}(\xi;\lambda,\eps) - \mathcal{P}^{\ss,\uu,\cc}_\mathrm{slow}(\xi;\lambda,\eps)  \right| = \mathcal{O}(\eps) 
\end{align}
 where $\mathcal{P}^{\ss,\uu,\cc}_\mathrm{slow}(\xi;\lambda,\eps)$ denote the spectral projections onto the stable, unstable, center eigenspaces of $A_\mathrm{slow}(\xi;\lambda,\eps)$.
}

\item There exists a non-trivial compact interval $\Iabs\subset \R$ around $\lambda_0$ such that Hypothesis~\ref{h:osc} holds for all $\lambda\in \Iabs$, and a family of coordinate changes, continuously differentiable in $(\xi, \lambda) \in I_\osc\times \Iabs$, which transforms \eqref{e:genslowlinear} to the block-diagonal form
\begin{align}\label{e:genslowlinearDiag}
\begin{pmatrix}\dot{X}\\ \dot{Y}\end{pmatrix} = \begin{pmatrix}A_{11}(\eps \zeta;\lambda, \eps)&0\\ 0 & A_{22}(\eps \zeta;\lambda, \eps) \end{pmatrix} \begin{pmatrix}X\\ Y\end{pmatrix},
\end{align}
and satisfies
\begin{align}\label{e:oscMat}
A_{11}(\xi;\lambda, \eps) = \begin{pmatrix}0&\red{-\omega(\xi;\lambda)}
\\\red{\omega(\xi;\lambda)}& 0\end{pmatrix} +\calO(\eps).
\end{align}
\end{enumerate}
\end{lemma}
\begin{proof}
By assumption $A_\slow(\xi;\lambda,\eps)$ has a uniform spectral splitting for $0\leq \eps\ll1$ and $\xi\in I_\osc$.  
Together with slowly varying coefficients in $A_\slow(\eps \zeta;\lambda, \eps)$ and the exponential weight $\Re\nu_{i_\infty}$, this implies an exponential trichotomy on $\zeta \in \eps^{-1}I_\osc$ with uniform constants for $0\leq \eps \ll1$. An interval $\Iabs$ on which Hypothesis~\ref{h:osc} holds exists as mentioned in Remark \ref{r:Iabs}, and possibly shrinking it, we may choose a fixed $\Iabs$ for the following. The coordinate changes to block-diagonal form exist based on \cite[Lemma 1]{Pal82}. 
A normal form transformation of the center block yields the claimed form as follows. Let $A_{11}$ be the upper left block whose eigenvalues form a complex conjugate pair. For each $\xi,\lambda,\eps$ there is a change of variables matrix $S(\xi;\lambda,\eps)$ with 
\[
S(\xi;\lambda,\eps) A_{11}(\xi;\lambda,\eps) S^{-1}(\xi;\lambda,\eps) = 
\begin{pmatrix}0&\red{-}\omega(\xi;\lambda,\eps)\\
\red{\omega}(\xi;\lambda,\eps)& 0\end{pmatrix}
\]
Smoothness of $A_{11}$ makes $S$ and $\nu_{i_\infty}, \nu_{i_\infty+1}$ continuously differentiable in $\xi,\lambda,\eps$. Setting $\tX = S(\xi;\lambda,\eps) X$ we thus obtain,
\[
\tX_\zeta = S A_{11} S^{-1}\tX + \eps (\partial_\xi S) S^{-1} \tX = \left(\begin{pmatrix}0&\red{-\omega}(\eps\zeta;\lambda,\eps)\\\red{\omega}(\eps\zeta;\lambda,\eps)& 0\end{pmatrix} + \eps (\partial_\xi S) S^{-1}\right) \tX,
\]
and on the compact $\Iabs,I_\osc$ we have $(\partial_\xi S) S^{-1}$ bounded\blue{, and may choose $\omega(\xi;\lambda)=\omega(\xi;\lambda,0)$ in \eqref{e:oscMat}}.
%
\end{proof}


\subsubsection{Airy transition regime \blue{$I_\mathrm{slow}=I_\Airy\cup I_\mathrm{osc}$}}\label{s:exptri_airy}
The analogue in the Airy transition case is as follows.

\begin{hypothesis}\label{h:Airy}
$\lambda_0\in\R$ is such that for all $\xi\in I_\Airy$, cf.\ \eqref{e:decomp}, the eigenvalues $\nu_i(\xi;\lambda_0,\eps)$ of $A_\slow(\xi;\lambda_0,\eps)$ satisfy, at $\eps=0$, 
\begin{align}\label{e:spatialevalsAiry}
\Re \nu_{i_\infty-1} >\Re \nu_{i_\infty}\geq \Re \nu_{i_\infty+1}>\Re \nu_{i_\infty+2},
\end{align} 
with strict inequality for $\xi< \xi_\Airy$, a pinched double root at $\xi=\xi_\Airy$, which is a geometrically simple double eigenvalue, and \eqref{e:spatialevals} holds for $\xi>\xi_\Airy$. 
 \end{hypothesis}
  
Similar to Hypothesis~\ref{h:osc}, there is a non-trivial interval containing $\lambda_0$ for which the hypothesis holds. And we have the following analogous to Lemma~\ref{l:diag}.

\begin{lemma}\label{l:diagAiry}
Under Hypotheses~\ref{h:osc} and~\ref{h:Airy}, for $\lambda$ in a neighborhood of $\lambda_0$ in $\C$, equation \eqref{e:genslowlinear0} admits exponential trichotomies for \red{$\zeta \in \eps^{-1} I_\slow$} relative to the exponential weight $(\Re\nu_{i_\infty}+\Re\nu_{i_\infty+1})/2$ \blue{with $\eps$-independent constants $C,\mu$, stable/unstable/center trichotomy projections $P^{\ss,\uu,\cc}(\zeta;\lambda,\eps)=P^{\ss,\uu,\cc}_\mathrm{slow}(\eps\zeta;\lambda,\eps)$, subspaces $E^{\ss,\uu,\cc}(\zeta;\lambda,\eps)=E^{\ss,\uu,\cc}_\mathrm{slow}(\eps\zeta;\lambda,\eps)$, and evolutions $\Phi^{\ss,\uu,\cc}(\zeta,\tilde{\zeta};\lambda,\eps)=\Phi^{\ss,\uu,\cc}_\mathrm{slow}(\eps\zeta,\eps \tilde{\zeta};\lambda,\eps)$, satisfying the following}
\begin{enumerate}[(i)]
\item \blue{The projections are analytic in $\lambda$, and for $\xi\in I_\osc$, we have that
\begin{align}\label{e:proj-airy}
\left| P^{\ss,\uu,\cc}_\mathrm{slow}(\xi;\lambda,\eps) - \mathcal{P}^{\ss,\uu,\cc}_\mathrm{slow}(\xi;\lambda,\eps)  \right| = \mathcal{O}(\eps)
\end{align}
 where $\mathcal{P}^{\ss,\uu,\cc}_\mathrm{slow}(\xi;\lambda,\eps)$ denote the spectral projections onto the stable, unstable, center eigenspaces of $A_\mathrm{slow}(\xi;\lambda,\eps)$.
}
\item There exists a non-trivial compact interval $\Iabs\subset \R$ around $\lambda_0$ such that Hypothesis~\ref{h:osc} holds for all $\lambda\in \Iabs$, and a family of coordinate changes, continuously differentiable in $(\xi, \lambda) \in I_\slow\times \Iabs$, 
which transforms \eqref{e:genslowlinear} to the block-diagonal form
\begin{align}\label{e:genslowlinearDiag-old}
\begin{pmatrix}\dot{X}\\ \dot{Y}\end{pmatrix} = \begin{pmatrix}B_{11}(\eps \zeta;\lambda, \eps)&0\\ 0 & B_{22}(\eps \zeta;\lambda, \eps) \end{pmatrix} \begin{pmatrix}X\\ Y\end{pmatrix},
\end{align}
and \blue{for all $\xi_\pm\approx \xi_\Airy$} with $\xi_- < \xi_\Airy < \xi_+$ \blue{there is $b_0>0$ such that for any} $\xi\in [\xi_-,\xi_+]$ we have
\begin{align*}
B_{11}(\xi;\lambda, \eps) = \begin{pmatrix}0&-b(\xi\red{;\lambda})\\1& 0 \end{pmatrix} +\calO(\eps),
\quad b(\xi\red{;\lambda}) &= b_0(\xi - \xi_\Airy\red{(\lambda)} + \calO((\xi-\xi_\Airy\red{(\lambda)})^2).
\end{align*}
For $\xi\geq \xi_+$ we have $\blue{T B_{11} T^{-1}}= A_{11}$ from \eqref{e:oscMat}, \blue{where $T:=\frac 1{\sqrt{2}}\begin{pmatrix}1 & \omega\\ \red{-}1 & \omega\end{pmatrix}$, $\omega:=\sqrt{b}$},
 and for $\xi\leq \blue{\xi_-}$ we have \blue{$T B_{11} T^{-1} = \mathrm{diag}(-\omega,\omega)$ with $\omega:=\sqrt{-b}$ and $\pm\omega$ being} $\eps$-close to $\nu_{i_\infty}$, $\nu_{i_\infty+1}$, respectively.
\end{enumerate}
\end{lemma}
\begin{proof}
Up to the different normal form in the center part, the proof is the same as that of Lemma~\ref{l:diag}. At the Airy point, by Hypothesis~\ref{h:Airy}, the block on the center eigenspace has Jordan normal form $\begin{pmatrix}0&0\\1&0\end{pmatrix}$ so that we may choose coordinates where $B_{11}=\begin{pmatrix}b_{11}&b_{12}\\b_{21}&b_{22}\end{pmatrix}$ has $b_{21}\neq0$ for $\xi\in [\xi_-,\xi_+]$ \blue{and all $\xi_\pm\approx \xi_\Airy$,} $\xi_- < \xi_\Airy < \xi_+$. For $\xi\in [\xi_-,\xi_+]$, the pointwise $\xi$-dependent coordinate change given by 
$S=\begin{pmatrix}1&-b_{22}\\0&b_{21}\end{pmatrix}$ yields 
\begin{align}
S^{-1}B_{11}S = \begin{pmatrix}\mathrm{tr}(B_{11})&-\det(B_{11})\\1&0\end{pmatrix}
\end{align} in terms of trace and determinant. Let $W$ denote these coordinates and $X$ those of Lemma~\ref{l:diag}. At $\xi=\xi_+$ we can change coordinates to the form \eqref{e:oscMat} by setting 
\[
X_1 = (\sqrt{b}W_2\red{+}W_1)/\sqrt{2},\; X_2 = (\sqrt{b}W_2\red{-}W_1)/\sqrt{2}
\]
and analogously at $\xi=\xi_-$. The error terms now follow as in the proof of Lemma~\ref{l:diag} using $\det(B_{11})=0$ at $\xi=\xi_\Airy$ and $\mathrm{tr}(B_{11})\equiv 0$ due to the a priori chosen exponential weight; the dichotomy for $\xi\leq \xi_-$ allows a complete diagonalisation. 
\end{proof}


\section{Eigenvalue accumulation along slow manifolds}\label{s:gentheory}

\blue{In this section, we use the exponential trichotomies constructed in the previous section to describe the accumulation of eigenvalues along the slow absolute spectrum for the  cases described in \S\ref{s:cases}. We begin in~\S\ref{s:accumulation_nolayer} with the base case of accumulation \emph{without} fast layers (in both the uniformly oscillating and Airy transition regimes), and then consider the necessary modifications to treat fast layers in~\S\ref{s:accumulation_withlayer}.}

\subsection{Eigenvalue accumulation without fast layers}\label{s:accumulation_nolayer}

\blue{Recall that, as discussed  in \S\ref{s:cases}, in this case the interval $[0,T/\eps]$ on the slow time scale can be written as
\begin{equation}\label{e:decompnew}
[0,T] = I_\slow = I_\Airy \cup I_\osc,
\end{equation}
 In the `Airy case', $I_\Airy$ is a non-trivial compact interval that contains an Airy point in its interior, and $I_\osc$ is a non-trivial compact interval with uniform oscillations, whereas $I_\Airy$ is taken to be empty in the uniformly oscillating case.

With $\Phi(\zeta,\tilde{\zeta};\lambda,\eps)$ the evolution operator of~\eqref{e:genslowlinear} on $[0,T/\eps]=\eps^{-1}I_\slow$, we consider the existence problem for an eigenfunction on the right boundary of the slow region, i.e., 
\begin{align}\label{e:condefct}
\Phi(T/\eps,0;\lambda,\eps)Q^-(\lambda,\eps)\cap  Q^+(\lambda,\eps) \neq \{0\},
\end{align}
under the assumptions on this boundary value problem as in~\S\ref{s:slowabssetup}.
According to \eqref{e:condefct} we define the following point spectrum and refer to \cite{SSabs} for details on the relation to spectrum in case of travelling waves. 

\begin{definition}
We say $\lambda\in \C$ is an eigenvalue for \eqref{e:genslowlinear} if and only if a \blue{nontrivial} solution of~\eqref{e:genslowlinear}, referred to as eigenfunction, exists such that $U(0)\in Q^-(\lambda,\eps)$ and $U(T/\eps)\in Q^+(\lambda,\eps)$. We refer to the set of eigenvalues $\specpt\subset\C$ as the point spectrum.
\end{definition}}

In contrast with the prototype example of \S\ref{s:proto}, here we may additionally have hyperbolic directions, which requires a non-degeneracy of the boundary spaces $Q^\pm(\lambda,\eps)$ at the boundaries $\xi=0,T$.
\red{
\begin{hypothesis}\label{h:osc_vector}
There exists $C>0$ such that for all $\eps\in(0,\eps_0]$, the following hold. 
\begin{enumerate}[(i)]
\item There exists a unit vector $v_-=v_-(\lambda,\eps)\in E_\slow^\cc(0;\lambda,\eps)\oplus E_\slow^{\ss}(0;\lambda,\eps)$ such that $Q^-(\lambda,\eps) = \mathrm{span}\{v_-\}\oplus\tilde{Q}^-$ where
\begin{align}\label{e:layerBndryLeft}
\tilde{Q}^-\oplus E_\slow^\cc(0;\lambda,\eps)\oplus E_\slow^{\ss}(0;\lambda,\eps)=\mathbb{C}^n,
\end{align}
and we can decompose any vector in $\mathrm{span}\{v_-\}$ as $v_-^\cc+v_-^\ss$ for some $v_-^\cc\in E^\cc(0;\lambda,\eps)$ and  $v_-^\ss\in E^{\ss}(0;\lambda,\eps)$ with $|v^\ss_-|\leq C|v^\cc_-|$. In the Airy case, we further assume 
that  $|v_-|\leq C|\mathcal{P}_{i_\infty}v_-|$, where $\mathcal{P}_{i_\infty}$ is the spectral projection of $A_\mathrm{slow}(0;\lambda,0)$ onto the eigenspace corresponding to the eigenvalue $\nu_{i_\infty}$.

\item We similarly assume there exists a unit vector $v_+=v_+(\lambda,\eps)\in E_\slow^\cc(T;\lambda,\eps)\oplus E_\slow^{\uu}(T;\lambda,\eps)$ such that $Q^+(\lambda,\eps) = \mathrm{span}\{v_+\}\oplus\tilde{Q}^+$ where
\begin{align}\label{e:layerBndryRight}
\tilde{Q}^+\oplus E_\slow^\cc(T;\lambda,\eps)\oplus E_\slow^{\uu}(T;\lambda,\eps)=\mathbb{C}^n,
\end{align}
and we can decompose any vector in $\mathrm{span}\{v_+\}$ as $v^\cc_++v^\uu_+$ for some $v_+^\cc\in E_\slow^\cc(T;\lambda,\eps)$ and  $v_+^\uu\in E_\slow^{\uu}(T;\lambda,\eps)$ with $|v^\uu_+|\leq C|v^\cc_+|$.
\end{enumerate}
Furthermore, we can choose $v_\pm$ with these properties such that $v_\pm$ is continuous for $\eps\in(0,\eps_0]$, and $\partial_\lambda v_\pm$ is continuous for $\eps\in(0,\eps_0]$, satisfying $|\partial_\lambda v_\pm| = o(\eps^{-1})$. 
\end{hypothesis}

\begin{remark}
The conditions in Hypothesis~\ref{h:osc_vector} regarding the decompositions $v_-^\cc+v_-^\ss$ and $v^\cc_++v^\uu_+$ guarantee that the projections of the boundary subspaces $Q^\pm$ onto the center trichotomy spaces at the endpoints of the slow interval are nontrivial as $\eps \to0$. In the Airy case, we further require that the projection onto the eigenspace corresponding to the eigenvalue $\nu_{i_\infty}$ is nontrivial in this limit.
\end{remark}}



In the case without and with layer, for a given $\lambda_0$, the following quantity will be relevant, analogous to \S\ref{s:proto}.
\begin{equation}\label{e:delpr}
\Delta':=\left.\frac{\rmd}{\rmd \lambda}\right|_{\lambda=\lambda_0} \int_0^T \omega(\xi;\lambda)\rmd\xi.
\end{equation}

\subsubsection{Uniformly oscillating case.} 

\begin{proposition}\label{p:uni}
Assume $[0,T]=I_\osc$ and Hypotheses~\ref{h:well}, \ref{h:osc},~\ref{h:osc_vector} hold for some $\lambda=\lambda_0\in\R$ and fix $r>0$. If $\Delta'\neq 0$, cf.\ \eqref{e:delpr}, then for all sufficiently small $\eps>0$, $\specpt\cap B_r(\lambda_0)\subset\C$ has cardinality $\calO(1/\eps)$, where $B_r(\lambda_0)$ is the disc around $\lambda_0$ with radius $r$.
\end{proposition}

\begin{proof}
By robustness of linear independence, Hypothesis~\ref{h:osc_vector} holds on an interval containing $\lambda_0$, where also Hypothesis~\ref{h:osc} holds without loss of generality. Consider \red{$v_\pm^\cc\neq 0$}  from Hypothesis~\ref{h:osc_vector}. 

\blue{The trichotomy projections give $\Phi(T/\eps,0;\lambda,\eps)Q^-$ as a direct sum of the evolved center, stable and unstable parts. By the linear independence in Hypothesis~\ref{h:osc_vector}, the projection of the left hand side of \eqref{e:condefct} onto the stable and unstable parts 
gives a curve of unique solutions for $0<\eps\ll 1$ by persistence of linear independence, which allows us to reduce the problem to solving the projection onto the center part. By possibly scaling $v_+$, to leading order this can be written as 
\begin{align}\label{e:unicenter}
 \Phi^\cc(T/\eps,0;\lambda,\eps) v_-^\cc =  v_+^\cc + \mathrm{h.o.t}.
\end{align}
We omit the details in deriving the leading order equation~\eqref{e:unicenter} for the center dynamics, as these are analogous to (and simpler than) in the proof of Theorem~\ref{t:layers} below.}

Starting from the coordinates of Lemma~\ref{l:diag} we change to polar coordinates $X=(X_1,X_2) = (r\cos\theta,r\sin\theta)$, which results in the system on $\zeta \in \eps^{-1}I_\osc$ given by
\begin{align}\begin{split}\label{e:slowpolar}
\dot{r}&=\calO(\eps r)\\
\dot{\theta}& =-\omega(\eps \zeta;\lambda)+ \calO(\eps).
\end{split}
\end{align}
In the corresponding form of \eqref{e:unicenter} the radial equation is trivially solved, but is non-vanishing since $v_\pm^\cc\neq 0$. 
Solving the remaining  angular equation $\theta$ is completely analogous to \eqref{e:protocond} with interval length $T/\eps$, i.e., 
\[
\Delta_T(\lambda,\eps) = \eps(\theta(0)-\theta(T/\eps)) =\eps\int_{0}^{T/\epsilon}\omega(\eps \zeta;\lambda)d\zeta + \calO(\eps) = \int_0^T \omega(\xi;\lambda)\rmd\xi + \calO(\eps),
\]
and existence of an eigenvalue is equivalent to
\blue{
\begin{align}
\Delta_T(\lambda,\eps) = \eps (\theta_--\theta_+)\mod \pi,
\end{align}}
where $\theta_\pm$ denote the angular coordinates corresponding to $v^\cc_\pm$, compare~\eqref{e:protocond0}. 
\blue{We next show that $\Delta'\neq0$ gives $\partial_\lambda\Delta_T(\lambda,0)\neq 0$, i.e., the $\calO(\eps)$ remainder term in $\Delta_T$ is not relevant for this derivative. This implies the claimed cardinality of $\calO(1/\eps)$ solutions near $\lambda_0$, cf.\ \eqref{e:protocond0}. 

We write the remainder term for $\dot\theta$ in \eqref{e:slowpolar} as $\eps a(\xi,\theta;\lambda,\eps)$, where $a$ and its derivatives \red{with respect to $\xi, \theta, \lambda$} are bounded, in particular $a_\theta(\zeta):=\partial_\theta a(\eps\zeta,\theta(\zeta);\lambda,\eps)$ and $a_\lambda(\zeta):=\partial_\lambda a(\eps\zeta,\theta(\zeta);\lambda,\eps)$. The derivative $\theta_\lambda:=\partial_\lambda \theta$ solves the variational equation, whose solution can be written as
\[
\theta_\lambda(T/\eps) = \exp\left(\eps \int_0^{T/\eps} a_\theta(\zeta)\rmd\zeta\right)\theta_\lambda(0) + 
\eps^{-1}\int_0^{T}\exp\left(\eps \int_{t/\eps}^{T/\eps} a_\theta(\zeta)\rmd\zeta\right) \big(\partial_\lambda\omega(t) + \eps a_\lambda(t/\eps)\big) \rmd t.
\]
Let $\tilde\theta(\zeta)\in\R$ denote the invertible lift of $\theta(\zeta)\in [0,2\pi)\!\!\mod 2\pi$ and note that with $\psi=\tilde\theta(\zeta)$ we have
\[
\int_{t/\eps}^{T/\eps} a_\theta(\zeta)\rmd\zeta 
= \int_{\tilde\theta(t/\eps)}^{\tilde\theta(T/\eps)} \beta(\psi,\zeta,\eps)\rmd\psi\,,\quad 
\beta(\psi,\zeta,\eps):=\frac{a_\theta}{\dot{\theta}}(\psi,\zeta,\eps)= \sigma \eps^{-1}\frac{\rmd}{\rmd \psi}\ln\big|\eps a(\eps\zeta,\psi;\lambda,\eps)-\omega(\eps\zeta;\lambda)\big|,
\]
for suitable $\sigma = \pm 1$. With $k(\eps):=\lceil \tilde\theta(t/\eps)/(2 \pi)\rceil, K(\eps):=\lfloor \tilde\theta(T/\eps)/(2 \pi)\rfloor\in\N$ we have $K(\eps)-k(\eps) = \calO(1/\eps)$ and 
\[
\int_{\tilde\theta(t/\eps)}^{\tilde\theta(T/\eps)} \beta(\psi,\zeta,\eps)\rmd\psi = 
\int_{\tilde\theta(t/\eps)}^{2\pi k(\eps)} \beta(\psi,\zeta,\eps)\rmd\psi
+ \sum_{j=k(\eps)}^{K(\eps)-1} \int_{2\pi j}^{2\pi(j+1)} \beta(\psi,\zeta,\eps)\rmd\psi
+ \int_{2\pi K(\eps)}^{\tilde\theta(T/\eps)} \beta(\psi,\zeta,\eps)\rmd\psi,
\]
where the first and last integrals are bounded uniformly in $\eps$. We claim that each summand of the second term is $\calO(\eps)$ so that the sum over $\calO(1/\eps)$ of these is bounded as well. To show this, note that $\beta$ is the derivative of a $2\pi$-periodic function with respect to $\theta$ and therefore, for each fixed $\eps$ and $\zeta$,  has zero average in terms of $\theta$.  The claim now follows from expanding $\beta(\psi,\zeta,\eps) = \beta(\psi,2\pi j,0)+\calO(\eps)$ in each summand. 

Hence, $\eps\int_{t/\eps}^{T/\eps} a_\theta(\zeta)\rmd\zeta\to 0$ as $\eps\to 0$ uniformly in $t$. Together with $\eps\theta_\lambda(0)= o(1)$ from Hypothesis~\ref{h:osc_vector} we have 
\[
\lim_{\eps\to 0} \eps \theta_\lambda(T/\eps) =  \int_0^{T}\partial_\lambda\omega(t) \rmd t,
\]
which equals $\Delta'$ at $\lambda=\lambda_0$.
}
\end{proof}

We note that in the degenerate case $v_\pm^c=0$ the boundary conditions are contained in the hyperbolic parts of the trichotomies and thus no additional eigenvalues are expected. Indeed, in this case the radial variable vanishes so the angular equation is meaningless. 

\begin{remark}
Let us consider the condition $\Delta'\neq0$ in the FHN case, \blue{cf.\ \S\ref{sec:pulse_geometry}, }where the relevant \blue{range for $u_0\in (u^-_{\mathrm{A},0},u^+_{\mathrm{A},0})$} is $0<\lambda < f'(u_0)-\frac{c^2}{4}$, and the frequencies are
\[
\omega(u_0;\lambda) = \sqrt{\frac{c^2}{4}-f'(u_0)+\lambda},
\]
\blue{compare \S\ref{sec:slowabs_middleslowmanifold} below}, so that 
\[
\partial_\lambda \omega(u_0;\lambda) = \frac 1 2 \omega(u_0;\lambda)^{-1}.
\]
Since $\lambda$ is real this has fixed sign and the integral condition $\partial_\lambda\Delta_T(\lambda,0)\neq0$ is indeed satisfied.
\end{remark}

\subsubsection{Airy transition case}

We now turn to the situation in which the solution enters the uniformly oscillating region via an Airy transition along the slow manifold. 

\begin{proposition}\label{p:Airy}
Assume $[0,T]=I_\Airy\cup I_\osc$ and Hypotheses~\ref{h:well}--\ref{h:osc_vector} hold for some $\lambda=\lambda_0\in\R$, and fix $r>0$. If $\Delta'\neq 0$, cf.\ \eqref{e:delpr}, then for all sufficiently small $\eps>0$, $\specpt\cap B_r(\lambda_0)\subset\C$ has cardinality $\calO(1/\eps)$, where $B_r(\lambda_0)$ is the disc around $\lambda_0$ with radius $r$.
\end{proposition}

\begin{proof}
\blue{
As in the proof of Proposition~\ref{p:uni}, now using Lemma~\ref{l:diagAiry}, the relevant component of the boundary condition is $v_-^\cc\in E_\slow^\cc(0;\lambda,\eps)$, which by Hypothesis~\ref{h:osc_vector} satisfies $\mathcal{P}_{i_\infty}v_-^\cc\neq0$. Employing the full diagonalisation \red{of $B_{11}$ in \eqref{e:genslowlinearDiag-old}} on the interval $\xi\in [0,\xi_-]$, or equivalently $\zeta \in[0,\eps^{-1}\xi_-]$, $\Phi(\eps^{-1}\xi_-,0)v_-^\cc$ spans a subspace which is $\mathcal{O}(\eps)$-close to the eigenspace corresponding to the eigenvalue $\nu_{i_\infty}$ of $A_\mathrm{slow}(\xi_-;\lambda,0)$.



\blue{In the coordinates on the interval $\xi\in [\xi_-,\xi_+]$ derived in the proof of Lemma~\ref{l:diagAiry}} the eigenspace corresponding to the \red{unstable} eigenvalue $\nu_{i_\infty}$ of $A_\mathrm{slow}(\xi_-;\lambda,0)$ corresponds to a fixed angle $\theta_{0,-}(\lambda)$, and we have that the angle corresponding to $\Phi^\cc(\zeta,0)v_-^\cc$ converges exponentially fast to $\theta_{0,-} + \calO(\eps)$ as $\zeta\to\eps^{-1}\xi_-$. Therefore at the entry to the Airy transition region $[\xi_-,\xi_+]$, the initial angle of $\Phi^\cc(\eps^{-1}\xi_-,0)v_-^\cc$ is given by $\theta_-=\theta_-(\lambda,\eps) = \theta_{0,-}(\lambda) + \calO(\eps)$. In these polar coordinates the equations for $X=(X_1,X_2)$ read 
\begin{align}\begin{split}\label{e:slowpolarAiry}
\dot{r}&=2(1-b)r\cos \theta \sin \theta+\calO(\eps r),\\
\dot{\theta}& = 1+(b-1)\sin^2\theta+ \calO(\eps)
\end{split}
\end{align}
and for $0<\eps\ll |\xi_+-\xi_-|$ we thus estimate
\begin{align}
 \theta_+(\lambda,\eps):=\theta(\eps^{-1} \xi_+) = \theta_-(\lambda,\eps) + \calO(\eps^{-1}|\xi_+-\xi_-|).
 \end{align}
  Hence, the scaled phase jump across the Airy region
  \begin{align}
   \Delta_\mathrm{Airy}(\lambda):=\eps(\theta_+-\theta_-)=\calO(|\xi_+-\xi_-|)
   \end{align} 
    from $\xi=\xi_-$ to $\xi=\xi_+$ is uniformly bounded. By taking $|\xi_+-\xi_-|$ sufficiently small, this phase jump is dominated by the phase jump $\Delta$ over the uniformly oscillating region. Therefore, provided the derivative $ \Delta_\mathrm{Airy}'(\lambda_0)$ \red{with respect to $\lambda$} is similarly dominated by the quantity $\Delta'(\lambda_0)$, the proof then proceeds as that of Proposition~\ref{p:uni} with scaled phase jump $\Delta$ under uniform oscillation with $I_\osc=[\xi_+,T]$, compare \S\ref{s:proto}. 

It remains to estimate the derivative $\Delta_\mathrm{Airy}'(\lambda_0)$ and show that $\Delta_\mathrm{Airy}'(\lambda_0)=o(\eps^{-1})$. Motivated by the blow-up rescaling analysis near Airy points as in~\cite[\S5]{CSbanana}, we split the Airy transition region $\zeta \in [\eps^{-1}\xi_-, \eps^{-1}\xi_+]=I_1\cup I_2\cup I_3$ into three regions
\begin{align*}
I_1&=[\eps^{-1}\xi_-, \eps^{-1}\xi_2^-], \quad I_2 = [\eps^{-1}\xi_2^-, \eps^{-1}\xi_2^+], \quad I_3 =  [\eps^{-1}\xi_2^+, \eps^{-1}\xi_+],
\end{align*}
where $\xi_2^\pm:=\xi_\mathrm{Airy}\pm K\eps^{2/3}$, and $K\gg1$ is an $(\eps,\lambda)$-independent constant (to be chosen). We then separately estimate the angular variational equation over each (sub)interval $I_j$. On each subinterval, we define a suitable angular coordinate $\theta_j(\zeta)$, which then evolves according to the equation
\begin{align}
\dot{\theta}_j = f_j(\theta_j, \zeta;\lambda,\eps), \qquad j=1,2,3,
\end{align}
where each function $f_j$ is $2\pi$-periodic in $\theta_j$. From this, we then determine the behavior of the associated variational equation for the derivative $\partial_\lambda \theta_j$, which can be written in the general form
\begin{align}\begin{split}\label{e:slowpolarAiryjvar}
\partial_\lambda \dot{\theta}_j& = a_j(\zeta)\partial_\lambda \theta_j+ r_j(\zeta),
\end{split}
\end{align}
where $a_j(\zeta):=\partial_{\theta_{\! j}}f_j(\theta_j(\zeta), \zeta;\lambda,\eps)$ and $r_j(\zeta):=\partial_\lambda f_j(\theta_j(\zeta), \zeta;\lambda,\eps)$. The solution of this variational equation can be estimated over each interval via the variation of constants formula
\begin{align}\begin{split}\label{e:slowpolarAiryjvarconsts}
\partial_\lambda \theta_j(\zeta)& = \exp \left(\int_{\zeta_j^\mathrm{in}}^\zeta a_j(s)\mathrm{d}s \right)\partial_\lambda \theta_j(\zeta_j^\mathrm{in})+ \int_{\zeta_j^\mathrm{in}}^\zeta \exp \left(\int_{t}^\zeta a_j(s)\mathrm{d}s \right) r_j(t)\mathrm{d}t,
\end{split}
\end{align}
where $\zeta_j^\mathrm{in}$ is taken to be the left endpoint of the interval $I_j$ for $j=1,2,3$. This allows us to estimate $\partial_\lambda \theta_j$ across each (sub)interval, provided we account for the $\lambda$-dependence in the change of angular coordinate at the endpoints of neighboring subintervals. We proceed by considering each interval in turn.

\textbf{Region $I_1$}: We begin with the region $I_1$. We recall that 
\begin{align}
b(\xi;\lambda) = b_0(\xi-\xi_\Airy(\lambda))+\mathcal{O}\left ( (\xi-\xi_\Airy(\lambda))^2\right)
\end{align}
where the $\lambda$-dependence of the location of the Airy point is captured within $\xi_\mathrm{\Airy}$, noting that $\xi_\Airy'(\lambda)$ is bounded due the smooth dependence of the location of the Airy point on $\lambda$ for $\epsilon=0$. Note that this is due to the fact that the Airy point is isolated, which follows from Hypothesis~\ref{h:Airy}. We therefore have that 
\begin{align}
|\partial_\lambda b(\xi;\lambda)| \leq C
\end{align}
for some $C>0$ uniformly on the interval $\xi \in[\xi_-, \xi_+]$ and $(\lambda,\eps)$ near $(\lambda_0,0)$. In the region $I_1$, we note that 
\begin{align}
b(\xi;\lambda) \in \left[b_0(\xi_--\xi_\Airy)+\mathcal{O}\left ( |\xi_\Airy-\xi_-|^2\right), -b_0K\eps^{2/3}+\mathcal{O}(\eps^{4/3})\right].
\end{align}
So that, in particular, $b$ is strictly negative and bounded away from $0$. In this region, we need to estimate the variational equation for the angular coordinate
\begin{align}\begin{split}\label{e:slowpolarAiry1}
\dot{\theta_1}& = f_1(\theta_1,\zeta;\lambda,\eps) = 1+(b-1)\sin^2\theta_1+ \calO(\eps).
\end{split}
\end{align}
The quantity $b(\eps \zeta;\lambda)$ evolves on the slow timescale and in the region where $\xi-\xi_\mathrm{Airy}\leq -\delta_0<0$ is negative and bounded away from zero uniformly in $\eps$, solutions converge exponentially to $\theta_1^*(\zeta)+\mathcal{O}(\eps)$, where 
\begin{align}
\theta_1^*(\zeta) &=\frac{\pi}{2}-\sqrt{-b}+\mathcal{O}\left( b^{3/2}\right),
\end{align}
satisfies $\sin^2\theta_1^* = \frac{1}{1-b}$. However, to determine how such solutions enter the region $I_2$, it is necessary to track these solutions up to the $\eps$-dependent endpoint $\xi=\xi_2^- = \xi_\mathrm{Airy}-K\eps^{2/3}$.

Within~\eqref{e:slowpolarAiry1}, setting $\theta_1(\zeta) = \theta_1^*(\zeta)+\tilde{\theta}_1(\zeta)$, we have that $\tilde{\theta}_1$ satisfies
\begin{align}\begin{split}\label{e:slowpolarAiry1tilde}
\dot{\tilde{\theta}}_1& =  -2\sqrt{-b}\tilde{\theta}_1+ \calO\left(\frac{\eps}{\sqrt{-b}}+ \tilde{\theta}_1^2\right)\\
\tilde{\theta}_1(\eps^{-1}\xi_-) &= \mathcal{O}\left(\eps\right)
\end{split}
\end{align}
A standard fixed point argument on the space of uniformly bounded continuous  functions on the interval $\zeta\in I_1$, for $K>0$ fixed sufficiently large independent of $\eps$, then shows that the solution of this initial value problem satisfies $|\tilde{\theta}_1(\zeta)| = \mathcal{O}(\eps^{1/3})$ on the interval $I_1$, so that we can write the solution $\theta_1(\zeta)$ as
\begin{align}\label{e:Airy1est}
\theta_1(\zeta) &=\theta_1^*(\zeta)+\mathcal{O}(\eps^{1/3}),
\end{align}
on $I_1$. Using this, we now turn to the variational equation for $\partial_\lambda \theta_1$, which is given by~\eqref{e:slowpolarAiryjvar} with $j=1$, where 
\begin{align*}
a_1(\zeta)&=2(b(\eps\zeta;\lambda)-1)\sin \theta_1\cos \theta_1+\mathcal{O}(\eps),\\
r_1(\zeta) &= \calO(\partial_\lambda b)=\calO(1).
\end{align*}
 Solving this system using variation of constants, we obtain~\eqref{e:slowpolarAiryjvarconsts} with $\zeta_1^\mathrm{in} = \eps^{-1}\xi_-$ and $\partial_\lambda \theta_1(\zeta_1^\mathrm{in})=\partial_\lambda \theta_-$. Using the expression~\eqref{e:Airy1est} for $\theta_1(\zeta)$, we can estimate the quantity
\begin{align}
\int_{t}^\zeta a_1(s)\mathrm{d}s = -\int_{t}^\zeta(2+\mathcal{O}(|\xi_\Airy-\xi_-|))\sqrt{\xi_\mathrm{Airy}-\eps s}+\mathcal{O}(\eps^{1/3})\mathrm{d}s
\end{align}
from which \eqref{e:slowpolarAiryjvarconsts} with $j=1$ gives 
the estimate
\begin{align}\begin{split}\label{e:Airy1thetaest}
  \partial_\lambda \theta_1(\eps^{-1}\xi_2^-)& = \mathcal{O}\left(e^{-\eta/\eps}|\partial_\lambda \theta_-|+ \eps^{-1/3}\right) ,
\end{split}
\end{align}
provided $K>0$ is fixed sufficiently large.

\textbf{Region $I_2$}: In the region $I_2$, we change variables $x=\eps^{1/3} \tilde{x}$ and transform to polar coordinates with respect to the Cartesian coordinates $(\tilde{x},y)$, which results in the angular equation
\begin{align}\begin{split}\label{e:slowpolarAiry2}
\dot{\theta_2}& = f_2(\theta_2, \zeta;\lambda,\eps)= \eps^{1/3}+\left(\frac{b}{\eps^{1/3}}-\eps^{1/3}\right)\sin^2\theta_2+ \calO(\eps^{2/3}),
\end{split}
\end{align}
where we note that on the interval $I_2$, we have that $b(\eps \zeta;\lambda)=\mathcal{O}(\eps^{2/3})$. We write the variational equation for $\theta_2$, obtaining~\eqref{e:slowpolarAiryjvar} with $j=2$ and 
\begin{align}
\begin{split}\label{e:slowAiry2est}
a_2(\zeta)&=\left(2\left(\frac{b}{\eps^{1/3}}-\eps^{1/3}\right)\sin \theta_2 \cos \theta_2+ \calO(\eps^{2/3})\right)= \mathcal{O}(\eps^{1/3}),\\
 r_2(\zeta) &= \calO(\eps^{-1/3}\partial_\lambda b )=\calO(\eps^{-1/3}).
 \end{split}
\end{align}
Similarly to the region $I_1$, we solve this system using variation of constants to obtain~\eqref{e:slowpolarAiryjvarconsts} with $j=2$ and $\zeta_2^\mathrm{in} = \eps^{-1}\xi_2^- = \eps^{-1}\left(\xi_\mathrm{Airy}-K\eps^{2/3}\right)$. Using the estimates~\eqref{e:slowAiry2est}, this gives
\begin{align}\begin{split}
\partial_\lambda \theta_2(\eps^{-1}\xi_2^+)& = \mathcal{O}\left(|\partial_\lambda \theta_2(\eps^{-1}\xi_2^-)|+ \eps^{-2/3}\right).
\end{split}
\end{align}
Using the relation $e^{i\theta_2} = \left(\cos \theta_1+i\eps^{1/3} \sin\theta_1\right)\left(\cos^2\theta_1+\eps^{2/3}\sin^2\theta_1\right)^{-1/2}$, we find that 
\begin{align}
\begin{split}
\partial_\lambda \theta_2(\eps^{-1}\xi_2^-) &=  \left. \frac{\mathrm{d}\theta_2}{\mathrm{d} \theta_1}\right|_{\theta_1 = \theta_1(\eps^{-1}\xi_2^-)}\partial_\lambda\theta_1(\eps^{-1}\xi_2^-) 
= \mathcal{O}\left(e^{-\eta/\eps}|\partial_\lambda \theta_-|+ \eps^{-2/3}\right),
\end{split}
\end{align}
where we used~\eqref{e:Airy1thetaest}, possibly taking $\eta$ slightly smaller. From this we obtain
\begin{align}\begin{split}
\partial_\lambda \theta_2(\eps^{-1}\xi_2^+)& = \mathcal{O}\left(e^{-\eta/\eps}|\partial_\lambda \theta_-| + \eps^{-2/3}\right).
\end{split}
\end{align}

\textbf{Region $I_3$:} It remains to estimate over the region $I_3$. We change variables $x=\sqrt{b(\eps \zeta; \lambda)} \bar{x}$ and
transform to polar coordinates with respect to the coordinates $(\bar{x},y)$, which results in the angular equation
\begin{align}\begin{split}\label{e:slowpolarAiry3}
\dot{\theta_3}& = f_3(\theta_3,\zeta;\lambda,\eps)=\sqrt{b}+\frac{ \dot{b}}{2b}\sin\theta_3 \cos\theta_3+ \calO\left(\frac{\eps}{\sqrt{b}}\right).
\end{split}
\end{align}
We first estimate $\theta_3(\zeta)$ satisfying $\theta_3(\eps^{-1}\xi_2^+)=\theta_{3,0}$, the transformed value of $\theta_2(\eps^{-1}\xi_2^+)$, by setting $\theta_3=\theta_3^*(\zeta)+\tilde{\theta}_3$, where 
\begin{align}
\theta_3^*(\zeta) &= \theta_{3,0}+\int_{\eps^{-1}\xi_2^-}^\zeta \sqrt{b(\eps s;\lambda)} \mathrm{d}s\\
&=\theta_{3,0}+ B(\zeta)-B(\eps^{-1}\xi_2^+)
\end{align}
where $B(\zeta)=B(\zeta;\lambda,\eps)$ is the antiderivative of $\sqrt{b(\eps \zeta;\lambda)}$, which satisfies the estimates
\begin{align*}
B(\zeta;\lambda,\eps)&=\frac{2\sqrt{b_0}\eps^{1/2}}{3}\left(\zeta-\eps^{-1}\xi_\mathrm{Airy}\right)^{3/2}\left(1+\mathcal{O}\left(|\eps \zeta-\xi_\mathrm{Airy}|\right)\right)\\
B(\eps^{-1}\xi_2^+;\lambda,\eps)&=\frac{2}{3}\sqrt{b_0}K^{3/2}(1+\calO(\eps^{2/3}))
\end{align*}
This results in the equation for \red{$\tilde{\theta}_3$ given by}
\begin{align}\begin{split}\label{e:slowpolarAiry3tilde}
\dot{\tilde{\theta}}_3& = \frac{ \dot{b}}{4b}\sin (2\theta_3^*+2\tilde{\theta}_3)+ \calO\left(\frac{\eps}{\sqrt{b}}\right),
\end{split}
\end{align}
\red{where $\tilde{\theta}_3(\eps^{-1} \xi_2^+)=0$.} Integrating this equation, we obtain 
\begin{align}\label{e:theta3est}
\tilde{\theta}_3(\zeta) = \mathcal{O}\left(\log K+\log\left|B(\zeta)\right|\right)
\end{align}
uniformly in $K$ sufficiently large and $|\xi_+-\xi_\Airy|$ sufficiently small for all $\zeta \in I_3$. 
We now consider the variational equation of~\eqref{e:slowpolarAiry3} associated with $\theta_3$, which is given by~\eqref{e:slowpolarAiryjvar} with $j=3$ and 
\begin{align}
\begin{split}\label{e:slowAiry3varterms}
a_3(\zeta)&=\frac{\dot{b}}{2b}\cos 2\theta_3+ \calO\left(\frac{\eps}{\sqrt{b}}\right),\\
 r_3(\zeta) &= \frac{\partial_\lambda b}{2\sqrt{b}}+\partial_\lambda\left[\frac{\dot{b}}{2b}\right] \sin \theta_3\cos\theta_3+\mathcal{O}\left(\frac{\eps\partial_\lambda b}{b^{3/2}}  \right).
 \end{split}
\end{align}
We solve using variation of constants to obtain~\eqref{e:slowpolarAiryjvarconsts} with $j=3$ and $\zeta_3^\mathrm{in} = \eps^{-1}\xi_2^+$. We now estimate the quantity
\begin{align}
\int_{t}^\zeta a_3(s)\mathrm{d}s &= \int_{t}^\zeta \frac{\dot{b}}{2b}\cos 2\theta_3(s)+ \calO\left(\frac{\eps}{\sqrt{b}}\right)\mathrm{d}s= \int_{t}^\zeta \frac{\dot{b}}{2b}\cos 2\theta_3(s)\mathrm{d}s+\calO\left(1\right)
\end{align}
for $t,\zeta\in I_3$ with $t<\zeta$. We have that
\begin{align}
\int_{t}^\zeta \frac{\red{\dot b}}{2b}\cos 2\theta_3(s)\mathrm{d}s &= \frac{1}{2}\int_{t}^\zeta \frac{\eps}{\eps s-\xi_\mathrm{Airy}}\cos 2\theta_3(s)\mathrm{d}s + \calO(1)
\end{align}
Changing variables $s\to \tilde{s}(s)$, where
\begin{align*}
\tilde{s}(s)&=2B(s)+2\tilde{\theta}_3(s)\\
&=2B(s)+\mathcal{O}\left(\log K+\log\left|B(s)\right|\right)
\end{align*}
using~\eqref{e:theta3est}, so that
$2\theta_3(s) = \tilde{s} +2\theta_{3,0}-2B(\eps^{-1}\xi_2^+)$, we have that
 \begin{align}
 \begin{split}\label{e:cossinintegral}
\int_{t}^\zeta \frac{\eps}{\eps s-\xi_\mathrm{Airy}}\cos 2\theta_3(s)\mathrm{d}s 
&= \int_{\tilde{s}(t)}^{\tilde{s}(\zeta)} \frac{2}{ 3s'}\cos \left( s' +2\theta_{3,0}-2B(\eps^{-1}\xi_2^+) \right)\mathrm{d}s' + \calO(1)\\
&=\cos \left( 2\theta_{3,0}-2B(\eps^{-1}\xi_2^+) \right)\int_{\tilde{s}(t)}^{\tilde{s}(\zeta)} \frac{2}{ 3s'}\cos s'\mathrm{d}s'\\
&\qquad -\sin \left(2\theta_{3,0}-2B(\eps^{-1}\xi_2^+) \right)\int_{\tilde{s}(t)}^{\tilde{s}(\zeta)} \frac{2}{ 3s'}\sin s'\mathrm{d}s' + \calO(1)
\end{split}
\end{align}
 Using boundedness of the trigonometric integral $\mathrm{Si}(x)=\int_0^x \frac{\sin x}{x}\mathrm{d}x$ and asymptotic estimates for the integral $\mathrm{Ci}(x)=-\int_x^\infty \frac{\cos x}{x}\mathrm{d}x$ for large $x\gg1$~\cite[\S5.2]{abramowitz1948handbook}, we have that~\eqref{e:cossinintegral} is bounded, uniformly in all $K$ sufficiently large and all $|\xi_+-\xi_\Airy|$ sufficiently small, for any $t,\zeta\in I_3$ with $t<\zeta$. Hence there exists a constant $C>0$ and a coefficient $C(\xi_+)$ such that 
\begin{align}\begin{split}
\left| \int_{\eps^{-1}\xi_2^+}^{\eps^{-1}\xi_+} \exp \left(\int_{t}^\zeta a_3(s)\mathrm{d}s \right) r_3(t)\mathrm{d}t\right| &\leq C  \int_{\eps^{-1}\xi_2^+}^{\eps^{-1}\xi_+} |r_3(t)|\mathrm{d}t\leq \frac{C(\xi_+)}{\eps},
\end{split}
\end{align}
uniformly for all sufficiently large $K>0$, where 
$C(\xi_+)\to 0$ as $|\xi_+-\xi_\mathrm{Airy}|\to0$. We therefore obtain that 
\begin{align}\begin{split}
|\partial_\lambda \theta_3(\eps^{-1}\xi_+)|& \leq \frac{C(\xi_+)}{\eps}+ \mathcal{O}\left(\partial_\lambda \theta_3(\eps^{-1}\xi_2^+)\right).
\end{split}
\end{align}
In order to estimate $\partial_\lambda \theta_3(\eps^{-1}\xi_2^+)$, we use the relation $e^{i\theta_3} = \left(\eps^{1/3}\cos \theta_2+i \sqrt{b} \sin\theta_2\right)\left(\eps^{2/3}\cos^2\theta_2+b\sin^2\theta_2\right)^{-1/2}$, from which we obtain
\begin{align}
\begin{split}
\partial_\lambda \theta_3(\eps^{-1}\xi_2^+) &=  \left. \frac{\partial \theta_3}{\partial b}\right|_{\zeta=\eps^{-1}\xi_2^+}\partial_\lambda b(\eps^{-1}\xi_2^+)+\left. \frac{\partial \theta_3}{\partial \theta_2}\right|_{\theta_2 = \theta_2(\eps^{-1}\xi_2^+)}\partial_\lambda\theta_2(\eps^{-1}\xi_2^+)\\
&= \mathcal{O}\left(e^{-\eta/\eps}|\partial_\lambda \theta_-|+|\eps^{-2/3}|\right).
\end{split}
\end{align}
We finally obtain
\begin{align}\begin{split}
|\partial_\lambda \theta_3(\eps^{-1}\xi_+)|& \leq \frac{C(\xi_+)}{\eps}+ \mathcal{O}\left(e^{-\eta/\eps}|\partial_\lambda \theta_-|+|\eps^{-2/3}|\right).
\end{split}
\end{align}
In particular, by adjusting $\xi_+$, the exit angle from the Airy region $I_3$, which becomes the entry angle to the oscillatory region $I_\mathrm{osc}$, satisfies $|\partial_\lambda \theta| = o(\eps^{-1})$ and we may proceed as in the oscillatory regime. Notably, the change from the coordinates in region $I_3$ to those in the uniformly oscillating region $\xi\geq\xi_+$ is given by $T\mathrm{diag}(\sqrt{b},1) = \sqrt{b/2}\begin{pmatrix}1&-1\\1&1\end{pmatrix}$, which (for the angle)  is simply a rotation by $\pi/4$.}
\end{proof}

We emphasize that, by reversing the direction of $\zeta$, the same proof applies to an Airy transition exit rather than an entry, i.e. $[0,T]= I_\osc\cup I_\Airy^-$, where on the reflection of $I_\Airy^-$ we have Hypothesis~\ref{h:Airy}. Moreover, in combination of these cases, the proof applies to the case $[0,T]=I_\Airy\cup I_\osc\cup I_\Airy^-$, i.e.,  Airy transitions at entry and exit, which is one of the cases in the FitzHugh--Nagumo system as discussed in \S\ref{sec:applytofhn}.


\subsection{Adding a fast exit layer}\label{s:accumulation_withlayer}
\blue{
In this section, we discuss the inclusion of a fast layer in which the solution leaves a neighborhood of the slow manifold $\mathcal{M}_\eps$ by ``jumping off" along one of the hyperbolic fast directions, and the exit boundary condition is prescribed at some small, but $\mathcal{O}(1)$, distance from the slow manifold. As mentioned in~\S\ref{s:cases}, we will restrict our attention to the case of a fast \emph{exit layer} from a uniformly oscillating region, as this is most relevant for our purposes with the FitzHugh--Nagumo system. However only minor modifications to the following arguments are needed in the case of an exit layer after an Airy transition, a fast entry layer, or the case of both fast entry and exit layers.

We consider the same setup and assumptions as in~\S\ref{s:slowabssetup} to hold on the slow interval 
\begin{equation}\label{e:decompnew2}
[0,T] = I_\slow = I_\Airy \cup I_\osc.
\end{equation}
However, to account for the fast layer, we append an interval of length $L_\eps=K|\log \eps|$ on the fast timescale for some $K>0$, that is we consider $\zeta\in[0,T_\eps]$, where $T_\eps:=T/\eps+L_\eps$. Equivalently, on the slow timescale, we consider the interval
\begin{equation}\label{e:decompnewfast}
\xi\in [0,\eps T_\eps] =  I_\Airy \cup I_\osc \cup I_\fast,
\end{equation}
as in Figure~\ref{f:intervals}, and we now prescribe the boundary condition $Q^+(\lambda,\eps)$ at $\xi=\eps T_\eps$, and the interval $I_\fast$ has length $\calO(\eps |\log\eps|)$. With $\Phi(\zeta,\tilde{\zeta};\lambda,\eps)$ the evolution operator of~\eqref{e:genslowlinear} now on $[0,T_\eps]$, analogously to~\S\ref{s:accumulation_nolayer} we consider the modified existence problem for an eigenfunction 
\begin{align}\label{e:condefct_layer}
\Phi(T_\eps-L_\eps,0;\lambda,\eps)Q^-(\lambda,\eps)\cap \Phi(T_\eps-L_\eps,T_\eps;\lambda,\eps) Q^+(\lambda,\eps) \neq \{0\},
\end{align}
by transporting the subspace $Q^+$ to $\zeta = T_\eps-L_\eps= T/\eps$.

We remark that in the case of a traveling wave solution $\phi_\eps(\zeta)$, due to the exponential decay along the corresponding fast fiber (due to normal hyperbolicity of the slow manifold), 
the interval of length $L_\eps$ on the fast timescale (for sufficiently large $K>0$ fixed independent of $\eps$) 
is large enough to guarantee that the solution $\phi_\eps(T/\eps)$ is within $\mathcal{O}(\eps)$ of the corresponding reduced trajectory on the slow manifold; this in turn implies that the assumptions of~\S\ref{s:slowabssetup} indeed hold on the slow interval $I_\slow$, despite the boundary condition $Q^+(\lambda,\eps)$ now being prescribed at $\zeta=T_\eps$, where $\phi_\eps(T_\eps)$ is an $\mathcal{O}(1)$ distance from $\mathcal{M}_\eps$. With the inclusion of so-called corner-type estimates, which we describe in Hypothesis~\ref{h:bl} below, the case of a fast layer does not affect the accumulation of eigenvalues as in~\S\ref{s:accumulation_nolayer}.

For a fast boundary layer, the underlying solution leaves the slow manifold along a singular fast trajectory of the layer problem. For the singular limit $\phi_0$ this means there is a well-defined `take-off point' on the slow manifold, located in terms of  $\xi$ at $\xi=T$, the limit of the right endpoint of $I_\slow$ at $\eps=0$. Relevant for estimates is the linearisation in that point which we denote as $A_\fast^\infty(\lambda):= A_\slow(T;\lambda,0)$, cf.\ \eqref{e:genslowlinear0}. The layer interval $I_\fast = [T,T+\eps L_\eps]$ viewed on the $\zeta$-scale is $\eps^{-1} I_\fast = [T/\eps,T/\eps+L_\eps]$ and is parameterised by $\tzeta\in[-L_\eps,0]$; for the limiting problem we then have $\tzeta\in(-\infty,0]$ with $A_\fast^\infty$ approached as $\tzeta\to -\infty$. We thus consider $L_\eps=-K\log \eps$, for sufficiently large $K>0$ and assume the following.
\begin{hypothesis}\label{h:bl}
(Boundary layer) There exists $A_\fast(\tzeta;\lambda)$ such that
\begin{align}
|A(\tzeta+T_\eps; \lambda, \eps) - A_\fast(\tzeta;\lambda)| = \mathcal{O}(\eps \log \eps)
\end{align}
uniformly on the interval $\tzeta\in[-L_\eps,0]$ for all sufficiently small $\eps>0$. Further, there is $\tnu>0$ such that for $\tzeta \in (-\infty,0]$
\begin{align}\label{e:decayfront}
|A_\fast(\tzeta;\lambda)-A_\fast^\infty(\lambda)| = \mathcal{O}(e^{\tnu \tzeta}).
\end{align}
\end{hypothesis}



The estimate \eqref{e:decayfront} relates to the convergence of the fast part of the travelling wave profile $\phi_\eps$ to its asymptotic state, and can be obtained in specific example systems (such as the FitzHugh--Nagumo system) through corner-type estimates, see e.g.~\cite{esz} or \cite[Theorem 4.3]{CdRS}. 
Under Hypothesis~\ref{h:osc}, and with the exponential weight $\Re\nu_{i_\infty}$, the matrix $A_\fast^\infty(\lambda)$ possesses a complex conjugate pair of eigenvalues $\pm \rmi \omega^\infty(\lambda)\neq 0$ for $\lambda\in\R$ near $\lambda_0$, and $A_\fast^\infty(\lambda)$ admits a spectral decomposition into stable, center, and unstable eigenspaces. These yield an exponential trichotomy on $\mathbb{R}^-$ for $U_\tzeta = B U$ with $B=A_\fast^\infty(\lambda)$ and, due to roughness, under Hypothesis~\ref{h:bl} also for $B=A_\fast(\tzeta;\lambda)$\blue{, cf.\ \cite{Pal82}}. We denote the corresponding subspaces by $E^{\mathrm{ss,c,uu}}_\fast(\tzeta,\lambda)$ and assume the following with respect to the boundary subspace $Q^+$.}

\begin{hypothesis}\label{h:bl_vector}
\red{There exists $C>0$ such that for all $\eps\in(0,\eps_0]$, the following holds. There exists a unit vector $v_+(\lambda,\eps)\in E^\cc_\fast(0;\lambda)\oplus E^{\uu}_\fast(0;\lambda)$ such that $Q^+(\lambda,\eps) = \mathrm{span}\{v_+\}\oplus\tilde{Q}^+$ where
\begin{align}
\tilde{Q}^+\oplus E^\cc_\fast(0;\lambda)\oplus E^{\uu}_\fast(0;\lambda)=\mathbb{C}^n,
\end{align}
\blue{and} we can decompose any vector in $\mathrm{span}\{v_+\}$ as $v^\cc_++v^\uu_+$ for some $v_+^\cc\in E^\cc_\fast(0;\lambda)$ and  $v_+^\uu\in E^{\uu}_\fast(0;\lambda)$ with $|v^\uu_+|\leq C|v^\cc_+|$. Furthermore, we can choose $v_+$ with these properties, such that $v_+$ is continuous for $\eps\in(0,\eps_0]$, and $\partial_\lambda v_+$ is continuous for $\eps\in(0,\eps_0]$, and there exists $\kappa>0$ such that $|\partial_\lambda v_+| = o\left(\eps^{-1+\kappa}\right)$ uniformly in $\eps\in(0,\eps_0]$. }

The entry boundary subspace $Q^-$ is assumed to satisfy the conditions in Hypothesis~\ref{h:osc_vector}. 
\end{hypothesis}
\blue{
\begin{remark}
We comment briefly on the discrepancy between the estimates required on the derivative $|\partial_\lambda v_\pm|$ in Hypotheses~\ref{h:osc_vector} and~\ref{h:bl_vector}. The slightly stronger estimate in Hypothesis~\ref{h:bl_vector} accounts for deviations which may occur due to passage along the fast interval of length $L_\eps=\mathcal{O}(|\log \eps|)$.
\end{remark}
}

The eigenvalue accumulation in the layer case can now be formulated analogous to Proposition \ref{p:Airy}; the statement also holds with empty $I_\Airy$.
\begin{theorem}\label{t:layers}
Assume \blue{$[0,T+\eps L_\eps] = I_\Airy\cup I_\osc\cup I_\fast$} with $L_\eps = K|\log\eps|$, $K>0$, and Hypotheses~\ref{h:well}--\ref{h:bl_vector} hold for $\lambda=\lambda_0$, and fix $r>0$. If $\Delta'\neq 0$, cf.\ \eqref{e:delpr}, then for all sufficiently small $\eps>0$, $\specpt\cap B_r(\lambda_0)\subset\C$ has cardinality $\calO(1/\eps)$, where $B_r(\lambda_0)$ is the disc around $\lambda_0$ with radius $r$.
\end{theorem}

\begin{proof}
\blue{
Noting $P^\cc_\fast(0)v_+\neq 0$ by Hypothesis~\ref{h:bl_vector}, under the evolution $\Phi_\fast$ of the reduced system
\begin{align}\label{e:bl_reduced}
U_\tzeta = A_\fast(\tzeta;\lambda)U,
\end{align}
using for instance~\cite[Lemma 2.2]{KS} and \eqref{e:decayfront} we have
\begin{align}
\Phi_\fast(\tzeta,0;\lambda)v_+^\cc = v_+^\infty e^{\rmi \omega^\infty(\lambda)\tzeta}+\mathrm{c.c.} +\calO(e^{\tnu \tzeta}) , \qquad -\infty < \tzeta \leq 0\label{e:layerBndryRight2}
\end{align}
where $v_+^\infty$ 
is an eigenvector for the eigenvalue $\rmi \omega^\infty(\lambda)$ 
in the center subspace of $A_\fast^\infty(\lambda)$ in the weighted space. \red{Without loss of generality, possibly taking a smaller value of $\tnu$ from Hypothesis~\ref{h:bl}, the exponents $\tnu$ in \eqref{e:layerBndryRight2} and in Hypothesis~\ref{h:bl} are equal.}

Combined with Hypothesis~\ref{h:bl}, using the reduced system~\eqref{e:bl_reduced} we next transport the boundary space $Q^+$ to the right endpoint $T_\eps-L_\eps$ of the rescaled slow interval $\eps^{-1}I_\slow$. Under the evolution of~\eqref{e:genslowlinear} we have 
\begin{align}
\Phi(T_\epsilon-L_\eps,T_\eps;\lambda,\eps)Q^+ &=\mathrm{span}\{ \Phi_\fast(-L_\eps,0;\lambda)v_+^\cc \}+ E^{\ss}_\fast(-L_\eps;\lambda)+\calO(\eps |\log\eps|^2), \label{e:layerBndryRight3}
\end{align}
where the error terms are in the sense of the distance between unit basis vectors. As in the proof of Proposition~\ref{p:uni}, decisive will be the projection onto the center space of the slow trichotomy.

Each solution on the interval $\zeta\in \eps^{-1}I_\slow$ can be written in terms of the trichotomy evolutions as
\begin{align}
\psi_\mathrm{slow}(\zeta)  = \Phi^\uu(\zeta,T_\eps-L_\eps;\lambda,\eps)\alpha+\Phi^\ss(\zeta,0;\lambda,\eps)\beta+ \Phi^\cc(\zeta,0;\lambda,\eps)\gamma
\end{align}
for some $\alpha\in E^\uu(T_\eps-L_\eps; \lambda, \eps), \beta \in E^\ss(0;\lambda, \eps), \gamma \in E^\cc(0;\lambda, \eps)$. We now match the solution  $\psi_\mathrm{slow}(\zeta)$ with the boundary subspaces by projecting onto each of the subspaces $E^{\uu,\ss,\cc}$ at $\zeta= 0$ and at $T_\eps-L_\eps$. 

At $\zeta=0$, by Hypothesis~\ref{h:bl_vector}, there exists $C>0$ such that we can decompose any vector in $\mathrm{span}\{v_-\}$ as $v_-^\cc+v_-^\ss$ with $v_-^\cc\in E^\cc(0;\lambda,\eps)$ and  $v_-^\ss\in E^{\ss}(0;\lambda,\eps)$, and $|v^\ss_-|\leq C|v^\cc_-|$. Thus, any vector $q_-$ in the space $Q^- $ can be written as $q_- = v_-^\cc+v^\ss_-+ \tilde{q}_-$ for some $\tilde{q}_-\in \tilde{Q}^-$ and $v^\cc_-, v^\ss_-$ as above.
We project with each of $P^\mathrm{c,uu,ss}(0;\lambda,\eps)$ and obtain the three equations:
\begin{equation}\label{e:trichblL}
\Phi^\uu(0,T_\eps-L_\eps)\alpha = P^\uu(0)\tilde{q}_-,\quad
\beta = v^\ss_- +P^\ss(0)\tilde{q}_-,\quad
\gamma = v_-^\cc+P^\cc(0)\tilde{q}_-.
\end{equation}

Similarly, by Hypothesis~\ref{h:bl_vector}, there exists $C>0$ such that any vector in $\mathrm{span}\{v_+\}$ can be written as $v^\cc_++v^\uu_+$ where $v_+^\cc\in E_\fast^\cc(0;\lambda)$ and  $v_+^\uu\in E_\fast^{\uu}(0;\lambda)$ and $|v^\uu_+|\leq C|v^\cc_+|$. Using \eqref{e:layerBndryRight3}, at $\zeta=T_\eps-L_\eps$, any vector $q_+$ in the space $\Phi(T_\epsilon-L_\eps,T_\eps;\lambda,\eps)Q^+ $ can therefore be written as 
\[
q_+=  \Phi_\fast(-L_\eps,0)v_+^\cc + q^\mathrm{ss}_++\calO\left(\eps |\log\eps|^2(|v^\cc_+|+|q^\ss_+|)\right)
\]
for some $q^\mathrm{ss}_+\in E^{\ss}_\fast(-L_\eps)$ and $v^\cc_+$ as above. We similarly project with each of $P^\mathrm{c,uu,ss}(T_\eps-L_\eps;\lambda,\eps)$. Using the proximity of the projections of the exponential trichotomy along $\eps^{-1}I_\mathrm{slow}$ with the spectral projections of $A(\zeta;\lambda, \eps)=A_\mathrm{slow}(\eps\zeta;\lambda, \eps)$ from \eqref{e:proj} of Lemma~\ref{l:diag}, we obtain the three equations
\begin{equation}\label{e:trichblR}
\begin{aligned}
\alpha &= H_1(v_+^\cc,q^\ss_+)\\
\Phi^\ss(T_\eps-L_\eps,0)\beta &= q^\mathrm{ss}_++H_2(v_+^\cc,q^\ss_+) \\
\Phi^\cc(T_\eps-L_\eps,0)\gamma&= \Phi_\fast(-L_\eps,0)v_+^\cc+H_3(v_+^\cc,q^\ss_+)
\end{aligned}
\end{equation}
where are $H_j, j=1,2,3$ are linear maps satisfying
\begin{align*}
H_j(v_+^\cc,q^\ss_+) &= \mathcal{O}\left(\eps|\log\eps|^2\left(|v^\cc_+|+|q^\ss_+|\right)\right), \quad j=1,2,3.
\end{align*}
Plugging in the expression for $\alpha$ from \eqref{e:trichblR} into \eqref{e:trichblL}, and that for $\beta$ from \eqref{e:trichblL} into \eqref{e:trichblR}, we obtain 
\begin{equation}\label{e:hypbl}
P^\uu(0)\tilde{q}_- =H_4(v_+^\cc,q^\ss_+),\quad
q^\mathrm{ss}_+=H_5(v^\ss_-,v_+^\cc,v^\uu_+,q^\ss_+),
\end{equation}
where the linear maps $H_4,H_5$ satisfy, for sufficiently small $\eta>0$ related to the trichotomy rates, 
\begin{align*}
H_4(v_+^\cc,v^\uu_+,q^\ss_+) &= \mathcal{O}\left(e^{-\eta/\eps}\left(|v^\mathrm{c}_+|+|q^\ss_+|  \right)\right),\\
H_5(v_+^\cc,v^\uu_+,v^\ss_-) &= \mathcal{O}\left(\eps|\log\eps|^2\left(|v_+^\cc|+|q^\ss_+|\right)+e^{-\eta/\eps}(|v^\cc_-|+|\tilde{q}_-|)\right).
\end{align*}
Noting that $P^\uu(0)$ is invertible on the space $\tilde{Q}^-$ due to Hypothesis~\ref{h:bl_vector}, we solve \eqref{e:hypbl} by the implicit function theorem to obtain
\[
\tilde{q}_-=H_6(v^\cc_-,v_+^\cc),\quad
q^\mathrm{ss}_+=H_7(v^\cc_-,v_+^\cc),
\]
where $H_6,H_7$ satisfy
\begin{align*}
H_6(v^\cc_-,v_+^\cc) &= \mathcal{O}\left(e^{-\eta/\eps}\left(|v^\cc_-|+|v^\mathrm{c}_+|\right)\right)\\
H_7(v^\cc_-,v_+^\cc) &= \mathcal{O}\left(\eps|\log\eps|^2|v_+^\cc|+e^{-\eta/\eps}|v^\cc_-| \right).
\end{align*}
Finally, substitution into the third equation of \eqref{e:trichblR} gives the remaining equation on the center subspace
\begin{align}
\begin{split}\label{e:finallayerbcs}
\Phi^\cc(T_\eps-L_\eps,0;\lambda,\eps)v^\mathrm{c}_-&= \Phi_\fast(-L_\eps,0)v_+^\cc+H_8(v^\cc_-,v_+^\cc),
\end{split}
\end{align}
where $H_8$ satisfies
\begin{align*}
H_8(v^\cc_-,v_+^\cc) &= \mathcal{O}\left(\eps|\log\eps|^2|v_+^\cc|+e^{-\eta/\eps}|v^\cc_-| \right).
\end{align*}
It remains to solve \eqref{e:finallayerbcs} to obtain eigenfunctions. This amounts to solving the boundary problem in the purely slow regime as in Proposition~\ref{p:Airy} with entry boundary condition $v^\mathrm{c}_-$ at $\zeta=0$, but with the exit boundary condition (now at $\zeta=T_\eps-L_\eps$) replaced by the right-hand-side of~\eqref{e:finallayerbcs} instead of $v^\mathrm{c}_+$. \red{We note that since the derivatives of $v_+^\cc$ with respect to $\lambda$ are by assumption all bounded as in Hypothesis~\ref{h:bl_vector}, the derivatives of the center subspace boundary conditions on the right hand side of~\eqref{e:finallayerbcs} with respect to $\lambda$ can be bounded by $o(\eps^{-1})$.} Hence the argument for the accumulation of eigenvalues proceeds as in the uniformly oscillatory case.}
\end{proof}

\section{Application to the FitzHugh--Nagumo system}\label{sec:applytofhn}

In this section, we apply our results to the FitzHugh--Nagumo system. Our aim is to explain the accumulation of eigenvalues which occurs along the $1$-to-$2$ pulse transition, and demonstrated by the numerical computations in Figures~\ref{f:fhn-eps0p01} and~\ref{f:fhn-eps0p02}. In particular, we will explain the appearance of $\mathcal{O}(1/\eps)$ eigenvalues as $\eps\to0$ as well as the phenomenon that the number of eigenvalues initially increases along the transition, reaching a maximum near the middle of the transition, then decreasing as the transition continues; see Remark~\ref{r:accumulation_s_vs_eps}. These phenomena are related to the time spent (in the traveling wave coordinate) of a given pulse near the middle slow manifold $\mathcal{M}^m_\eps$, which we will show exhibits slow absolute spectrum in the sense of~\S\ref{s:slowabsnew}. Based on the theory in~\S\ref{s:slowabsnew}-\ref{s:gentheory}, the number of accumulating eigenvalues is determined by the (spatial) time spent along such a manifold.

In preparation, the eigenvalue problem from linearising \eqref{e:fhn} about a traveling pulse solution $\phi(\zeta;\eps)=(u,w)(\zeta; \eps)$  reads
\begin{align}\label{e:fhnstab}
\lambda \begin{pmatrix}u\\w\end{pmatrix} = \begin{pmatrix}u_{\zeta\zeta} - cu_\zeta + f'(u(\zeta;\eps))u - w\\ -c w_\zeta + \eps(u-\gamma w)\end{pmatrix} 
\end{align}
again in the spatial comoving frame $\zeta= x+ct$. Equivalently, the first order ODE formulation reads
\begin{equation}\label{e:fhnlin1st}
\begin{aligned}
\dot{U} = A(\zeta;\lambda, \eps)U
\end{aligned}
\end{equation}
where $U=(u,u_\zeta,w)$ and
\begin{equation}
\begin{aligned}
A(\zeta;\lambda, \eps):=\begin{pmatrix}0 &1& 0 \\ \lambda-f'(u(\zeta;\eps)) & c & 1  \\ \frac{\eps}{c} & 0 & -\frac{(\lambda+\gamma \eps)}{c} \end{pmatrix}
\end{aligned}
\end{equation}
The eigenvalue problem~\eqref{e:fhnlin1st} is a first order, nonautonomous ODE which depends on the pulse solution through the term $f'(u(\zeta;\eps))$. As described in~\S\ref{sec:pulse_geometry}, this pulse solution is composed of a concatenated sequence of slow/fast orbits: the slow pieces are portions of slow manifolds which are perturbations of normally hyperbolic portions of the critical manifold $\mathcal{M}_0=\{(u,0,f(u)):u\in\mathbb{R}\}$. In the singular limit $\eps=a=0, c=1/\sqrt{2}$, we recall from~\S\ref{sec:pulse_geometry} that the critical manifold admits three normally hyperbolic segments
\begin{align}
\mathcal{M}^\ell_0 = \mathcal{M}_0\cap\{u<0\}, \quad \mathcal{M}^m_0 = \mathcal{M}_0\cap\{0<u<2/3\}, \quad \mathcal{M}^r_0 = \mathcal{M}_0\cap\{u>2/3\},
\end{align}
which meet at two nonhyperbolic fold points. Away from these fold points, compact segments of these branches of the critical manifold perturb to slow manifolds $\mathcal{M}^\ell_\eps, \mathcal{M}^m_\eps, \mathcal{M}^r_\eps$ which are $\mathcal{O}(\eps)$-perturbations of their singular counterparts. In addition to traversing these slow manifolds, the pulse solutions also follow fast jumps between these manifolds, which arise as perturbations of heteroclinic orbits of the singular fast subsystem~\eqref{e:layer}; see~\S\ref{sec:pulse_geometry}.

The precise details of this construction depend on the exact location along the single-to-double-pulse transition. To see how the slow absolute spectrum phenomenon arises in this context, we need to understand the spatial eigenvalue structure along the slow manifolds $\mathcal{M}^\ell_\eps, \mathcal{M}^m_\eps, \mathcal{M}^r_\eps$ relative to that of the asymptotic rest state $(u,w)=(0,0)$.

\subsection{Slow absolute spectrum in the FitzHugh--Nagumo system}\label{sec:slowabsmiddle}

We first note that here we are concerned only with the accumulation of \emph{unstable} eigenvalues, i.e. those in the right half plane. We show in the next subsection~\S\ref{sec:essFHN} that along the pulse transition, for $\gamma>3/2$, the associated essential spectrum, corresponding to the rest state, is stable, that is, contained entirely in the left half plane (see Lemma~\ref{l:fhness_gammabound} below). Therefore, in locating unstable eigenvalues, this automatically restricts our attention to values of $\lambda$ which lie to the right of the essential spectrum. 

In~\S\ref{sec:slowabs_middleslowmanifold}, we then show that the middle slow manifold $\mathcal{M}^m_\eps$ admits slow absolute spectrum, in the sense of Hypotheses~\ref{h:osc}--\ref{h:Airy}.

\subsubsection{Essential spectrum for transitional pulses in the FitzHugh--Nagumo system}\label{sec:essFHN}
In this section we prove the following lemma regarding the essential spectrum of transitional pulses in~\eqref{e:fhn}.
\begin{lemma}\label{l:fhness_gammabound}
Consider \eqref{e:fhn} with $\gamma>3/2$ and the associated stability problem~\eqref{e:fhnstab} for a transitional pulse $\phi(\zeta;s, \eps)$ with \blue{$(c,a)=(c,a)(s,\eps)$}. For all sufficiently small $\eps>0$, the essential spectrum $\Sigma_\mathrm{ess}$ is confined to the left half plane $\Re \lambda<0$.
\end{lemma}

The essential spectrum for homogeneous equilibria is determined by the dispersion relation as the set of solutions $\lambda\in \C$ to $d(\lambda, \rmi k)=0$ for $k\in \R$. Analogous to \cite[Prop. 4.1]{CdRS} we have 
\[
d(\lambda,\rmi k) = (\lambda-f'(0)+c \rmi k + k^2)(\lambda + c \rmi k+ c \gamma \eps)+ c\eps.
\]
Hence, to leading order the second factor gives the critical $\lambda=-\rmi c k$. We have the following.

\begin{lemma}\label{l:fhness}
Consider \eqref{e:fhn}. For $a\in(-1,0)$ the essential spectrum $\Sigma_\mathrm{ess}$ of the equilibrium $u=w=0$ in \eqref{e:fhn} is strictly stable if and only if 
$c\eps \gamma>-a$ and any $\lambda\in\Sigma_\mathrm{ess}$  satisfies 
\[
\Re(\lambda)\leq\lambda^*:=\frac 1 2 \left(-a-c\eps \gamma+\Re\left(\sqrt{(a+c\eps \gamma)^2-4c\eps (1+a)}\right)\right).
\]
The equality and thus the maximal real part is realised for wavenumber $k=0$, i.e. $d(\lambda^*,0)=0$.
\end{lemma}

\begin{proof}
Since the real part of $\lambda$ is that of $\tlambda:=\lambda + c \rmi k$ we solve the dispersion relation for $\tlambda$, which yields
\[
\tlambda^\pm (k) = \frac 1 2 \left(f'(0)-c\eps \gamma -k^2 \pm\sqrt{(f'(0)-c\eps \gamma-k^2)^2-4c\eps (k^2-f'(0)+1)}\right).
\]
It follows that the essential spectrum is unstable if $f'(0)-c\eps \gamma=-a-c\eps\gamma>0$ (since $f'(0)=-a$). 

On the other hand, if $-a-c\eps \gamma<0$ we use $a\in(-1,0)$ so that $k^2-f'(0)+1>0$, and with $c>0$ this means 
\[
\Re(\tlambda^-(k)) \leq \Re(\tlambda^+(k))<0,
\]
so that the essential spectrum is strictly stable. More precisely, if the discriminant is negative
\[
\Re(\tlambda^+(k)) = \frac 1 2 (f'(0)-c\eps \gamma -k^2) \leq \frac 1 2 (-a-c\eps \gamma)
\]
and if it is positive we argue as follows. Since $\tlambda$ is real in this case, the maximum real part lies at a double root with respect to $k$ of the modified dispersion relation with $\tlambda$ (which is a function of $K=k^2$) given by
\[
\tilde d(\tlambda,K) = (\tlambda-f'(0) + K)(\tlambda + c \gamma \eps)+ c\eps.
\]
A double root satifsies $2k \partial_K \tilde d(\tlambda,K)=2k(\tlambda+c\gamma\eps)=0$ so that either $\tlambda = -c\eps\gamma$ for some $k$ or $k=0$. In the latter case the above formula for $\tlambda^+(0)$ and $f'(0)=-a>0$ implies that $\Re(\tlambda^+(0))>-c\gamma\eps$. Finally, since $k=0$ always is a double root, the maximal real part lies at the real part of  $\tlambda^+(0)$ as claimed.
\end{proof}
From this, we immediately deduce the results of Lemma~\ref{l:fhness_gammabound}.
\begin{proof}[Proof of Lemma~\ref{l:fhness_gammabound}]
We recall that the transitional pulses occur in an exponentially thin interval in parameter space, so that all transitional pulses occur at parameter values $(c,a)$ exponentially close to the maximal canard values $(c_*,a_*)(\eps)$ defined in~\eqref{e:maxcanard}. Therefore, using Lemma~\ref{l:fhness} and in particular the stability condition $c\eps \gamma>-a$, we compute that $\Sigma_\mathrm{ess}$ is stable precisely when $\gamma>3/2$, and $\eps>0$ is taken sufficiently small.
\end{proof}

\subsubsection{Slow absolute spectrum along the middle slow manifold $\mathcal{M}^m_\eps$}\label{sec:slowabs_middleslowmanifold}

To the right of the essential spectrum, we can determine the Morse index of the asymptotic rest state by setting $\lambda=0$, and computing the spatial eigenvalues of the asymptotic matrix 
\begin{equation}
\begin{aligned}
A_{\pm\infty}(\lambda, \eps)=\lim_{\zeta\to\pm\infty}A(\zeta;\lambda, \eps):=\begin{pmatrix}0 &1& 0 \\ \lambda-f'(0) & c & 1  \\ \frac{\eps}{c} & 0 & -\frac{(\lambda+\gamma \eps)}{c} \end{pmatrix}
\end{aligned}
\end{equation}
 in the linear system~\eqref{e:fhnlin1st}, which admits a single eigenvalue of positive real part, and two eigenvalues of negative real part, so that Hypothesis~\ref{h:well} is satisfied with $i_\infty=1$. 
 
For the remainder of the analysis, it is convenient to replace~\eqref{e:fhnlin1st} with the weighted eigenvalue problem
\begin{equation}\label{e:fhnlin1st_weighted}
\begin{aligned}
\dot{U} = A_\eta(\zeta;\lambda, \eps)U, \qquad 
A_\eta(\zeta;\lambda, \eps):=\begin{pmatrix}-\eta(\zeta) &1& 0 \\ \lambda-f'(u(\zeta;\eps)) & c-\eta(\zeta) & 1  \\ \frac{\eps}{c} & 0 & -\frac{(\lambda+\gamma \eps)}{c} -\eta(\zeta) \end{pmatrix}
\end{aligned}
\end{equation}
where the weight function $\eta(\zeta)\geq 0$ is $\zeta$-dependent and bounded. Weighting the eigenvalue problem has multiple effects: firstly, this shifts the essential spectrum to the left, away from the imaginary axis, and secondly it shifts the spatial eigenvalues of the matrix $A(\zeta;\lambda,\eps)$ by an amount $\eta(\zeta)$, which allows for the construction of exponential dichotomies/trichotomies along the various slow manifolds which are uniform in $\eps$. Provided we choose the weight function to be smooth and bounded, and we fix the limits $\eta_\infty^\pm:=\lim_{\zeta\to\pm \infty}\eta(\zeta)$ so that the asymptotic matrices $A_{\eta,\pm\infty}(\lambda, \eps)=\lim_{\zeta\to\pm\infty}A_\eta(\zeta;\lambda, \eps)$ have the same splitting as $A_{\pm\infty}(\lambda, \eps)$ with $i_\infty=1$,  then any eigenvalues for the weighted problem, which lie to the right of the essential spectrum, are also eigenvalues of the unweighted problem, up to weighting the corresponding eigenfunction by a factor 
\begin{align}
U(\zeta)\to e^{\int_0^\zeta \eta(\tilde{\zeta})\mathrm{d} \tilde{\zeta}}U(\zeta).
\end{align}

We have the following from~\cite{CdRS}, which is deduced from a straightforward calculation.
\begin{proposition}{\cite[Proposition 6.5]{CdRS}}\label{p:leftrightsplit}
Fix $\delta>0$ sufficiently small, and consider the weighted system~\eqref{e:fhnlin1st_weighted} along any interval $I_0$ for which $\{u(\zeta;\eps):\zeta \in I_0\}\subseteq [-1-\delta, \delta]\cup[2/3-\delta, 1+\delta]$. There exists a constant $\eta_0>0$ such that if the weight function $\eta(\zeta)$ satisfies $\eta(\zeta)=\eta_0$ for $\zeta\in I_0$, for any $\lambda$ satisfying $\mathrm{Re}\lambda\geq0$ the matrices $A_\eta(\zeta;\lambda, \eps), \zeta \in I_0$ have the same spectral splitting as $A_{\pm\infty}(\lambda, \eps)$, that is, one eigenvalue of positive real part, and two of negative real part. 
\end{proposition}

The primary consequence of Proposition~\ref{p:leftrightsplit} is that neither of Hypotheses~\ref{h:osc}--\ref{h:Airy} can be satisfied along either of the left or right slow manifolds $\mathcal{M}^\ell_\eps$, $\mathcal{M}^r_\eps$. However, the situation is different for pulses which traverse portions of the middle branch $\mathcal{M}^m_\eps$. We have the following.
\begin{proposition}\label{p:fhn_slowabs}
Consider the eigenvalue problem~\eqref{e:fhnlin1st} for fixed $\lambda \in [ 0,5/24 )$ and sufficiently small $\eps>0$, where $u=u(\zeta;\eps)$ is viewed as a parameter and $\left|(c,a)-(1/\sqrt{2},0)\right|\leq C\eps$ for some fixed $C>0$. There exist
\begin{align}
u^-_\mathrm{A}(\lambda,\eps) =\frac{1}{3}-\frac{\sqrt{10-48\lambda}}{12}+\mathcal{O}(\eps), \qquad u^+_\mathrm{A}(\lambda,\eps) =\frac{1}{3}+\frac{\sqrt{10-48\lambda}}{12}+\mathcal{O}(\eps)
\end{align}
such that for any closed interval $I_u \subset\left(u^-_\mathrm{A}(\lambda,\eps),u^+_\mathrm{A}(\lambda,\eps)\right)$, Hypothesis~\ref{h:osc} is satisfied along the portion of the middle branch $\mathcal{M}^m_\eps\cap \{u\in I_u\}$.
The points at $u=u^\pm_\mathrm{A}(\lambda,\eps)$ correspond to Airy points, and Hypothesis~\ref{h:Airy} is satisfied in a neighborhood of each of these endpoints, up to a reversal of \blue{the interval} in the case of $u^+_\mathrm{A}(\lambda, \eps)$.
\end{proposition}
\begin{remark}
In light of Theorem~\ref{thm:mainexistence}, the results of Proposition~\ref{p:fhn_slowabs} are valid in particular for any of the transitional pulses $\phi(\cdot;s,\eps)$, in the sense that there exists a portion of the middle slow manifold $\mathcal{M}^m_\eps$ along which Hypothesis~\ref{h:osc} is satisfied. However, depending on the value of $s$, the transitional pulse in question may not actually traverse this portion of $\mathcal{M}^m_\eps$, and hence will not accumulate eigenvalues. A detailed discussion of which transitional pulses satisfy the requisite boundary conditions and the dependence on the parameter $s$ follows in~\S\ref{sec:fhnaccumulation}.
\end{remark}
\begin{proof}[Proof of Proposition~\ref{p:fhn_slowabs}]
In the singular limit $(c,a,\eps)=(1/\sqrt{2},0,0)$, the middle branch is defined by the set $\mathcal{M}^m_0:=\{(u,0,f(u)): u\in(0,2/3)\}$. Away from the fold points, the perturbed slow manifold $\mathcal{M}^m_\eps$ lies within $\mathcal{O}(\eps)$ of this set. We examine the linearized system~\eqref{e:fhnlin1st} along this set in the limit $\eps=0$. The corresponding spatial eigenvalues $\nu=\nu(\lambda,u;\eps)$ satisfy
\begin{align}
(\nu(\lambda,u;0)(\nu(\lambda,u;0)-c)-\lambda+f'(u))\left(\frac{\lambda}{c}-\nu(\lambda,u;0)\right)=0,
\end{align}
whence we compute the three spatial eigenvalues $\nu_i, i=1,2,3$ as\blue{
\begin{align}\label{e:slowevals}
\nu_{1}(\lambda,u;0)=\frac{c+\sqrt{c^2+4\lambda-4f'(u)}}{2}, \quad \nu_{2}(\lambda,u;0)=\frac{c-\sqrt{c^2+4\lambda-4f'(u)}}{2}, \quad \nu_3(\lambda,u;0)=-\frac{\lambda}{c}.
\end{align}}

From this, we see that in the region $\Re \lambda\geq0$, Hypotheses~\ref{h:osc} is satisfied only when \blue{$\Re (\nu_1)= \Re (\nu_2)$}, which occurs when 
\begin{align}
c^2+4\lambda-4f'(u)<0.
\end{align}
Using $c=1/\sqrt{2}$, and $f'(u)=2u-3u^2$, we find that for each fixed $\lambda \in [ 0,5/24 )$, this condition is satisfied for the interval
\begin{align}
u\in\left(\frac{1}{3}-\frac{\sqrt{10-48\lambda}}{12}, \frac{1}{3}+\frac{\sqrt{10-48\lambda}}{12}\right)\subset \left(0,2/3\right).
\end{align} 
At the endpoints of this interval, the eigenvalues \blue{$ \nu_{1,2}(\lambda,u;0)$} coincide at pinched double roots, or Airy points. We remark that for $\lambda=0$, these coincide precisely with the Airy points $u^\pm_{\mathrm{A},0}$ from the existence analysis (see~\S\ref{sec:pulse_geometry}), as expected.

The result is then is a direct consequence of the above discussion regarding the singular system, noting that the perturbed manifold $\mathcal{M}^m_\eps$ is $\mathcal{O}(\eps)$ close to $\mathcal{M}^m_0$.
\end{proof}
\begin{remark}
As the manifold $\mathcal{M}^m_\eps$ is a $C^1-\mathcal{O}(\eps)$-perturbation of the critical manifold $\mathcal{M}^m_0$, given as the graph $\{v=0, w=f(u), u\in(0,2/3)\}$, away from the fold points, the manifold $\mathcal{M}^m_\eps$ can also be represented as a graph, with a one-to-one correspondence between the $u$ and $w$ coordinates $w=f(u)+\mathcal{O}(\eps)$. In particular, we can equivalently identify the Airy points via their $w$-coordinates $w=w^\pm_\mathrm{A}(\lambda, \eps)$ along $\mathcal{M}^m_\eps$ where $w^\pm_\mathrm{A}(\lambda, 0) = f(u^-_\mathrm{A}(\lambda,0))$ and 
\begin{align}
w^-_\mathrm{A}(\lambda, \eps) = f(u^-_\mathrm{A}(\lambda,\eps))+\mathcal{O}(\eps), \qquad w^+_\mathrm{A}(\lambda, \eps) =f(u^+_\mathrm{A}(\lambda,\eps))+\mathcal{O}(\eps).
\end{align}
\end{remark}
An immediate consequence of Proposition~\ref{p:fhn_slowabs} is the following. Suppose the solution $\phi(\zeta;s,\eps)$ passes along any such portion of the middle slow manifold $\mathcal{M}^m_\eps$ for some interval $\zeta\in I_m$. Then by fixing the exponential weight $\eta(\zeta)=\frac{1}{2\sqrt{2}}+\mathcal{O}(\eps)$ for $\zeta\in I_m$, chosen so that the eigenvalues \blue{$\nu_{1,2}$} are shifted onto the imaginary axis (see Figure~\ref{f:fhn_allmiddle}), by Lemma~\ref{l:diagAiry} the system~\eqref{e:fhnlin1st} admits exponential trichotomies on $I_m$ with stable, unstable, and center trichotomy projections \blue{ $P^{\ss,\uu,\cc}(\zeta;\lambda,\eps)=P^{\ss,\uu,\cc}_\slow(\eps\zeta;\lambda,\eps)$, subspaces $E^{\ss,\uu,\cc}(\zeta;\lambda,\eps)=E^{\ss,\uu,\cc}_\slow(\eps\zeta;\lambda,\eps)$, and evolutions $\Phi^{\ss,\uu,\cc}(\zeta,\tilde{\zeta};\lambda,\eps)=\Phi^{\ss,\uu,\cc}_\slow(\eps\zeta;\lambda,\eps)$,  analytic in $\lambda$ and satisfying 
\begin{align}
\left| P^{\ss,\uu,\cc}(\zeta;\lambda,\eps) - \mathcal{P}^{\ss,\uu,\cc}(\zeta;\lambda,\eps)  \right| = \mathcal{O}(\eps)
\end{align}
where $\mathcal{P}^{\ss,\uu,\cc}(\zeta;\lambda,\eps)$ denote the spectral projections associated with the matrix $A(\zeta;\lambda,\eps)$. In particular, in this case, the subspace $E^{\ss}(\zeta;\lambda,\eps)$ is one-dimensional, with projection $P^{\ss}(\zeta;\lambda,\eps)$ which is $\mathcal{O}(\eps)$-close to the spectral projection $\mathcal{P}^{\ss}(\zeta;\lambda,\eps)$ onto the eigenspace $\mathcal{E}^{\ss}(\zeta;\lambda,\eps)$ associated with the eigenvalue $\nu_3$, while the subspace $E^{\cc}(\zeta;\lambda,\eps)$  is two-dimensional, with projection $P^{\cc}(\zeta;\lambda,\eps)$ which is $\mathcal{O}(\eps)$-close to the spectral projection $\mathcal{P}^{\cc}(\zeta;\lambda,\eps)$ onto the eigenspace $\mathcal{E}^{\cc}(\zeta;\lambda,\eps)$ associated with the eigenvalues $\nu_{1,2}$. The projection $P^{\uu}(\zeta;\lambda,\eps)$ is trivial.} The block-diagonalization formula~\eqref{e:genslowlinearDiag-old} is valid, with the particular form of the matrix $B_{11}$ valid near each of the Airy points $u=u^\pm_\mathrm{A}$, up to a reversal of time near the upper Airy point at $u=u_\mathrm{A}^+$.

\subsection{Accumulation of unstable eigenvalues for transitional pulses}\label{sec:fhnaccumulation}

Based on the discussion in~\S\ref{sec:slowabsmiddle}, and in particular Proposition~\ref{p:fhn_slowabs}, to determine which transitional pulses admit an accumulation of unstable eigenvalues due to the presence of slow absolute spectrum, we must determine which pulses pass along the absolutely unstable portion of the middle slow manifold $\mathcal{M}^m_\eps$. To this end, we briefly review some aspects of the existence analysis and geometric construction of the pulses, in particular their parametrization by $s$, followed by a statement of our main result on unstable eigenvalue accumulation, Theorem~\ref{thm:fhnaccumulation}.

\subsubsection{Parameterization of the transitional pulses}

In order to determine which of the transitional pulses pass along the absolutely unstable portion of the middle slow manifold $\mathcal{M}^m_\eps$, we first review some relevant details of their construction. It was shown rigorously in~\cite{CSbanana} that each transitional pulse can be constructed in essentially three pieces: a primary excursion, a secondary excursion, and an oscillatory tail. The primary excursion is the same for all transitional pulses; this trajectory traverses the right and left slow manifolds $\mathcal{M}^{\ell,r}_\eps$ and the two fast jumps $\phi_{f,b}$ between these in the sequence $\phi_f\to \mathcal{M}^r_\eps \to \phi_b \to\mathcal{M}^\ell_\eps$ -- see Figure~\ref{f:singular_limit} for a diagram of the singular limiting geometry. 

Next each transitional pulse completes a secondary excursion which traverses a canard trajectory near the origin, continuing along a portion of the repelling middle branch $\mathcal{M}^m_\eps$; the length of this portion depends on the specific transitional pulse under consideration (see Figure~\ref{f:slowabs_cases}). Pulses early on in the transition leave $\mathcal{M}^m_\eps$ via a fast jump $\phi_\ell$ to the left slow manifold $\mathcal{M}^\ell_\eps$, while further along the transition the secondary excursion continues all the way up $\mathcal{M}^m_\eps$ to the upper right fold point, and those later on in the transition leave $\mathcal{M}^m_\eps$ via a fast jump $\phi_r$ to the right slow manifold $\mathcal{M}^r_\eps$. Along the transition, it is this secondary excursion which essentially undergoes a canard explosion to complete the transition to a double pulse solution. 
\begin{remark}\label{r:parameterization}
The parameter $s$ in Theorem~\ref{thm:mainexistence} which parameterizes the transition for fixed $\eps$ is related to the height or $w$-coordinate of the fast jump $\phi_\ell$ or $\phi_r$ which is traversed. In~\cite{CSbanana}, the pulses are paramterized by $s\in(0,2w^\dagger)$, where $w^\dagger=4/27$ is the $w$-coordinate of the upper right fold point. Transitional pulses $\phi(\zeta;s,\eps)$ for $s\in(0,4/27)$ therefore leave $\mathcal{M}^m_\eps$ via the fast jump $\phi_\ell(\cdot;w)$ in the plane $w=s$, and those for $s\in(4/27,8/27)$ therefore leave $\mathcal{M}^m_\eps$ via a fast jump $\phi_r(\cdot; w)$ in the plane $w=8/27-s$. We note that near the endpoints $s=0,8/27$ and near the point $s=4/27$ where the pulses switch from ``left-jumping" to ``right-jumping", the relation between $s$ and the height $w$ of the corresponding fast jump is blurred, due to the nonhyperbolic dynamics near the fold points. However, for our purposes, this particular ambiguity will not cause any issues.
\end{remark}
Finally the secondary excursion is followed by an oscillatory tail, which is contained in a two-dimensional normally repelling center-stable manifold $\mathcal{Z}_\eps$ of the fixed point at the origin. 

Sample pulses are depicted in Figure~\ref{f:transition} which shows results of numerical continuation along the single-to-double pulse transition. Each pulse completes the sequence $I,II,III,IV$ which comprises the primary excursion, followed by a colored trajectory which passes near the origin along a canard trajectory, and then traverses a portion of the middle branch, before jumping off to either the right or left branch; this corresponds to the secondary excursion. Finally, each pulse eventually enters the locally invariant center-stable manifold $\mathcal{Z}_\eps$ of the origin, containing the oscillatory tails. Schematic diagrams of these different pulses are also depicted in Figure~\ref{f:slowabs_cases}, with colors chosen to match those from the numerical plots in Figure~\ref{f:transition}. 

\begin{figure}
\begin{subfigure}{.48 \textwidth}
\centering
\includegraphics[width=1\linewidth]{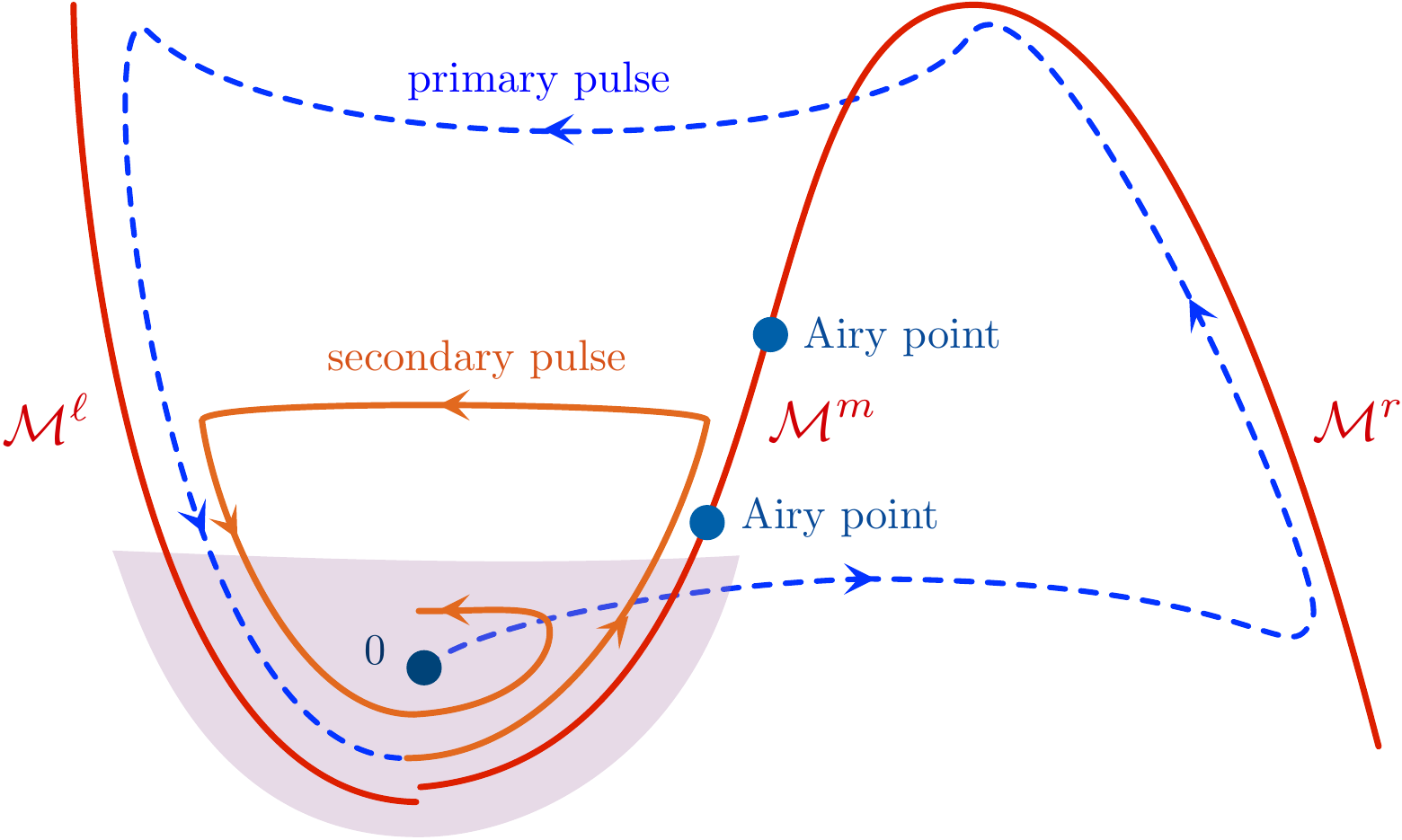}
\caption{Transitional pulse which jumps to $\mathcal{M}^\ell_\eps$ beyond the first Airy point $u=u^-_\mathrm{A}(\lambda, \eps)$ but before reaching the second Airy point at $u=u^+_\mathrm{A}(\lambda,\eps)$.}
\label{f:fhn_left}
\end{subfigure}
\hspace{0.04 \textwidth}\begin{subfigure}{.48 \textwidth}
\centering
\includegraphics[width=1\linewidth]{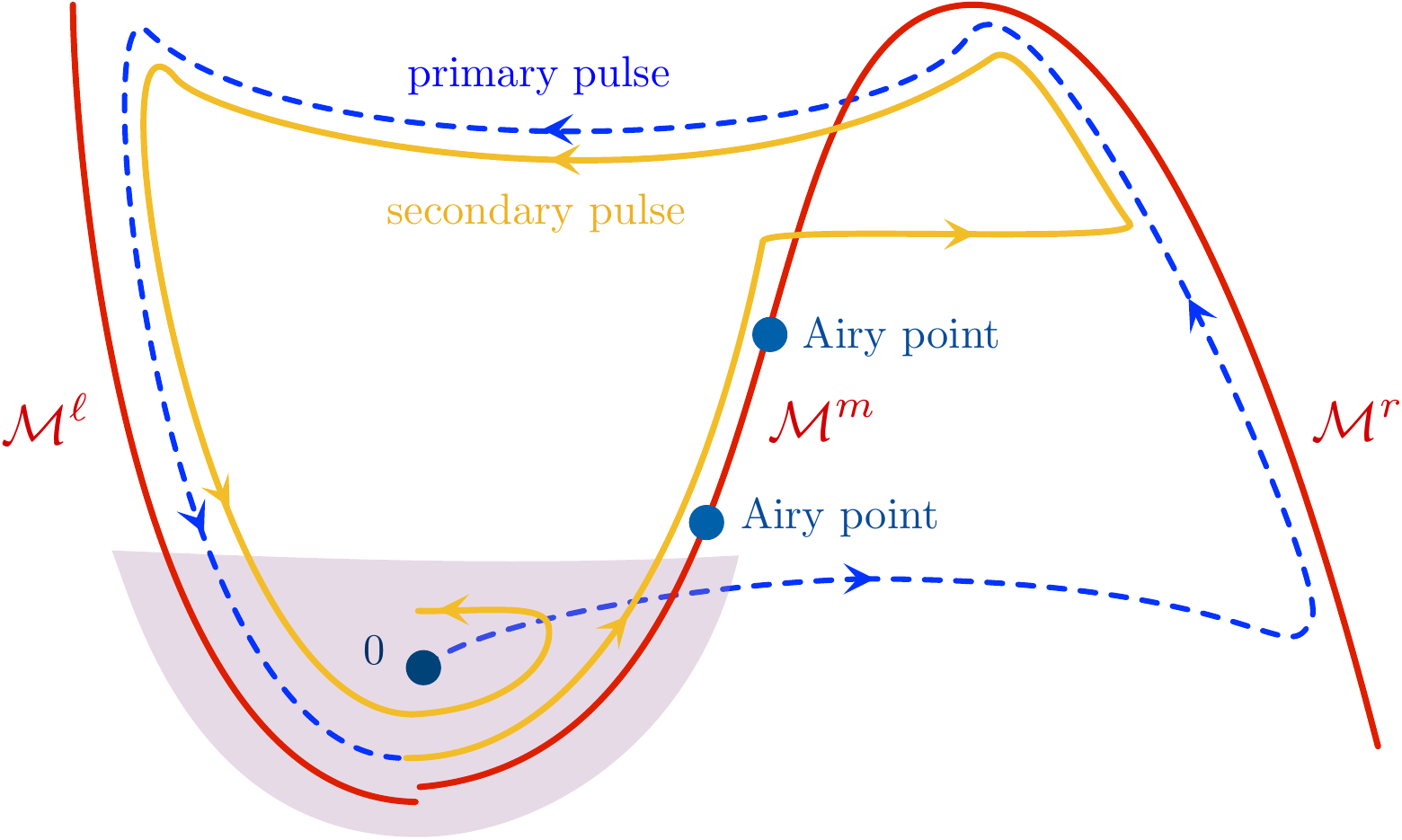}
\caption{Transitional pulse which passes through both Airy points $u=u^\pm_\mathrm{A}(\lambda,\eps)$.}
\label{f:fhn_doubleairy}
\end{subfigure}
\vspace{10pt} \begin{subfigure}{.48 \textwidth}
\centering
\includegraphics[width=1\linewidth]{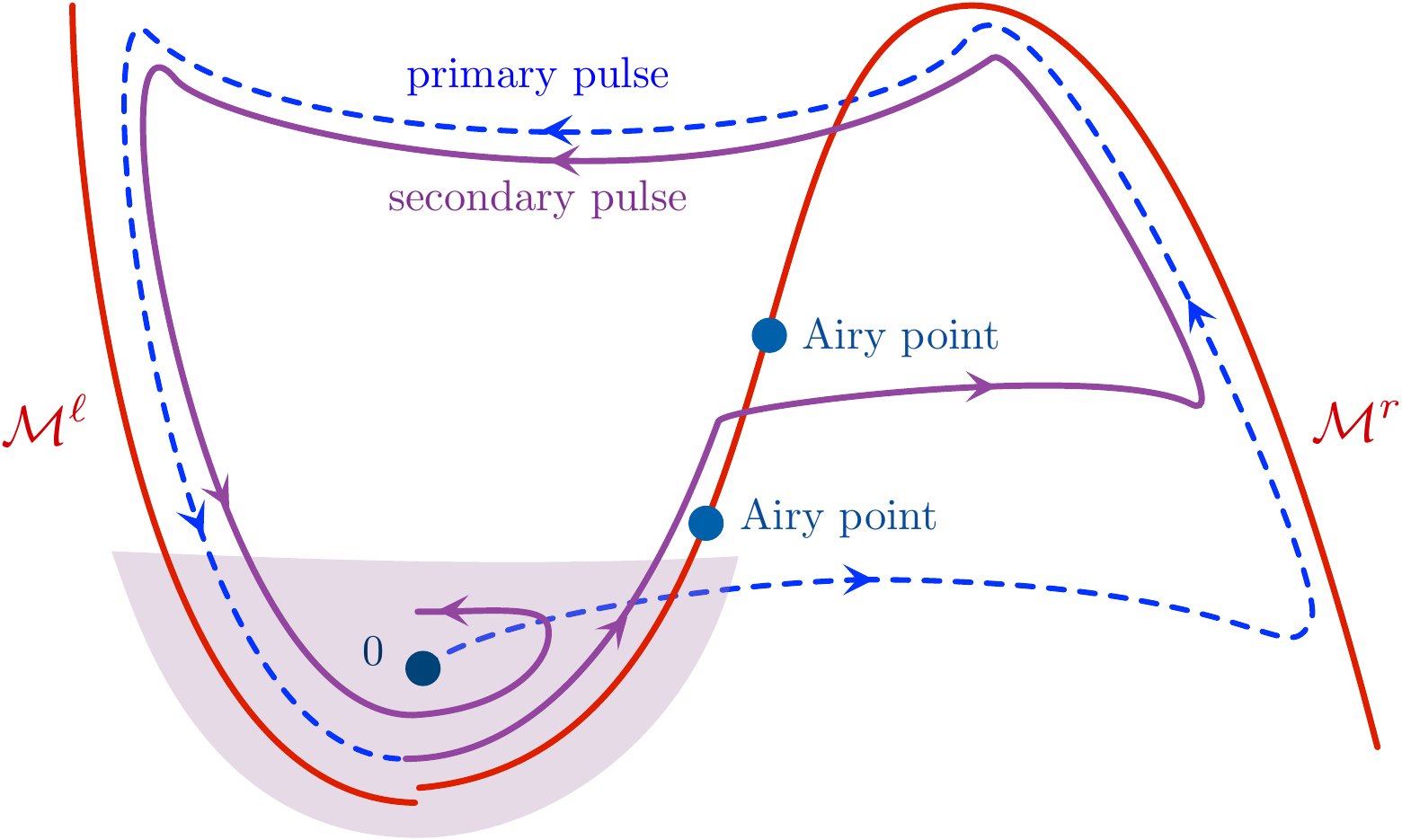}
\caption{Transitional pulse which jumps to $\mathcal{M}^r_\eps$ beyond the first Airy point $u=u^-_\mathrm{A}(\lambda, \eps)$ but before reaching the second Airy point at $u=u^+_\mathrm{A}(\lambda,\eps)$.}
\label{f:fhn_right}
\end{subfigure}
\hspace{0.04 \textwidth}
\begin{subfigure}{.48 \textwidth}
\centering
\includegraphics[width=1\linewidth]{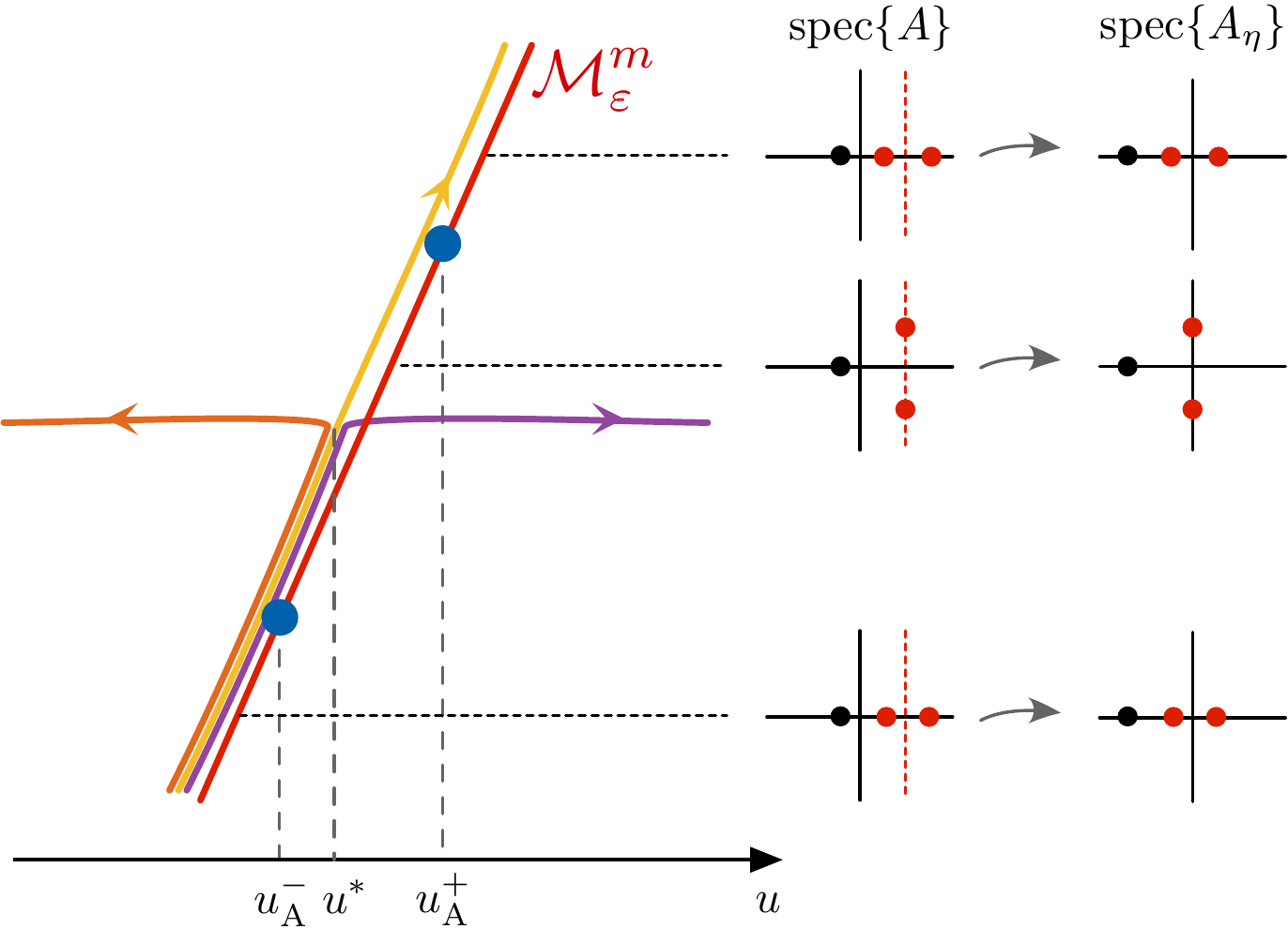}
\caption{Schematic depicting the behavior of each of the pulses (a), (b), (c), along the middle slow manifold $\mathcal{M}^m_\eps$. Also shown is the spatial eigenvalue structure of the matrices $A,A_\eta$ along this portion of $\mathcal{M}^m_\eps$.}
\label{f:fhn_allmiddle}
\end{subfigure}
\caption{Shown are the different possible schematic behaviors for pulses which pass near the region of the middle slow manifold $\mathcal{M}^m_\eps$ which admits slow absolute spectrum relative to the rest state $(u,w)=(0,0)$ for values of $\lambda \in (0,5/24)$. Compare the geometry of the pulses with the numerically computed profiles from Figure~\ref{f:transition}.}
\label{f:slowabs_cases}
\end{figure}


\subsubsection{Accumulation of unstable eigenvalues}
For a fixed transitional pulse $\phi(\cdot;s,\eps)$ for $s\in(0,8/27)$, we determine the corresponding unstable slow absolute spectrum $\Sigma^\mathrm{slow}_\mathrm{abs}(s)$ as follows. For fixed $\lambda\in(0,5/24)$, this pulse will traverse a portion of the middle slow manifold $\mathcal{M}^m_\eps$ between the two Airy points $w=w^\pm_\mathrm{A}(\lambda,\eps)$ provided $s\in(w^-_\mathrm{A}(\lambda,\eps),8/27-w^-_\mathrm{A}(\lambda,\eps))$. 

Equivalently, if we track the locations of the Airy points as a function of $\lambda$ for $\eps=0$, we see that for $\lambda=0$, the Airy points are furthest apart at 
\begin{align}
w^\pm_\mathrm{A}(0,0) = \frac{64\pm19\sqrt{10}}{864},
\end{align}
and approach each other for $\lambda\in(0,5/24)$, before finally colliding and disappearing at $w^\pm_\mathrm{A}(5/24,0) = 2/27$. Thus for values of $s\in(0,w^-_\mathrm{A}(0,0))$, in the singular limit the corresponding transitional pulses do not follow any absolutely unstable portion of the middle branch $\mathcal{M}^m_0$; by a similar reasoning, we conclude the same for values of $s\in(8/27-w^-_\mathrm{A}(0,0),8/27)$.

For values of $s\in(w^-_\mathrm{A}(0,0),8/27-w^-_\mathrm{A}(0,0))$, some absolutely unstable portion of $\mathcal{M}^m_0$ is traversed. For $s\in(2/27,2/9)$, the corresponding transitional pulse traverses some portion of $\mathcal{M}^m_0$ between the two Airy points for \emph{any} $\lambda\in(0,5/24)$. For $s\in(w^-_\mathrm{A}(0,0),2/27)$, the pulse traverses an absolutely unstable portion of $\mathcal{M}^m_0$ provided $\lambda \in (0,\lambda_{\mathrm{max},\ell}(s))$ where 
\begin{align}
\lambda_{\mathrm{max},\ell}:[w^-_\mathrm{A}(0,0),2/27]\to [0,5/24]
\end{align}
denotes the unique value of $\lambda$ such that $s=w^-_\mathrm{A}(\lambda_{\mathrm{max},\ell}(s),0)$; note that $\lambda_{\mathrm{max},\ell}(s)$ is an increasing function of $s$ with $\lambda_{\mathrm{max},\ell}(w^-_\mathrm{A}(0,0))=0$ and $\lambda_{\mathrm{max},\ell}(2/27)=5/24$. Similarly for $s\in(2/9, 8/27-w^-_\mathrm{A}(0,0))$, the pulse traverses an absolutely unstable portion of $\mathcal{M}^m_0$ provided $\lambda \in (0,\lambda_{\mathrm{max},r}(s))$ where
\begin{align}
\lambda_{\mathrm{max},r}:[2/9,8/27-w^-_\mathrm{A}(0,0)]\to [0,5/24]
\end{align}
 is the unique value of $\lambda$ such that $s=8/27-w^-_\mathrm{A}(\lambda_{\mathrm{max},r}(s),0)$.  Note that by symmetry, we have 
 \begin{align}
\lambda_{\mathrm{max},\ell}(s) = \lambda_{\mathrm{max},r}(8/27-s)
\end{align}
for $s\in(w^-_\mathrm{A}(0,0),2/27)$.

We therefore conclude that for values of $s\in(0,w^-_\mathrm{A}(0,0))\cup(8/27-w^-_\mathrm{A}(0,0),8/27)$, the corresponding transitional pulses do not follow any absolutely unstable portion of the middle branch $\mathcal{M}^m_0$. However, for $s\in(w^-_\mathrm{A}(0,0),8/27-w^-_\mathrm{A}(0,0))$, the corresponding transitional pulse will traverse a portion of the middle slow manifold $\mathcal{M}^m_\eps$ between the two Airy points \emph{provided} $\eps$ is sufficiently small and $\lambda\in\Sigma^\mathrm{slow}_\mathrm{abs}(s):= (0,\lambda_\mathrm{max}(s))$ where $\lambda_\mathrm{max}(s)$ is defined as
\begin{align}\label{e:lambda_max}
\lambda_\mathrm{max}(s) = \begin{cases} \lambda_{\mathrm{max},\ell}(s) & s\in (w^-_\mathrm{A}(0,0),2/27) \\ 5/24 & s\in [2/27,2/9]\\ \lambda_{\mathrm{max},r}(s) & s\in (2/9,8/27-w^-_\mathrm{A}(0,0)) \end{cases},
\end{align}
and satisfies the symmetry $\lambda_\mathrm{max}(s) = \lambda_\mathrm{max}(8/27-s)$ for $s\in(w^-_\mathrm{A}(0,0),8/27-w^-_\mathrm{A}(0,0))$. We have the following (see Figure~\ref{f:lambda_max}).

\begin{figure}
\centering
\includegraphics[width=0.6\textwidth]{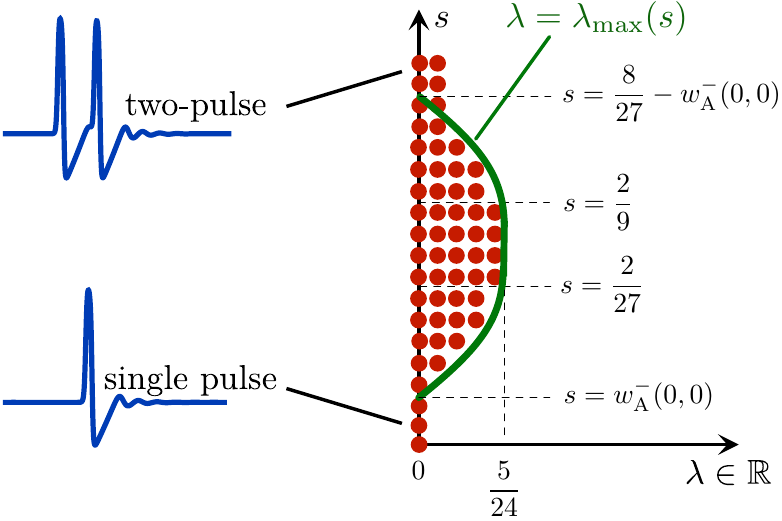}
\caption{Depicted are the results of Theorem~\ref{thm:fhnaccumulation}. For fixed $0<\eps\ll1$, as the single pulse transitions to a double pulse along the homoclinic banana as $s$ varies from $s=0$ to $s=8/27$, $\mathcal{O}(1/\eps)$ eigenvalues accumulate on the positive real axis on the interval $(0,\lambda_\mathrm{max}(s))$ where the endpoint $\lambda_\mathrm{max}$ is defined by~\eqref{e:lambda_max}. At the end of the transition, numerical evidence suggests that a single positive $\mathcal{O}(\eps)$ eigenvalue remains~\cite{CSosc}, and the resulting $2$ pulse is unstable. A translational eigenvalue at $\lambda=0$ is present throughout the transition. }
\label{f:lambda_max}
\end{figure}

\begin{theorem}\label{thm:fhnaccumulation}
Consider the eigenvalue problem~\eqref{e:fhnstab} associated with a transitional pulse $\phi(\zeta; s, \eps)$ of the FitzHugh--Nagumo PDE~\eqref{e:fhn} for $s\in(w^-_\mathrm{A}(0,0),8/27-w^-_\mathrm{A}(0,0))$ and $\lambda=\lambda^* \in \Sigma^\mathrm{slow}_\mathrm{abs}(s)$, and fix any $\delta_\lambda>0$.  For all sufficiently small $\eps>0$, the eigenvalue problem~\eqref{e:fhnstab} admits $\mathcal{O}(\eps^{-1})$ eigenvalues in the interval $(\lambda^*-\delta_\lambda, \lambda^*+\delta_\lambda)$. 
\end{theorem}
Theorem~\ref{thm:fhnaccumulation} guarantees the accumulation of unstable eigenvalues on the set $\Sigma^\mathrm{slow}_\mathrm{abs}(s)$. We note that the length of this interval is initially an increasing function of $s$ as pulses traverse longer portions of $\mathcal{M}^m_\eps$, reaching a maximal set $\Sigma^\mathrm{slow}_\mathrm{abs}(s)=(0,5/24)$ for $s\in [2/27,2/9]$, before eventually decreasing in $s$. Therefore, for pulses which traverse longer portions of $\mathcal{M}^m_\eps$, we expect to see a larger accumulation set of eigenvalues, as well as larger maximum eigenvalues.

In order to apply the general theory from~\S\ref{s:gentheory} to prove Theorem~\ref{thm:fhnaccumulation}, we must verify the relevant boundary conditions satisfied by the relevant entry/exit spaces along the critical region of the absolutely unstable slow manifold $\mathcal{M}^m_\eps$. These conditions are examined in detail in~\S\ref{sec:boundaryverification}, followed by a brief conclusion of the proof of Theorem~\ref{thm:fhnaccumulation} in~\S\ref{sec:fhnaccumulationproof}.

\subsection{Verification of boundary condition hypotheses}\label{sec:boundaryverification}

Based on the discussion in the previous section, we expect that unstable eigenvalues should accumulate for any transitional pulse which traverses a portion of the middle branch $\mathcal{M}^m_\eps$ between the Airy points. It remains to determine the appropriate boundary condition spaces and verify the entry/exit conditions in Hypotheses~\ref{h:osc_vector}--\ref{h:bl_vector}, where different hypotheses will be relevant depending on the pulse being considered; for this we present a sketch of the argument, referring to~\cite{CdRS,CSbanana} for techniques which provide the details, without repeating the content of these prior works.

We first note that all transitional pulses which traverse the critical region of the middle slow manifold $\mathcal{M}^m_\eps$ between the two Airy points do so by entering this region via an Airy transition through the lower Airy point at $u=u_\mathrm{A}^-$. However, the manner of exit differs depending on exactly which pulse is being considered. Pulses during the early part of this transition  jump to the left slow manifold $\mathcal{M}^\ell_\eps$ before reaching the higher Airy point at $u=u_\mathrm{A}^+$, and those during the later part of the transition jump to the right slow manifold $\mathcal{M}^r_\eps$ before reaching the higher Airy point; in each of these cases, the manner of exit is therefore a \emph{fast layer}. In contrast, the pulses during the middle part of the transition spend the longest time along $\mathcal{M}^m_\eps$ and actually pass through the second, higher Airy point, and therefore exit the critical region through a second Airy transition; see Figure~\ref{f:slowabs_cases}.

In each of these cases, we are concerned with eigenvalues, i.e. values of $\lambda$ to the right of $\Sigma_\mathrm{ess}$ such that the eigenvalue problem~\eqref{e:fhnlin1st} admits an exponentially localized solution. The asymptotic matrix $A_{\pm\infty}(\lambda,\eps)$ has a two-dimensional stable eigenspace $Q^+_\infty(\lambda,\eps)$, and a one-dimensional unstable eigenspace $Q^-_\infty(\lambda, \eps)$. Therefore, the space of solutions of~\eqref{e:fhnlin1st} which decay as $\zeta\to\infty$ is two-dimensional; at a given value of $\zeta\in\mathbb{R}$, we denote this subspace by $Q^+(\zeta;\lambda,\eps)$, and we note that $Q^+(\zeta;\lambda,\eps)\to Q^+_\infty(\lambda,\eps)$ as $\zeta\to\infty$, \blue{where we say that a sequence of subspaces converges if their unit spheres converge in the symmetric Hausdorff distance}. Similarly the space of solutions which decay as $\zeta\to-\infty$ is one-dimensional, and we denote this space by $Q^-(\zeta;\lambda,\eps)$, noting that $Q^-(\zeta;\lambda,\eps)\to Q^-_\infty(\lambda, \eps)$ as $\zeta\to-\infty$. Eigenfunctions then correspond to intersections of these two subspaces, and we will construct unstable eigenvalue/eigenfunction pairs which arise from the slow absolute spectrum phenomenon by tracking the subspace $Q^-(\zeta;\lambda,\eps)$ (resp. $Q^+(\zeta;\lambda,\eps)$) to the corresponding entry (resp. exit) point along the absolutely unstable middle slow manifold $\mathcal{M}^m_\eps$, and showing how the resulting boundary value problem falls into the framework established in~\S\ref{s:slowabsnew}-\ref{s:gentheory}.

\subsubsection{Tracking the entry subspace $Q^-(\zeta;\lambda,\eps)$}\label{s:fhn_entryspace}

We now argue that the asymptotic boundary subspaces can be tracked to the endpoints of this critical region for each of the transitional pulses, and that the appropriate boundary condition hypotheses can be verified in each case. The incoming boundary subspace $Q^-(\zeta;\lambda,\eps)$ can be tracked similarly for all transitional pulses. This subspace must be tracked along the primary excursion, through the canard trajectory near the origin and up along the middle branch $\mathcal{M}^m_\eps$. Fixing a value of $\lambda \in (0,5/24)$, we consider a transitional pulse which jumps from the middle branch $\mathcal{M}^m_\eps$ at some value of $w$ beyond the Airy point, that is, $w>w^-_\mathrm{A}(\lambda,\eps)$. We then track the subspace $Q^-(\zeta;\lambda,\eps)$ to a location just below the Airy point.

\begin{proposition}\label{p:fhn_entry}
Consider the eigenvalue problem~\eqref{e:fhnstab} associated with a transitional pulse $\phi(\zeta; s, \eps)$ of the FitzHugh--Nagumo PDE~\eqref{e:fhn} for $s\in(w^-_\mathrm{A}(\lambda^*,0),8/27-w^-_\mathrm{A}(\lambda^*,0))$ where $\lambda^* \in (0, 5/24)$. Fix $\delta>0$, and fix a translate of the pulse such that $\zeta=0$ is the $\zeta$ value at which the pulse $\phi(\zeta; s, \eps)$ first reaches the section $\{w=w^-_\mathrm{A}(\lambda^*,0)-\delta\}$ after completing the primary excursion. \red{Then  there exists $\epsilon_0>0$ such that for all $\epsilon\in(0,\epsilon_0]$  the subspace $Q^-(0;\lambda^*,\eps)$ satisfies Hypothesis~\ref{h:osc_vector}.}
\end{proposition}

In order to prove Proposition~\ref{p:fhn_entry}, we call on results regarding the behavior of the primary excursion, or primary pulse. In particular, it was shown in~\cite{CdRS} that for any $\lambda$ with $\Re(\lambda)>0$, it is possible to construct exponential dichotomies along the primary excursion, independently of $\eps>0$ sufficiently small. This is (part of) the argument which guarantees the stability of the $1$-pulses in~\cite{CdRS}. In particular, any eigenvalues which arise for the primary pulse must be nearby eigenvalues of the two fast jumps $\phi_{f,b}$ in the singular limit. Along the entire primary excursion, we fix the weight $\eta(\zeta)=\eta_0$, where $\eta_0$ is the constant from Proposition~\ref{p:leftrightsplit}.

Setting $\eps=0$ in the eigenvalue problem~\eqref{e:fhnlin1st_weighted}, yields the limiting eigenvalue problems
\begin{equation}\label{e:fhnlin1st_0}
\begin{aligned}
\dot{U} = A_{\eta,j}(\zeta;\lambda, 0)U, \qquad 
A_{\eta,j}(\zeta;\lambda, \eps):=\begin{pmatrix}-\eta_0 &1& 0 \\ \lambda-f'(u_j(\zeta)) & c-\eta_0 & 1  \\ 0 & 0 & -\frac{\lambda}{c}-\eta_0 \end{pmatrix}
\end{aligned}
\end{equation}
along the fast jumps $\phi_j(\zeta) = (u_j,u'_j)(\zeta), j=f,b$, which are solutions to the layer problem~\eqref{e:layer} for the values $w=0$ in the case of $\phi_f$, and $w=4/27$ in the case of $\phi_b$. The upper triangular block system obtained by restricting this system to the subspace spanned by the first two components is precisely the system one obtains by considering the reduced eigenvalue problem for the fast jumps in the layer problem~\eqref{e:layer}
\begin{equation}
\begin{aligned}\label{e:fhnlin1st_reduced}
\dot{Y} = C_{\eta, j}(\zeta;\lambda, 0)Y, \qquad
C_{\eta,j}(\zeta;\lambda, 0):=\begin{pmatrix}-\eta_0 &1 \\ \lambda-f'(u_j(\zeta)) & c-\eta_0 \end{pmatrix}, \qquad j=f,b, 
\end{aligned}
\end{equation}
where $Y=(u,v)$.

\blue{The full system~\eqref{e:fhnlin1st_0} can be solved from the lower dimensional reduced system~\eqref{e:fhnlin1st_reduced} using variation of constants formulae. The two-dimensional eigenvalue problems~\eqref{e:fhnlin1st_reduced} are of Sturm-Liouville-type and admit no eigenvalues in the right half plane, and therefore exponential dichotomies with $\eps$-independent constants can be constructed for the reduced systems~\eqref{e:fhnlin1st_reduced}, which can be extended to the full systems~\eqref{e:fhnlin1st_0} using variation of constants formulae~\cite[Proposition 6.18]{CdRS}, and finally extended to~\eqref{e:fhnlin1st_weighted} using roughness.

Along the right and left slow manifolds, the weighted eigenvalue problem has slowly varying coefficients, and by Proposition~\ref{p:leftrightsplit}, the spatial eigenvalues $\nu_i, i=1,2,3$ admit a consistent splitting with two eigenvalues of negative real part, and one of positive real part. It is then straightforward using classical methods~\cite{Cop78} to construct exponential dichotomies for~\eqref{e:fhnlin1st_weighted} along these slow manifolds, also independently of $\eps>0$ sufficiently small~\cite[Proposition 6.5]{CdRS}. 

Taken together, these fast and slow pieces divide the primary excursion into four disjoint (sub)intervals, along each of which the system admits $\eps$-independent exponential dichotomies. We now make use of the following lemma.
\begin{lemma}\cite{Cop78,bdrthesis}\label{lem:paste}
Let $-\infty\leq \zeta_0<\zeta_1<\zeta_2\leq \infty$ and suppose the system
\begin{align}\label{e:pasteeqn}
U_\zeta=A(\zeta)U
\end{align}
admits exponential dichotomies on the intervals $J_1=[\zeta_0,\zeta_1]$ and $J_2=[\zeta_1,\zeta_2]$, with constants $C,\mu$, projections $P^{\mathrm{u,s}}_i(\zeta),i=1,2$, and subspaces $E^{\mathrm{u,s}}_i(\zeta),i=1,2$. If $E^\mathrm{u}_1(\zeta_1)\oplus E^\mathrm{s}_2(\zeta_1)=\mathbb{C}^n$, then~\eqref{e:pasteeqn} admits exponential dichotomies on $J=J_1\cup J_2$ with constants $\tilde{C},\mu$, where $\tilde{C}$ depends only on $C$ and $\|P\|$, where $P$ is the projection onto $E^\mathrm{s}_2(\zeta_1)$ along $E^\mathrm{u}_1(\zeta_1)$.
\end{lemma}
\begin{proof}
As in the lemma statement, we define $P$ to be the projection onto $E^\mathrm{s}_2(\zeta_1)$ along $E^\mathrm{u}_1(\zeta_1)$, and we set $P^\mathrm{s}(\zeta)=\Phi(\zeta,0)P\Phi(0,\zeta)$ and $P^\mathrm{u}(\zeta)=1-P^\mathrm{s}(\zeta)$, where $\Phi(\zeta,\bar{\zeta})$ is the evolution operator associated with~\eqref{e:pasteeqn}. It is straightforward to verify that this results in an exponential dichotomy with the desired properties; see~\cite[Lemma 4.11]{bdrthesis}.
\end{proof}

In particular, Lemma~\ref{lem:paste} applies if the projections associated with exponential dichotomies on neighboring intervals are close to each other; then these dichotomies can be pasted to form exponential dichotomies on the union of these intervals. Using this, the exponential dichotomies along the four pieces of the primary pulse can then be pasted together to build exponential dichotomies along the entire primary pulse. This is the content of the following.

\begin{proposition}\label{p:fhn_primary_ed}
Fix any $M>0$, and any sufficiently small $\delta>0$. Consider the eigenvalue problem~\eqref{e:fhnlin1st_weighted} for a transitional pulse $\phi(\zeta;s,\eps)$ for $s\in(0,8/27)$, and let $\zeta=\zeta_0$ denote the $\zeta$-value at which the pulse passes through the section $\{u=0\}$ after completing the primary excursion. There exists $\eps_0>0$ such that for any $\eps\in(0,\eps_0)$, the system~\eqref{e:fhnlin1st_weighted} admits exponential dichotomies on $(-\infty, \zeta_0]$ with constants $C,\mu$ independent of $\epsilon\in(0,\epsilon_0)$ and $\lambda$ satisfying $\delta<|\lambda|<M$ and $\Re\lambda>-\delta$.
\end{proposition}
\begin{proof}
As described above, the primary excursion can be decomposed into four sub-intervals (two fast jumps and two portions of slow manifolds), and along each of these intervals~\eqref{e:fhnlin1st_weighted} admits exponential dichotomies with $\eps$-independent constants. The analysis in~\cite[Proposition 6.20]{CdRS} shows that at the endpoints of each neighboring subinterval, the projections associated with the neighboring exponential dichotomies are close. Then Lemma~\ref{lem:paste} allows these exponential dichotomies to be pasted, resulting in exponential dichotomies along the entire primary excursion.
\end{proof}
}

We denote the corresponding projections by $P^\mathrm{s,u}_\mathrm{p}(\zeta;\lambda,\eps)$ and subspaces by $E^\mathrm{s,u}_\mathrm{p}(\zeta;\lambda,\eps)$. Using this, we can now complete the proof of Proposition~\ref{p:fhn_entry}.

\begin{proof}[Proof of Proposition~\ref{p:fhn_entry}]
We argue that the dichotomies guaranteed by Proposition~\ref{p:fhn_primary_ed} along the primary excursion can be extended for the transitional pulses which pass near the equilibrium along a canard trajectory and continue up the middle branch $\mathcal{M}^m_\eps$. Along the middle branch, below the first Airy point $w=w^-_\mathrm{A}(\lambda,\eps)$, it still holds that \blue{$\Re \nu_1>\Re\nu_2$. Below the section $\{w=w^-_\mathrm{A}(\lambda^*,0)-\delta\}$, i.e. a small $\mathcal{O}(1)$-distance away from the Airy point, this spectral gap is uniform in all $\eps$ sufficiently small.} Hence by possibly shifting $\zeta_0$ and adjusting the exponential weight $\eta(\zeta)$ on the interval $[\zeta_0,0]$, along this portion of $\mathcal{M}^m_\eps$, the eigenvalue problem~\eqref{e:fhnlin1st_weighted} has slowly varying coefficients and admits exponential dichotomies with constants $C,\mu$ independent of $\eps>0$ sufficiently small, with projections close to the spectral projections of the matrix $A_\eta$. The \blue{(weighted)} stable subspace is two-dimensional and close to the eigenspace corresponding to the eigenvalues \blue{$\nu_{2,3}$} while the unstable subspace is one-dimensonal and aligned $\mathcal{O}(\eps)$-close to the eigenspace corresponding to the eigenvalue \blue{$\nu_1$}. 

\blue{Using Lemma~\ref{lem:paste},} these dichotomies can be pasted to those coming from the primary excursion, and hence we continue to denote the projections and spaces by $P^\mathrm{s,u}_\mathrm{p}(\zeta;\lambda,\eps)$ and $E^\mathrm{s,u}_\mathrm{p}(\zeta;\lambda,\eps)$, respectively, on the full interval $\zeta\in \mathbb{R}^-$. This gives a means of tracking the one-dimensional subspace $Q^-(\zeta;\lambda,\eps)$ along the middle branch up to a neighborhood of the Airy point. In particular, since $Q^-$ is the subspace of solutions which decay as $\zeta\to-\infty$, $Q^-(\zeta;\lambda,\eps)$ is aligned in the direction of the unstable subspace $E^\mathrm{u}_\mathrm{p}(\zeta;\lambda,\eps)$. Furthermore, as the dichotomy projections $P^\mathrm{s,u}_\mathrm{p}(\zeta;\lambda,\eps)$ are close to the spectral projections of $A_\eta$, $Q^-(\zeta;\lambda,\eps)$ is therefore aligned close to the eigenspace associated with the eigenvalue \blue{$\nu_1$} upon approaching a neighborhood of the Airy point.

From here, we see that the conditions of Hypothesis~\ref{h:osc_vector} are satisfied.  The exponential dichotomies above can be constructed with constants independent of $\eps>0$ sufficiently small, and since the projections are tied to the spectral projections of the matrix $A_\eta$ along the slow manifolds $\mathcal{M}^{\ell,m,r}_\eps$, these depend in a regular fashion on $\lambda$. \red{Along the middle slow manifold, the projection $P^\mathrm{c}(\zeta;\lambda, \eps)$ is $\mathcal{O}(\eps)$-close to the spectral projection onto the eigenspace $\mathcal{E}^{\cc}(\zeta;\lambda,\eps)$ associated with the eigenvalues $\nu_{1,2}$ and is therefore also close to the projection $P^\mathrm{u}_\mathrm{p}(\zeta;\lambda,\eps)$, while the projection $P^\mathrm{ss}(\zeta;\lambda, \eps)$ is similarly close to $P^\mathrm{s}_\mathrm{p}(\zeta;\lambda,\eps)$. The one-dimensional subspace $Q^-$ projects nontrivially onto the one-dimensional eigenspace associated with the eigenvalue $\nu_1$, and thus also the two-dimensional subspace $E^{\cc}(\zeta;\lambda,\eps)$, which is itself $\mathcal{O}(\eps)$-close to $\mathcal{E}^{\cc}(\zeta;\lambda,\eps)$.}
\end{proof}

\subsubsection{Tracking the exit subspace $Q^+(\zeta;\lambda,\eps)$}\label{sec:exittracking}
Tracking the exit subspace is more challenging, due to the fact that the manner of exit from the slow manifold $\mathcal{M}^m_\eps$ depends on which transitional pulse is being considered. However, we identify three primary cases for the fate of the pulse as it passes along $\mathcal{M}^m_\eps$; see Figure~\ref{f:slowabs_cases} for a visualization of the different cases. The first (see Figure~\ref{f:fhn_left}) is that it departs along a fast jump to the left slow manifold $\mathcal{M}^\ell_\eps$ before reaching the second Airy point at $w=w^+_\mathrm{A}(\lambda, \eps)$; the second (see Figure~\ref{f:fhn_doubleairy}) is that the pulse continues along the middle branch beyond the second Airy point, and the third (see Figure~\ref{f:fhn_right}) is that the pulse jumps to the right slow manifold $\mathcal{M}^r_\eps$ before reaching $w=w^+_\mathrm{A}(\lambda, \eps)$.
\paragraph{Case 1:}
We focus on the first case -- see Figure~\ref{f:fhn_left}. In this case, the pulse departs $\mathcal{M}^m_\eps$ via the fast jump $\phi_\ell(\zeta;s)$ to the left slow manifold $\mathcal{M}^\ell_\eps$. The subspace $Q^+(\zeta;\lambda,\eps)$ corresponds to the two-dimensional space of solutions which decay as $\zeta\to\infty$. This space must be tracked backwards along the tail of the pulse, contained in the invariant two-dimensional center-stable manifold of the origin, then tracked up the left slow manifold $\mathcal{M}^\ell_\eps$, and back across the corresponding fast jump $\phi_\ell(\cdot; s)$ from $\mathcal{M}^m_\eps$. We have the following.

\begin{proposition}\label{p:fhn_exit1}
Consider the eigenvalue problem~\eqref{e:fhnstab} associated with a transitional pulse $\phi(\cdot; s, \eps)$ of the FitzHugh--Nagumo PDE~\eqref{e:fhn} for $s\in(w^-_\mathrm{A}(\lambda^*,0),w^+_\mathrm{A}(\lambda^*,0))$ where $\lambda^* \in (0, 5/24)$. For all sufficiently small $\eps>0$, let $\zeta = T_\eps$ denote the $\zeta$ value at which the pulse $\phi(\zeta; s, \eps)$ first reaches the section $\{u=0\}$ after jumping from the absolutely unstable middle slow manifold $\mathcal{M}^m_\eps$. \red{Then there exists $\epsilon_0>0$ such that for all $\epsilon\in(0,\epsilon_0]$,~\eqref{e:fhnstab} satisfies Hypothesis~\ref{h:bl} at $\zeta=T_\eps$, and the subspace $Q^+(T_\eps;\lambda^*,\eps)$ satisfies Hypothesis~\ref{h:bl_vector}.}
\end{proposition}
\begin{proof}
To track the subspace $Q^+(\zeta;\lambda,\eps)$ backwards from the equilibrium $(u,w)=(0,0)$ along the tail of the pulse, we argue as follows. This oscillatory tail is contained in the two dimensional invariant center stable manifold $\mathcal{Z}_\eps$ of the equilibrium. Within a small neighborhood of the equilibrium, the system~\eqref{e:fhnlin1st_weighted} has slowly varying coefficients (due to the proximity to the equilibrium) and admits exponential dichotomies with constants $C,\mu$ independent of $\eps>0$ sufficiently small, with projections close to the spectral projections of the matrix $A_\eta$. We denote these projections and spaces by $P^\mathrm{s,u}_\mathrm{tail}(\zeta;\lambda,\eps)$ and $E^\mathrm{s,u}_\mathrm{tail}(\zeta;\lambda,\eps)$, respectively. As the subspace $Q^+(\zeta;\lambda,\eps)$ corresponds to the two-dimensional space of solutions which decay as $\zeta\to\infty$, we have that $Q^+(\zeta;\lambda,\eps)=E^\mathrm{s}_\mathrm{tail}(\zeta;\lambda,\eps)$. 

Outside of this neighborhood, the system~\eqref{e:fhnlin1st_weighted} no longer has consistently slowly varying coefficients; this is only the case along the slow manifolds $\mathcal{M}^\ell_\eps$ and $\mathcal{M}^m_\eps$. Hence in order to extend the exponential dichotomies for the tail of the pulse, we must examine the slow/fast structure of the tail. 

The tail is contained entirely in the two-dimensional extended center manifold $\mathcal{Z}_\eps$. The manifold $\mathcal{Z}_\eps$ is formed as a perturbation of the set formed by the union of all of the singular heteroclinic orbits $\mathcal{Z}_0:=\{\phi_\ell(\cdot; w):w\in(0,w^-_\mathrm{A}(0,0)\}$. Since the middle branch $\mathcal{M}^m_0$ has the structure of a repelling node for each $w<w^-_\mathrm{A}(0,0)=f(u^-_\mathrm{A}(0,0))$, each of these heteroclinic orbits connects the weak unstable manifold of the repelling node from the middle branch $\mathcal{M}^m_0$ to the stable manifold of the saddle fixed point on the left branch $\mathcal{M}^\ell_0$ in the layer problem~\eqref{e:layer}. The singular set $\mathcal{Z}_0$ then defines a normally repelling invariant manifold from which the center manifold $\mathcal{Z}_\eps$ can be constructed; see~\cite[\S3]{CSbanana} for details of this construction. 

Within the manifold $\mathcal{Z}_\eps$, the pulse completes a sequence of alternating canard segments from $\mathcal{M}^\ell_\eps$ to $\mathcal{M}^m_\eps$ and fast jumps $\phi_\ell(\cdot; w_j)$ from $\mathcal{M}^m_\eps$ to $\mathcal{M}^\ell_\eps$, for monotonically decreasing values of $w_j>0$, approaching nearer the equilibrium after each such jump; see Figure~\ref{f:tail_sequence}. In~\cite{CSbanana}, it was shown that $w_j<w^-_A(0,0)+\eps^{2/3}$ for all $j$. \blue{Furthermore, given any sufficiently small fixed neighborhood $\mathcal{V}$ of the equilibrium, chosen independently of $\eps$ sufficiently small, there exists $N\in \mathbb{N}$ also independent of $\eps$, such that after $N$ such jumps the tail of the pulse is confined to $\mathcal{V}$ as $\xi\to\infty$~\cite[Theorem 2.2(ii)]{CSbanana}. } 

\begin{figure}
\centering
\includegraphics[width=0.6\linewidth]{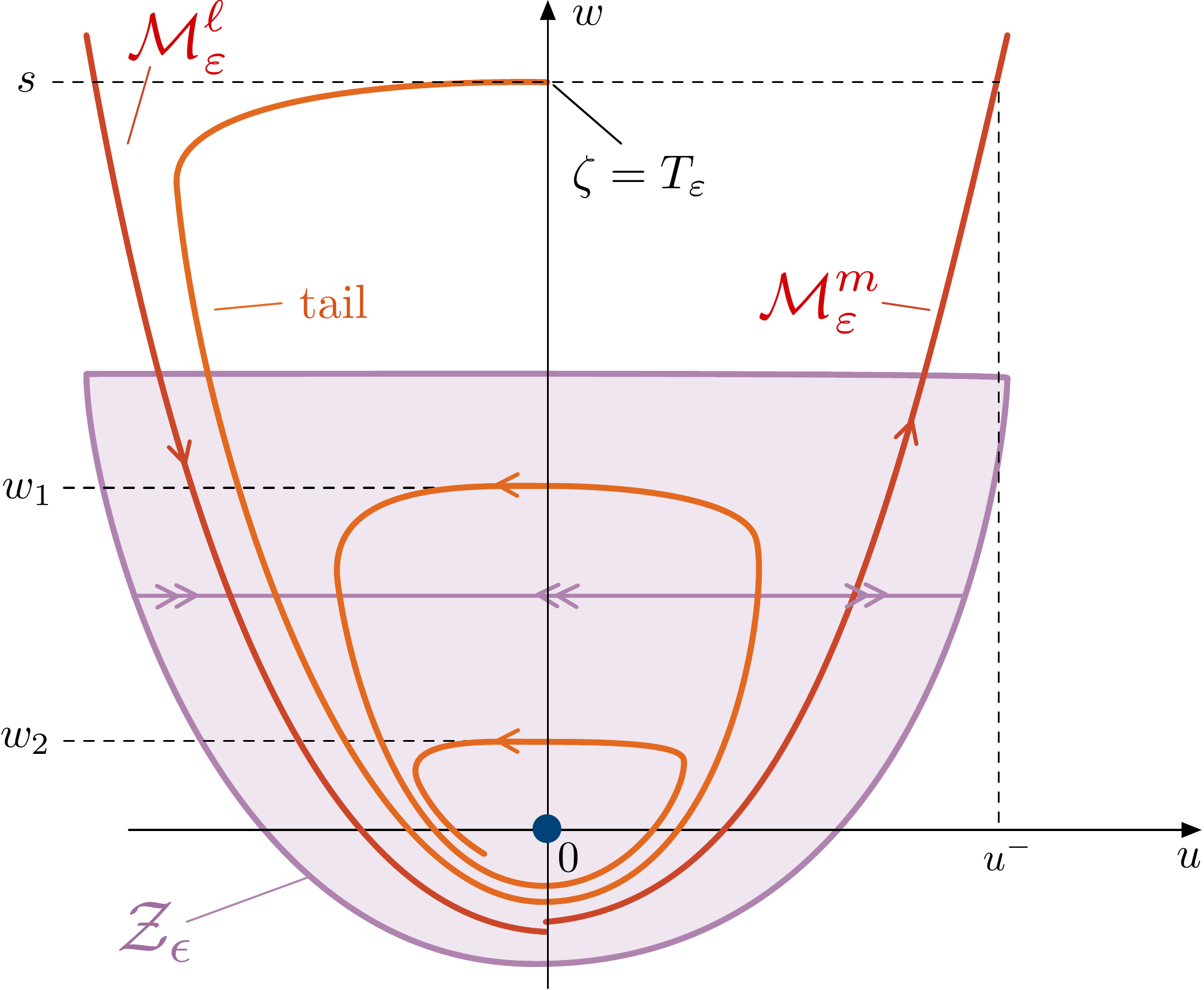}
\caption{Shown are the dynamics of the tail of the transitional pulse $\phi(\zeta;s,\eps)$ in the center manifold $\mathcal{Z}_\eps$.}
\label{f:tail_sequence}
\end{figure}

As argued in~\S\ref{s:fhn_entryspace}, the exponential dichotomies along the slow manifold $\mathcal{M}^\ell_\eps$ can be extended through the canard transition along the slow manifold $\mathcal{M}^m_\eps$, due to the consistent spatial eigenvalue splitting of $A_\eta$ for any $u<u^-_\mathrm{A}(\lambda,\eps)$, i.e. below the first Airy point: \blue{fixing $\lambda \in (0,5/24)$, we have that $u^-_\mathrm{A}(\lambda,0)>u^-_\mathrm{A}(0,0)$, and likewise $w^-_\mathrm{A}(\lambda,0)>w^-_A(0,0)+\mathcal{O}(\eps^{2/3})$ for all sufficiently small $\eps$, so that this spatial eigenvalue splitting is preserved along $\mathcal{M}^m_\eps$ above the `highest' fast jump which occurs in the tail sequence. Due to the fact that~\eqref{e:fhnlin1st_weighted} has slowly varying coefficients along the slow manifolds $\mathcal{M}^\ell_\eps$ and $\mathcal{M}^m_\eps$, these exponential dichotomies can be extended along the sequence of fast tail jumps, provided the weight function $\eta(\zeta)$ is adjusted appropriately along each segment. }

In fact, recalling the computation~\eqref{e:slowevals} of the spatial eigenvalues $\nu_i$ of the linearized system~\eqref{e:fhnlin1st}, fixing $\lambda\in (0,5/24)$ and fixing $\delta>0$ sufficiently small, we find that the matrix $A_\eta(\zeta;\lambda,\eps)$ of the weighted system~\eqref{e:fhnlin1st_weighted} admits two eigenvalues of negative real part, and one of positive real part for any $\zeta$ such that $\eta(\zeta) = \eta_m := \frac{1}{2\sqrt{2}}$ and 
\begin{align}
u(\zeta;s,\eps)<u^-_\mathrm{A}(\lambda,\eps)-\delta.
\end{align}
The corresponding spectral gap can be bounded away from zero uniformly for all $\eps>0$ sufficiently small. In particular, this means that it suffices to choose $\eta(\zeta)=\eta_m$ along the entire interval $\zeta\in[T_\eps,\infty)$ to guarantee this splitting  along the slow manifolds $\mathcal{M}^\ell_\eps$ and $\mathcal{M}^m_\eps$.

Furthermore, each of the fast jumps are perturbations from singular fronts $\phi_\ell(\zeta;w_j), j=1,2,\ldots$ in the layer problem~\eqref{e:layer} for $w=w_j$ for a decreasing sequence of values $w=w_1>w_2>\ldots$, each of which is contained in the center manifold $\mathcal{Z}_\eps$. Similarly to the fronts $\phi_{f,b}$, setting $\eps=0$ in the eigenvalue problem~\eqref{e:fhnlin1st_weighted}, yields the limiting eigenvalue problems
\begin{equation}\label{e:fhnlin1st_wj}
\begin{aligned}
\dot{U} = A_{\eta,j}(\zeta;\lambda, 0)U, \qquad 
A_{\eta,j}(\zeta;\lambda, \eps):=\begin{pmatrix}-\eta_m &1& 0 \\ \lambda-f'(u_\ell(\zeta;w_j)) & c-\eta_m & 1  \\ 0 & 0 & -\frac{\lambda}{c}-\eta_m\end{pmatrix}, \qquad j=1,2,\ldots
\end{aligned}
\end{equation}
along the fast jumps $\phi_\ell(\zeta;w_j) = (u_\ell,u'_\ell)(\zeta; w_j)$ for values of $w=w_j$, and this system can be solved using the lower-dimensional reduced systems
\begin{equation}
\begin{aligned}\label{e:fhnlin1st_wj_reduced}
\dot{Y} = C_{\eta, j}(\zeta;\lambda, 0)Y, \qquad
C_{\eta,j}(\zeta;\lambda, 0):=\begin{pmatrix}-\eta_m &1 \\ \lambda-f'(u_\ell(\zeta;w_j)) & c-\eta_m \end{pmatrix}, \qquad j=1,2,\ldots 
\end{aligned}
\end{equation}
where $Y=(u,v)$. These systems correspond to the linearization of the layer problem~\eqref{e:layer} about the fronts $\phi_\ell(\zeta;w_j)$, and are Sturm-Liouville-type eigenvalue problems.
Because each such front $u_\ell(\zeta;w_j)$ is monotonically decreasing, the weighted derivative $e^{ \eta_m \zeta}(u_\ell',u_\ell'')(\zeta; w_j)$ of the front corresponds to an eigenfunction with no zeros, and hence $\lambda=0$ is the largest eigenvalue. Hence in the right half plane, we can construct exponential dichotomies along each such front, which can be extended to the full system~\eqref{e:fhnlin1st_wj} using variation of constants formulae. By roughness, these persist for $\eps>0$.

By a similar argument as in Proposition~\ref{p:fhn_primary_ed}, the dichotomies along the slow manifolds $\mathcal{M}^\ell_\eps$ and $\mathcal{M}^m_\eps$ can be pasted to those across each of the fronts, which allows for the extension of the exponential dichotomies with projections and spaces $P^\mathrm{s,u}_\mathrm{tail}(\zeta;\lambda,\eps)$ and $E^\mathrm{s,u}_\mathrm{tail}(\zeta;\lambda,\eps)$ along the entirety of the tail of the pulse. \blue{We note that this pasting procedure is only performed finitely many times, since the tail completes only finitely many such jumps (independent of $\eps$) before entering and remaining in a small neighborhood of the equilibrium, within which~\eqref{e:fhnlin1st_weighted} admits a uniform exponential dichotomy.} For portions of the tail passing near each of the slow manifolds $\mathcal{M}^\ell_\eps$ and $\mathcal{M}^m_\eps$ the corresponding dichotomy subspaces $E^\mathrm{s,u}_\mathrm{tail}(\zeta;\lambda,\eps)$ are $\mathcal{O}(\eps)$-close to the spectral projections of the matrix $A_\eta$.

Tracking the subspace $Q^+$ backwards along the tail of the pulse, this ensures that $Q^+$ is aligned with $E^\mathrm{s}_\mathrm{tail}(\zeta;\lambda,\eps)$, and therefore also aligned with the stable eigenspace of the matrix $A_\eta$ along $\mathcal{M}^\ell_\eps$ when approaching the jump $\phi_\ell(\cdot;s)$. Along the fast jump $\phi_\ell(\cdot;s)$, the eigenvalue problem~\eqref{e:fhnlin1st_weighted} is a perturbation of the singular system
\begin{equation}\label{e:fhnlin1st_0_ell}
\begin{aligned}
\dot{U} = A_{\eta,\ell}(\zeta;\lambda)U, \qquad
A_{\eta,\ell}(\zeta;\lambda):=\begin{pmatrix}-\eta_m &1& 0 \\ \lambda-f'(u_\ell(\zeta;s)) & c-\eta_m & 1  \\ 0 & 0 & -\frac{\lambda}{c}-\eta_m \end{pmatrix}
\end{aligned}
\end{equation}
where we choose a translate of the front $\phi_\ell(\zeta;s)= (u_\ell(\zeta;s),u_\ell'(\zeta;s))$ to be such that $u_\ell(0;s)=0$. We similarly choose $\zeta=T_\eps$ to be the first $\zeta$-value such that the pulse intersects the section $\{u=0\}$ after jumping from the middle slow manifold $\mathcal{M}^m_\eps$. It is now straightforward to see that Hypothesis~\ref{h:bl} is satisfied as follows.

Standard corner-type estimates~\cite[Theorem 4.5]{CdRS} imply that $|u(\zeta+T_\eps;s,\eps)-u_\ell(\zeta;s)|=\mathcal{O}(\eps \log \eps)$ on the interval $[-L^+_\eps,0]$, where $L^+_\eps = \mathcal{O}(\log \eps)$, which in turn gives the estimate
\begin{align}
|A_{\eta}(\zeta+T_\eps; \lambda, \eps) -A_{\eta,\ell}(\zeta;\lambda)|=\mathcal{O}(\eps \log \eps)
\end{align}
on the same interval, where $A_{\eta,\ell}(\zeta;\lambda)$ plays the role of $A_\mathrm{fast}(\zeta, \lambda)$ in Hypothesis~\ref{h:bl}. Furthermore, the solution $u_\ell(\zeta;s)$ decays exponentially to its limit $u^-:=\lim_{\zeta \to -\infty}u_\ell(\zeta;s)$, so that 
\begin{align}
|A_{\eta,\ell}(\zeta;\lambda)-A^{-}_{\eta,\ell}(\lambda)|=\mathcal{O}(e^{\tilde{\nu}\zeta})
\end{align}
for some $\tilde{\nu}>0$, as required, where 
\begin{align}
A^{-}_{\eta,\ell}(\lambda)=\lim_{\zeta\to -\infty}A_{\eta,\ell}(\zeta;\lambda).
\end{align}

Provided $\lambda$ is such that $u^-:=\lim_{\zeta \to -\infty}u_\ell(\zeta;s)$ satisfies $u^-\in(u^-_\mathrm{A}(\lambda,0), u^+_\mathrm{A}(\lambda,0))$, the Morse indices $i^\pm_\ell$ of the limiting matrices  $A_{\eta, \ell}^\pm:= \lim_{\zeta \to \pm \infty} A_{\eta,\ell}(\zeta;\lambda, 0)$ differ; in particular $i^-_\ell=2$ while $i^+_\ell=1$. The system~\eqref{e:fhnlin1st_0_ell} therefore admits exponential dichotomies on $\mathbb{R}^\pm$ with dichotomy spaces $E^\mathrm{s,u}_\pm(\zeta,\lambda)$ where $\mathrm{dim} E^\mathrm{u}_-(0,\lambda)=2=\mathrm{dim} E^\mathrm{s}_+(0,\lambda)$, and these subspaces generically intersect transversely due to the fact that $E^\mathrm{u}_-(0,\lambda)$ is constructed from a reduced upper triangular block system as in~\eqref{e:fhnlin1st_wj_reduced} and thus has trivial component in the $w$-direction, while $E^\mathrm{s}_+(0,\lambda)$ has a nontrivial $w$-component. We deduce that the singular problem~\eqref{e:fhnlin1st_0_ell} admits a one-dimensional space of bounded solutions $\psi_\ell(\zeta;\lambda)$ where $\mathrm{sp}\{\psi_\ell(0;\lambda)\}=E^\mathrm{u}_-(0,\lambda)\cap E^\mathrm{s}_+(0,\lambda)$.

Futhermore, the subspace $E^\mathrm{s}_+(\zeta,\lambda)$ converges to the stable eigenspace of the limiting matrix \begin{align}
A^{+}_{\eta,\ell}(\lambda)=\lim_{\zeta\to +\infty}A_{\eta,\ell}(\zeta;\lambda).
\end{align}
which itself represents the linearization~\eqref{e:fhnlin1st_weighted} at $\eps=0$ at the point $u^+=\lim_{\zeta \to \infty}u_\ell(\zeta;s)$ on the critical manifold $\mathcal{M}^\ell_0$. Hence this subspace is $\mathcal{O}(\eps)$-close to the stable subspace of the exponential dichotomy along the slow manifold $\mathcal{M}^\ell_\eps$, which is itself aligned with $Q^+$. It follows that $Q^+$ intersects $E^\mathrm{u}_-(0,\lambda)$ in a one-dimensional subspace. \red{The space $E^\mathrm{u}_-(\zeta,\lambda)$ converges to the eigenspace of $A^-_{\eta,\ell}$ corresponding to the eigenvalues $\nu_{1,2}(\lambda,u_-;0)$, so that $E^\mathrm{u}_-(0,\lambda)$ plays the role of $E^\cc_\mathrm{fast}$ in Hypothesis~\ref{h:bl_vector}, while $E^\uu_\mathrm{fast}$ is empty. Thus Hypothesis~\ref{h:bl_vector} is satisfied by the boundary space $Q^+$.} The statements regarding the derivatives of this intersection with respect to $\lambda$ follow from the fact that all of the above subspaces depend smoothly on $\lambda$ \blue{in combination with standard geometric singular perturbation estimates}.
\end{proof}

\paragraph{Cases 2 \& 3:} For the other two cases, we have the following results, which are proved similarly to Proposition~\ref{p:fhn_exit1}, and we omit the details of the proofs. For transitional pulses with secondary excursions which continue along $\mathcal{M}^m_\eps$ beyond the second Airy point (see Figure~\ref{f:fhn_doubleairy}), we have the following.
\begin{proposition}\label{p:fhn_exit2}
Consider the eigenvalue problem~\eqref{e:fhnstab} associated with a transitional pulse $\phi(\zeta; s, \eps)$ of the FitzHugh--Nagumo PDE~\eqref{e:fhn} for $s\in(w^+_\mathrm{A}(\lambda^*,0), 8/27-w^+_\mathrm{A}(\lambda^*,0))$ where $\lambda^* \in (0, 5/24)$. Fix $\delta>0$ sufficiently small, and let $\zeta = T_\eps$ denote the $\zeta$ value at which the pulse $\phi(\zeta; s, \eps)$ first reaches the section $\{w=w^+_\mathrm{A}(\lambda^*,0)+\delta\}$ after completing the primary excursion. \red{Then there exists $\epsilon_0>0$ such that for all $\epsilon\in(0,\epsilon_0]$ the subspace $Q^+(T_\eps;\lambda^*,\eps)$ satisfies Hypothesis~\ref{h:osc_vector}.}
\end{proposition}

Lastly, for transitional pulses with secondary excursions which jump to the right slow manifold $\mathcal{M}^r_\eps$ along a fast jump $\phi_r$ before reaching the second Airy point (see Figure~\ref{f:fhn_right}), we have the following.
\begin{proposition}\label{p:fhn_exit3}
Consider the eigenvalue problem~\eqref{e:fhnstab} associated with a transitional pulse $\phi(\cdot; s, \eps)$ of the FitzHugh--Nagumo PDE~\eqref{e:fhn} for $s\in(8/27-w^+_\mathrm{A}(\lambda^*,0),8/27-w^-_\mathrm{A}(\lambda^*,0))$ where $\lambda^* \in (0, 5/24)$. For all sufficiently small $\eps>0$, let $\zeta = T_\eps$ denote the $\zeta$ value at which the pulse $\phi(\zeta; s, \eps)$ first reaches the section $\{u=2/3\}$ after jumping from the absolutely unstable middle slow manifold $\mathcal{M}^m_\eps$. \red{Then there exists $\epsilon_0>0$ such that for all $\epsilon\in(0,\epsilon_0]$,~\eqref{e:fhnstab} satisfies Hypothesis~\ref{h:bl} at $\zeta=T_\eps$, and the subspace $Q^+(T_\eps;\lambda^*,\eps)$ satisfies Hypothesis~\ref{h:bl_vector}.}
\end{proposition}

\subsection{Proof of Theorem~\ref{thm:fhnaccumulation}}\label{sec:fhnaccumulationproof}

Based on the discussion above, and the analysis from~\S\ref{s:gentheory}, we briefly conclude the proof of our main result regarding eigenvalue accumulation for the transitional pulses along the homoclinic banana in the FitzHugh--Nagumo system.
\begin{proof}[Proof of Theorem~\ref{thm:fhnaccumulation}]
The restriction $\lambda^* \in \Sigma^\mathrm{slow}_\mathrm{abs}(s)$ implies that $w^-_\mathrm{A}(\lambda^*,0)<s<8/27-w^-_\mathrm{A}(\lambda^*,0)$, and therefore this scenario falls into one of the three cases outlined in~\S\ref{sec:exittracking}. The result follows from Theorem~\ref{t:layers} applied to~\eqref{e:fhnlin1st_weighted} on the slow interval $[0,\eps T_\eps]$, in combination with Proposition~\ref{p:fhn_slowabs} for verification of Hypothesis~\ref{h:osc} as well as Hypothesis~\ref{h:Airy} near the Airy points $w=w^\pm_\mathrm{A}(\lambda,\eps)$ and Proposition~\ref{p:fhn_entry} for verification of the entry boundary condition in Hypothesis~\ref{h:osc_vector} for the entry subspace $Q^-(0;\lambda,\eps)$. For the exit subspace $Q^+(T_\eps;\lambda,\eps)$, Propositions~\ref{p:fhn_exit1}--\ref{p:fhn_exit3} provide verification for the relevant subset of Hypotheses~\ref{h:osc_vector}--\ref{h:bl_vector} which apply in each case \blue{(in the case of Proposition~\ref{p:fhn_exit2}, we treat $L_\eps=0$, so that $T_\eps=T/\eps$, as the fast layer does not play a role in this case)}. 

We note that there are two critical cases which are excluded from Propositions~\ref{p:fhn_exit1}--\ref{p:fhn_exit3}, that is when $s=w^+_\mathrm{A}(\lambda^*,0)$ and $s=8/27-w^+_\mathrm{A}(\lambda^*,0)$ when the transitional pulse departs $\mathcal{M}^m_\eps$ along a fast jump precisely at the upper Airy point. However by adjusting $\lambda$ by an amount $\delta<\delta_\lambda/2$, so that  $\lambda=\lambda^*+\delta$ and applying the argument to the smaller interval $(\lambda^*, \lambda^*+2\delta)$, we obtain $\mathcal{O}(\eps^{-1})$ eigenvalues on this interval, and hence on the larger interval as well.
\end{proof}

\section{Discussion} \label{s:disc}
In this paper, we examined the stability of traveling pulses in the FitzHugh--Nagumo equation~\eqref{e:1} along a parametric single-to-double pulse transition which occurs along the so-called homoclinic banana (Figure~\ref{f:ccurvebanana}). The pulses are individual traveling wave profiles, and the transition from the single pulse to double pulse occurs in an $\mathcal{O}(e^{-q/\eps})$-neighborhood in parameter space. The initially stable one-pulse solutions lose stability as a sequence of eigenvalues crosses into the right half plane along the real axis. As $\eps\to0$, these eigenvalues accumulate on an interval of the real axis which increases as the transition curve is traversed, with the 'most unstable' pulses admitting $\mathcal{O}(1/\eps)$ eigenvalues on the interval $\lambda\in (0,5/24)$. In particular, the most unstable eigenvalues are of $\mathcal{O}(1)$ size with respect to the singular perturbation parameter $\eps$.

We showed that this accumulation of eigenvalues is due to the presence of unstable `slow absolute spectrum' along the middle slow manifold $\mathcal{M}^m_\eps$ which is traversed by the transitional pulses, with the pulses in the middle of the transition traversing the longest portion of this manifold, and thus accumulating the most eigenvalues. This slow absolute spectrum is generated due to the pulse spending long `times' (in the traveling wave coordinate) near this slow manifold, along which the linearization admits a different relative spatial Morse index than that of the asymptotic rest state $(u,w)=(0,0)$. \blue{We only consider when this occurs for real values of the eigenvalue parameter. On the one hand, we are not aware of a natural example with complex absolute spectrum in this context. On the other hand, treating the complex case appears to be technically challenging. A major obstacle is a lack of rigidity in the sense that complex absolute spectra at any two points along the slow manifold might be disjoint.}

While our results are motivated by, and applied to, the FitzHugh--Nagumo equation, the results in~\S\ref{s:gentheory} regarding eigenvalue accumulation are formulated in much greater generality. When considering an eigenvalue problem with a small parameter controlling the rate of passage along such a slow manifold, the accumulation is reduced to understanding a linear boundary value problem along this manifold, and we outlined several possible cases which arise generically. Our analysis focused on two cases, namely (1) Airy transitions where the change in spatial Morse index occurs along the slow manifold itself and (2) fast entry/exit layers where the traveling wave jumps directly onto such a slow manifold. 

We remark that the instability mechanism described here is generic and likely occurs in many systems with timescale separation, such as the Klausmeier equation~\cite{BCDKlausmeier} and the Oregonator model~\cite{Rademacher_thesis}. In particular, traveling waves that traverse canard segments which involve long portions near a repelling slow manifold are likely to develop such instabilities as the fast eigenvalue structure evolves along the slow manifold. While it requires some effort to verify the respective boundary conditions stated in  Hypotheses~\ref{h:osc_vector}--\ref{h:bl_vector} in a given specific system, the results nevertheless provide a geometric explanation for the accumulation phenomenon; furthermore, we performed this analysis in~\S\ref{sec:boundaryverification} in the case of the FitzHugh--Nagumo system as a proof of concept.

\begin{figure}
\centering
\includegraphics[width=0.7\linewidth]{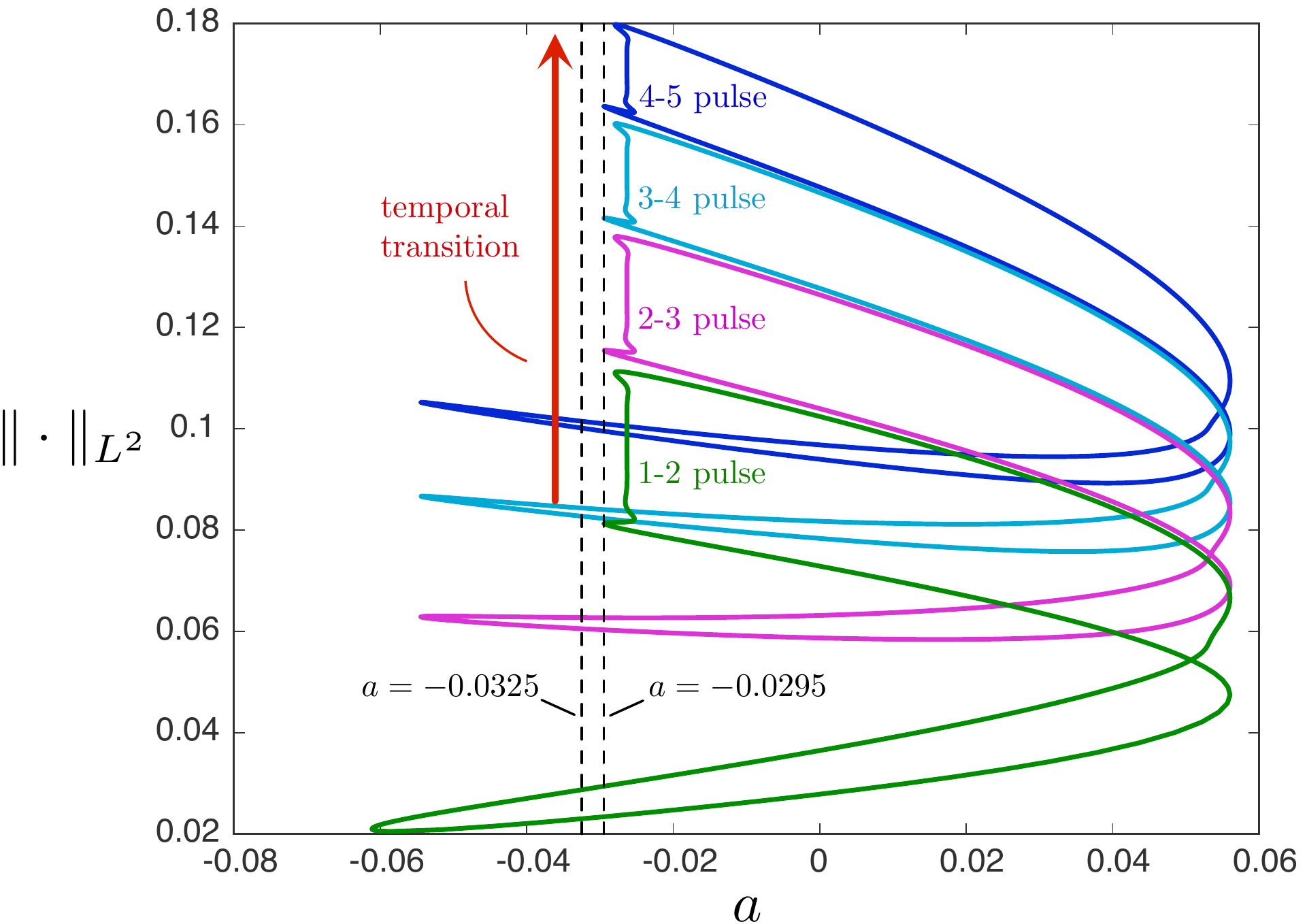}
\caption{Shown are the `stacked homoclinic bananas' obtained by numerical continuation for fixed $\eps=0.015, \gamma=2$. Each closed curve or banana is the result of numerical continuation in the parameters $(a,c)$, and the upper left corner of each banana contains the parametric transition from a spectrally stable $n$-pulse to an unstable $(n+1)$-pulse. We depict the bananas for the $1$-to-$2$-pulse transition (green), as well as those for $2$-to-$3$-pulse (magenta), $3$-to-$4$-pulse (cyan), and $4$-to-$5$-pulse (blue).}
\label{f:stacked_bananas}
\end{figure}

These results on linear (in)stability of the transitional pulses raise questions on the resulting temporal dynamics in the PDE. Though the transitional pulses (except for the initial single pulse) are all unstable, direct numerical simulations nevertheless suggest that solutions track the parametric transition, with additional pulses appearing out of the tail of the primary pulse as time increases. As the pulses all exist at numerically indistinguishable parameter values, this suggests the presence of an invariant manifold containing all of the intermediate traveling waves, and the temporal transition amounts to drift along this manifold. This idea is further evidenced by the existence of many such `stacked' homoclinic bananas, shown in Figure~\ref{f:stacked_bananas}, connecting $n$-pulses to $(n+1)$-pulses along analogous parametric transitions to the single-to-double pulse banana; this sequence of parametric transitions could guide the continual growth of additional pulses in the temporal transition in the PDE, as depicted in Figure~\ref{f:temporal_replication}.

\begin{figure}
\hspace{.01 \textwidth}
\begin{subfigure}{.45 \textwidth}
\centering
\includegraphics[width=1\linewidth]{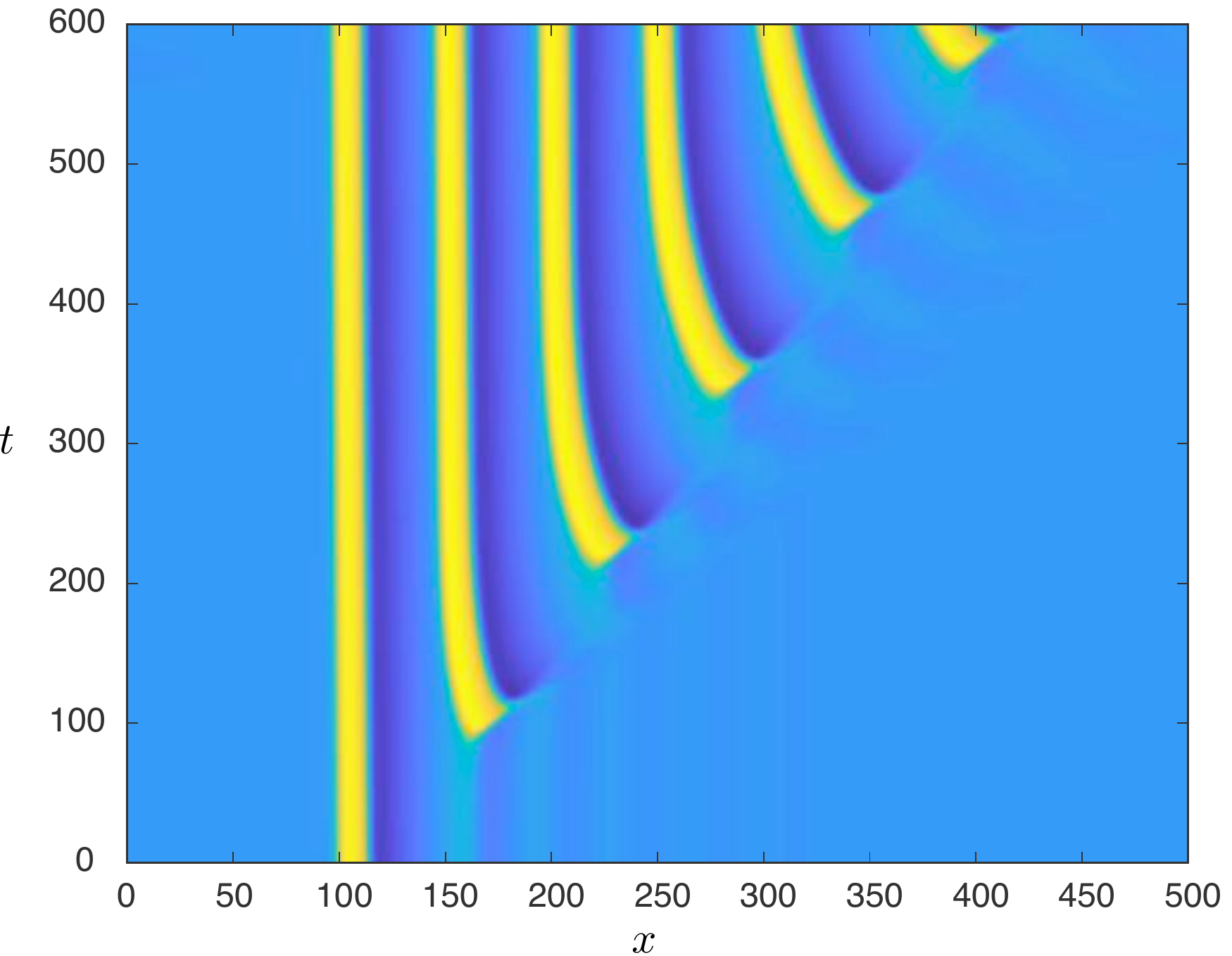}
\end{subfigure}
\hspace{.01 \textwidth}
\begin{subfigure}{.45 \textwidth}
\centering
\includegraphics[width=1\linewidth]{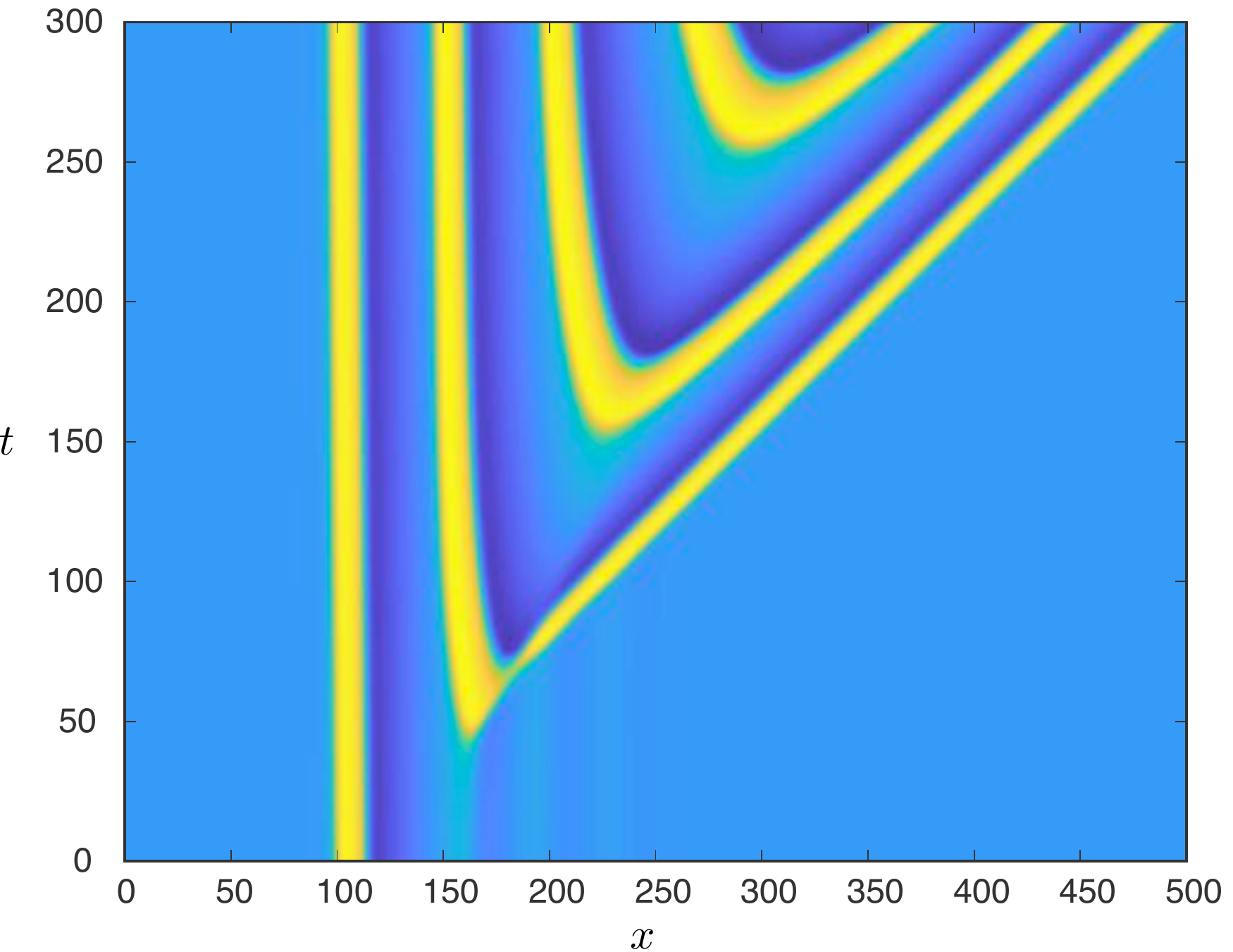}
\end{subfigure}
\hspace{.01 \textwidth}
\begin{subfigure}{.045 \textwidth}
\centering
\includegraphics[width=1\linewidth]{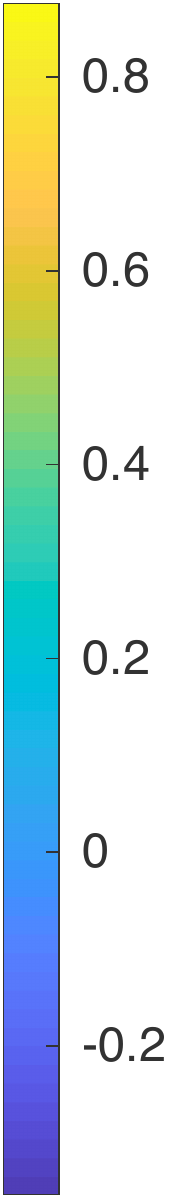}
\vspace{0.5pt}
\end{subfigure}
\caption{Shown are results of direct numerical simulations in the FitzHugh--Nagumo PDE~\eqref{e:1}. The system was initialized with an initial $u$-profile given by the oscillatory $1$-pulse depicted in the first panel of Figure~\ref{f:temporal_replication}, and evolved forward in time in a co-moving frame in which the primary pulse is approximately stationary, with periodic boundary conditions and parameter values $\eps=0.015, \gamma=2$ and $a=-0.0295$ (left panel), $a=-0.0325$ (right panel); \blue{see Figure~\ref{f:stacked_bananas} for the location of these parameter values in relation to the stacked homoclinic bananas.} \emph{Left panel:} space-time plot of $u(x,t)$ corresponding to the pulse-adding dynamics shown in Figure~\ref{f:temporal_replication}. \emph{Right panel:} space-time plot of $u(x,t)$ in a parameter regime in which~\eqref{e:1} exhibits pulse-splitting. As each pulse is added to the tail of the primary pulse (similar to the pulse-adding scenario), an additional pulse ``splits off", traveling to the right.  }
\label{f:splitting_vs_adding}
\end{figure}

Our results show that the instability along the parametric transition is rather severe in the sense that the number of eigenvalues increases without bound as $\eps\to0$ and that the exponential growth rates extend an $\mathcal{O}(1)$-distance into the right half plane. Thus, it is not clear why solutions would follow the parametric transition curve $\Gamma_\eps$. One possible explanation relies on the fact that the initial single pulses are spectrally stable: solutions that start sufficiently close to $\Gamma_\eps$ converge therefore exponentially fast to this curve, and the subsequently emerging unstable eigenvalues may not be strong enough to repel the solutions away from $\Gamma_\eps$ before these eigenvalues stabilize again and the solution resembles a double pulse. If this explanation is correct, then we would expect that solutions for values of $a$ that are further away to the left from the parameters corresponding to $\Gamma_\eps$ to be stronger affected by the instability along $\Gamma_\eps$: as demonstrated in Figure~\ref{f:splitting_vs_adding}, this is indeed the case as solutions for smaller values of $a$ experience pulse-splitting instead of pulse-adding in the tail of the primary pulse. The explanation proposed here draws on the analogy with the dynamics near a slow passage through a bifurcation point which is known to lead to an $\mathcal{O}(1)$-delay of the onset of instability. In our case, a large number of eigenvalues becomes unstable, and the construction of such an invariant manifold near the singular limit is therefore likely a very challenging problem.

\paragraph{Acknowledgments.}
\blue{We thank the reviewers for their insightful comments that helped improve the manuscript.} Carter gratefully acknowledges support through NSF Grant DMS-2016216 (formerly DMS-1815315). Rademacher gratefully acknowledges support by the German Research Fund (DFG) through grant RA 2788/1-1. Sandstede gratefully acknowledges support by the National Science Foundation through grant DMS-1714429.

\bibliographystyle{abbrv}
\bibliography{my_bib}

\end{document}